\documentclass[a4paper,UKenglish,%
numberwithinsect, %
]{lipics-v2021}
\hideLIPIcs

\pdfoutput=1

\nolinenumbers
\title{Model Checking as Program Verification by Abstract Interpretation}
\subtitle{(Extended version)}

\author{Paolo Baldan}
{Department of Mathematics, University of Padua, Italy}
{baldan@math.unipd.it}
{https://orcid.org/0000-0001-9357-5599}{}

 \author{Roberto Bruni}
 {Department of Computer Science, University of Pisa, Italy}
 {bruni@di.unipi.it}
 {https://orcid.org/0000-0002-7771-4154}{}

 \author{Francesco Ranzato}
 {Department of Mathematics, University of Padua, Italy}
 {ranzato@math.unipd.it}
 {https://orcid.org/0000-0003-0159-0068}{}

 \author{Diletta Rigo}
 {Department of Mathematics, University of Padua, Italy}
 {diletta.rigo@math.unipd.it}
 {https://orcid.org/0009-0008-4473-3202}{}

\authorrunning{P. Baldan, R. Bruni, F. Ranzato and D. Rigo} %

\Copyright{P. Baldan, R. Bruni, F. Ranzato and D. Rigo} %



\category{} %

\relatedversion{} %

\usepackage{todonotes}

\acknowledgements{The authors were partially funded by
the \emph{Italian MUR}, under the PRIN 2022 PNRR project no.~P2022HXNSC on
``Resource Awareness in Programming: Algebra, Rewriting, and Analysis''.
Francesco Ranzato was partially funded by 
a \emph{WhatsApp Research Award} and  
an \emph{Amazon Research Award} for \emph{AWS Automated Reasoning}.}

\EventEditors{John Q. Open and Joan R. Access}
\EventNoEds{2}
\EventLongTitle{42nd Conference on Very Important Topics (CVIT 2016)}
\EventShortTitle{CVIT 2016}
\EventAcronym{CVIT}
\EventYear{2016}
\EventDate{December 24--27, 2016}
\EventLocation{Little Whinging, United Kingdom}
\EventLogo{}
\SeriesVolume{42}
\ArticleNo{23}
\renewcommand{\frame}[1]{ \langle #1 \rangle}

\usepackage{cite}

\newcommand{\concat}{ \!\mathop{::}\! }

\newcommand{\set}[1]{\left \{ #1 \right \}}

\newcommand{\N}{\mathbb{N}}

\newcommand{\ACTL}{{\rm \text{ACTL }}}

\newcommand{\class}[2]{[#1]_{#2}}

\newcommand{\sem}[2][\eta]{\llbracket {#2} \rrbracket_{#1}}
\newcommand{\bsem}[1]{\left ( \!| #1|\! \right )}

\newcommand{\cstrut}{\rule[0.3\baselineskip]{0pt}{2pt}}
\newcommand{\PtoR}[1]{\PtoRmap{\overline{#1\cstrut}\,}}
\newcommand{\PtoRmap}[1]{\lfloor #1 \rceil}
\newcommand{\PtoE}[1]{\PtoRmap{#1\cstrut\,}}

\newcommand{\PtoRn}[1]{\PtoRmapn{\overline{#1\cstrut}\,}} %
\newcommand{\PtoRmapn}[1]{\llfloor #1 \rrceil}

\newcommand{\AX}{\mathsf{AX}~}
\newcommand{\AF}{\mathsf{AF}~}
\newcommand{\AG}{\mathsf{AG}~}
\newcommand{\AU}{~\mathsf{AU}~}
\newcommand{\trel}{\mathop{\rightarrowtriangle}}

\newcommand{\MOKA}{\mathsf{MOKA}}
\newcommand{\kaf}{\mathsf{KAF}}
\newcommand{\kat}{\mathsf{KAT}}
\newcommand{\e}{\mathsf{e}}
\newcommand{\one}{\mathsf{1}}
\newcommand{\zero}{\mathsf{0}}
\newcommand{\X}{\mathsf{X}}
\newcommand{\Y}{\mathsf{Y}}
\newcommand{\Z}{\mathsf{Z}}

\renewcommand{\r}{\mathsf{r}}
\newcommand{\s}{\mathsf{s}}
\renewcommand{\a}{a}
\renewcommand{\b}{b}
\newcommand{\p}{\mathsf{p?}}
\newcommand{\ps}{\mathsf{p}}

\newcommand{\np}{\mathsf{\neg p?}}

\newcommand{\loops}{\mathsf{loop?}}
\newcommand{\nloops}{\mathsf{\neg loop?}}

\newcommand{\nexte}{\mathsf{next}}

\newcommand{\add}{\mathsf{add}}
\newcommand{\push}{\mathsf{push}}
\newcommand{\pop}{\mathsf{pop}}
\newcommand{\reset}{\mathsf{reset}}

\newcommand{\Igs}{\mathbb{Z}}
\newcommand{\ev}{\Igs_k}

\newcommand{\ssim}{{\scalebox{0.5}{$\sim$}}}
\newcommand{\ssimplus}{{\scalebox{0.5}{$\sim_+$}}}
\newcommand{\aapprox}{{\scalebox{0.5}{$\approx$}}}

\newcommand{\As}{A^s_{\ssim}}
\newcommand{\Aa}{A^s_{\aapprox}}

\newcommand{\subst}[3]{\ensuremath{{#1}[{#3}/{#2}]}}

\usepackage{stackengine}

\def\defemb#1#2{\expandafter\def\csname #1\endcsname
                              {\relax\ifmmode #2\else\hbox{$#2$}\fi}}
 
\def\2c-math#1#2{{\par\medskip\noindent ${#1}$
                      \par\smallskip
                        \noindent\hspace*{\fill} ${#2}$}
                           \\[10pt]}
\defemb{aF}{\mathcal{F}^\cA}
\defemb{sF}{\mathcal{F}^{\star}}
\defemb{cA}{\mathcal{A}}
\defemb{cB}{\mathcal{B}}
\defemb{cC}{\mathcal{C}}
\defemb{cD}{\mathcal{D}}
\defemb{cE}{\mathcal{E}}
\defemb{cF}{\mathcal{F}}
\defemb{cG}{\mathcal{G}}
\defemb{cH}{\mathcal{H}}
\defemb{cI}{\mathcal{I}}
\defemb{cJ}{\mathcal{J}}
\defemb{cK}{\mathcal{K}}
\defemb{cL}{\mathcal{L}}
\defemb{cM}{\mathcal{M}}
\defemb{cN}{\mathcal{N}}
\defemb{cO}{\mathcal{O}}
\defemb{cP}{\mathcal{P}}
\defemb{cQ}{\mathcal{Q}}
\defemb{cR}{\mathcal{R}}
\defemb{cS}{\mathcal{S}}
\defemb{cT}{\mathcal{T}}
\defemb{cU}{\mathcal{U}}
\defemb{cV}{\mathcal{V}}
\defemb{cW}{\mathcal{W}}
\defemb{cY}{\mathcal{Y}}
\defemb{cX}{\mathcal{X}}
\defemb{cZ}{\mathcal{Z}}
\defemb{cHys}{{\cH_{\cY}^{\cS}}}
\defemb{cHy}{{\cH_{\cY}}}
\defemb{ld}{\dashv}
\defemb{rd}{\vdash}
\defemb{defo}{\mathfrak{D}}

\def\part{P}

\defemb{can}{\vert \! \vert}

\defemb{ptau}{\tau}

\newcommand{\nat}{\mathbb{N}}

\newcommand{\cel}{c}
\newcommand{\ela}{a}

\DeclareMathOperator{\true}{\textbf{tt}}
\DeclareMathOperator{\false}{\textbf{ff}}

\newcommand{\Var}{\ensuremath{\mathit{Var}}}

\newcommand{\dirIm}[1]{\ensuremath{#1^{>}}}
\newcommand{\invIm}[1]{\ensuremath{#1^{<}}}

\newcommand{\cone}[1]{\ensuremath{\downarrow{#1}}}

\newcommand{\pow}[1]{\ensuremath{\mathcal{P}({#1})}}

\newcommand{\poweq}[2]{\ensuremath{\pow{{#1}}_{#2}}}

\newcommand{\RVar}{\ensuremath{\mathsf{Var}}}
\newcommand{\LVar}{\ensuremath{\operatorname{Var}}}
\newcommand{\TVar}{\ensuremath{\mathit{Var}}}

\newcommand{\supp}[1]{\ensuremath{\mathit{supp}({#1})}}

\newcommand{\str}[1]{\ensuremath{#1^+}}

\newcommand{\bvar}[1]{\ensuremath{\mathit{bv}({#1})}}

\usepackage{subcaption} %
\usepackage{amsmath}
\usepackage{epsfig,amsfonts,latexsym,upref}
\usepackage{stmaryrd}
\usepackage{centernot}

\definecolor{darkgreen}{rgb}{0.0, 0.5, 0.0} %

\usepackage{galois}
\usepackage[ruled,vlined,linesnumbered]{algorithm2e}
\usepackage[inference]{semantic}

\usepackage{tikz}
\usepackage{mathtools,tikz-cd}
\usetikzlibrary{fit}

\usetikzlibrary{datavisualization.formats.functions}
\usetikzlibrary{decorations.markings}
\usetikzlibrary{arrows,backgrounds}
\usetikzlibrary{positioning}

\usepackage[%
appendix=append,
]{apxproof}
\newtheoremrep{theorem}{Theorem}[section]
\newtheoremrep{lemma}[theorem]{Lemma}
\newtheoremrep{proposition}[theorem]{Proposition}
\newtheoremrep{corollary}[theorem]{Corollary}
\newtheoremrep{definition}[theorem]{Definition}
\newtheoremrep{example}[theorem]{Example}
\newtheoremrep{remark}[theorem]{Remark}
\usepackage{mathrsfs}
\usepackage[all]{nowidow}

\usepackage{xfrac}

\usepackage{mdframed}

\usepackage{alltt}
\usepackage{upgreek}

\allowdisplaybreaks

\begin{document}

\maketitle

\begin{abstract}
  Abstract interpretation offers a powerful toolset for static
  analysis, tackling precision, complexity and state-explosion issues.  
  In the literature, state partitioning abstractions
  based on (bi)simulation and property-preserving state relations have
  been successfully applied to abstract model checking.  Here, we
  pursue a different track in which model checking is seen as an
  instance of program verification.  To this purpose, we introduce a
  suitable language---called $\MOKA$ (for $\mathsf{MO}$del
checking as abstract interpretation of $\mathsf{K}$leene $\mathsf{A}$lgebras)---which is used to encode temporal formulae as programs.  In
  particular, we show that (universal fragments of) temporal logics, such as ACTL or, more generally,
  universal $\mu$-calculus %
  can be
  transformed into $\MOKA$ programs.  
  Such programs return
  all and only the initial states which violate the formula.
  By applying abstract interpretation to $\MOKA$
  programs, we pave the way for reusing more general
  abstractions than partitions as
  well as for tuning the precision of the abstraction to
  remove or avoid false alarms.
  We show how to perform
  model checking via a program logic
  that combines under-approximation and abstract interpretation
  analysis to avoid false
  alarms.  The notion of locally complete abstraction is used to
  dynamically improve the  analysis precision via
  counterexample-guided domain refinement.
\end{abstract}

\section{Introduction}
Abstraction is a fundamental craft for mastering complexity.
In model checking, abstraction-guided space reduction allows to
mitigate the well-known state explosion problem~\cite{DBLP:conf/laser/ClarkeKNZ11}.
Abstract interpretation~\cite{cousot21,DBLP:conf/popl/CousotC77} is
the de facto standard framework for designing sound analyses.
The idea of applying abstract interpretation to model checking has
been widely investigated, e.g.,
\cite{BallKY05,DBLP:reference/mc/DamsG18,DBLP:conf/tacas/ShohamG04,DBLP:reference/mc/JhalaPR18,DBLP:conf/cav/GuptaS05,DBLP:conf/cav/GrafS97,DBLP:journals/toplas/ClarkeGL94,DBLP:journals/tocl/BustanG03,DBLP:conf/tacas/BallPR01,DBLP:journals/logcom/RanzatoT07,DBLP:conf/sas/CousotGR07,DBLP:conf/cav/ClarkeGJLV00,DBLP:journals/jacm/ClarkeGJLV03,DBLP:conf/sas/SchmidtS98,DBLP:conf/sas/GiacobazziQ01,DBLP:journals/iandc/GiacobazziR06,DBLP:conf/esop/Ranzato01,DBLP:journals/toplas/DamsGG97,DBLP:journals/ase/CousotC99,DBLP:conf/sas/Masse02,DBLP:conf/sas/RanzatoT02,DBLP:conf/popl/CousotC00,DBLP:journals/iandc/GrumbergLLS07,DBLP:conf/concur/Baldan0P20,DBLP:journals/fmsd/LoiseauxGSBB95,schmidt,DBLP:conf/lics/LarsenT88,DBLP:conf/lpar/BandaG10,NielsonN10,ZhangNN12}
(\S~\ref{sec:rw} accounts for more closely related work).
These approaches often rely on over-approximations
that preserve some logical properties, like
state partitioning abstractions or simulation-preserving
relations.

\subparagraph{Contribution.} 
In this work, we pursue a different approach, which consists in
recasting model checking of temporal formulae directly in terms of
program verification, for which the whole tool-suite of abstract
interpretation is readily available.
To this aim, we exploit an instance of a Kleene algebra
with tests and fixpoints ($\kaf$)~\cite{DBLP:conf/csl/Leiss91}, whose primitives allow to map each temporal formula to a
program that can single out all and only the counterexamples to the formula.
This program can then be analysed through a sound abstract interpretation
that over-approximates the concrete
behaviour, so that all possible
counterexamples
are exposed. However, this
over-approximation might not faithfully model the program behaviour
in the abstract domain, thus possibly mixing true and false alarms.
This lack of precision cannot happen when the abstract
interpretation is
complete~\cite{DBLP:conf/popl/CousotC79}.
However, completeness happens quite seldom, and even if, in principle, it can be
achieved by refining the abstract domain~\cite{DBLP:journals/jacm/GiacobazziRS00}, the most abstract domain refinement
ensuring completeness is often way too concrete (it might well coincide with the concrete domain itself).

To remove false alarms, the validity of temporal formulae can be
analysed by deriving certain judgements for the corresponding
program in a suitable program logic.
This enables the use of techniques borrowed from locally complete abstract
interpretation~\cite{DBLP:journals/jacm/BruniGGR23,BruniGGR21}, where the
over-approximation provided by the abstract domain is paired with
an under-approximating specification, in the style of O'Hearn's
incorrectness logic~\cite{DBLP:journals/pacmpl/OHearn20}.
This way, any alarm raised by an incorrectness logic proof
corresponds to a true counterexample and local completeness
guarantees that derivable judgements either exposes some true
counterexamples (if any) or proves that there are none.
Moreover, the failure of a proof obligation during inference can point out
how to dynamically refine the underlying abstraction to enhance the
precision and expressiveness of the analysis, a technique called abstract interpretation 
repair~\cite{DBLP:conf/pldi/BruniGGR22}.

\subparagraph{Methodology (and a toy example).}
Our main insight is to design a meta-programming language, called
$\MOKA$ ($\mathsf{MO}$del
checking as abstract interpretation of $\mathsf{K}$leene $\mathsf{A}$lgebras), where suitable temporal formulae can be mapped to.
We show how this can be done for ACTL or, more generally, for the
single variable $\mu$-calculus without the existential
diamond modality.
$\MOKA$ is a language of regular commands, based on Kleene
algebra with tests ($\kat$)~\cite{DBLP:journals/toplas/Kozen97} augmented with
fixpoint operators ($\kaf$)~\cite{DBLP:conf/csl/Leiss91}. $\MOKA$ leverages a small
set of primitives to extend  and filter 
computation paths. They can be combined with the usual
$\kaf$ operators of sequential composition, join (i.e., nondeterministic
choice), Kleene iteration, and least fixpoint computation.
The corresponding programs are then analysed by abstract interpretation.

The language $\MOKA$ operates on computation paths $\frame{\sigma,
\Delta}$, where $\sigma \in \Sigma$ represents the current system
state, and $\Delta \in \pow{\Sigma}$ represents the set of traversed
states. In the following, we informally overload
the symbol $\sigma$ to denote $\frame{\sigma, \varnothing}$. Several
computation paths can be stacked and unstacked through the $\MOKA$ primitives
$\push$ and $\pop$ for dealing with nested temporal formulae.
A generic stack is denoted $\frame{\sigma,\Delta}\concat S$ and we
call $\sigma$ its \emph{current} state.  Each temporal formula
$\varphi$ is mapped to a $\MOKA$ program $\PtoR{\varphi}$ that
intuitively computes counterexamples to $\varphi$, whence the bar over
$\varphi$ in the application of the encoding $\PtoRmap{\cdot}$.
Letting $\sem[]{\cdot}$ denote the usual (collecting) denotational
semantics, a key result is, therefore, that a state $\sigma$ satisfies
the formula $\varphi$ iff the execution of $\PtoR{\varphi}$ on
$\sigma$ returns the empty set, written $\sigma\models \varphi$ iff
$\sem[]{\PtoR{\varphi}}\sigma =\varnothing$.  Of course, the results
can be generalised additively to sets of stacks, for which
$\PtoR{\varphi}$ filters out exactly those stacks whose current state
satisfies $\varphi$.
For example, for an atomic proposition $p$, we let
$\PtoR{p} \triangleq \mathsf{\neg}\p$, where $\mathsf{\neg}\p$ is
a $\MOKA$ primitive that filters out those states where $p$ is valid.
As a further example, the $\MOKA$ program for the formula $\AG
\varphi$ (stating that $\varphi$ along every possible path always
holds) is $\PtoR{\AG \varphi} =
\textcolor{lipicsLineGray}{\push;}\,
\nexte^*;\PtoR{\varphi}\textcolor{lipicsLineGray}{; \pop}$.
The intuition is quite simple: $\nexte^*$ generates all the successors by
iterating the $\nexte$ operation, and they are filtered by
$\PtoR{\varphi}$, to expose those which fail to satisfy $\varphi$.
The operations $\push$ and $\pop$, in faded font, merely manage the stack structure which plays a role only for nested formulae, hence here they can be safely ignored.

To get a flavour of the proposed approach, we sketch an easy
example, which is a variation of~\cite[Examples 3.4,
3.7]{DBLP:journals/jacm/ClarkeGJLV03}.
Consider the transition system in Figure~\ref{fig:semaphore-car},
which models the behaviour of cars at a US traffic light. State names
consist of two letters, denoting, respectively, the traffic light status,
i.e.\ \textsf{r}ed, \textsf{y}ellow, \textsf{g}reen, and the car
behaviour, i.e.\ \textsf{s}top, \textsf{d}rive.
We model check the validity in the initial state $\mathsf{rs}$  of
the safety property $\varphi = \AG(\mathsf{\neg rd})$, 
 stating that it
will never happen that the light is red and the car is driving. This
is obviously true for the system, since $\mathsf{rd}$ is an unreachable state.
We therefore consider the $\MOKA$ program $\PtoR{\varphi} =
\textcolor{lipicsLineGray}{\push;}\,
\nexte^*;\mathsf{rd}?\textcolor{lipicsLineGray}{; \pop}$, and compute
its semantics $\sem[]{\PtoR{\varphi}}\set{\mathsf{rs}}$, which turns
out to be the empty set, thus allowing us to conclude that
$\mathsf{rs} \models \varphi$. In fact, after a few concrete steps of
iteration, $\sem[]{\mathsf{rd}?}\sem[]{\nexte^*}\set{\mathsf{rs}} =
\sem[]{\mathsf{rd}?}\set{\mathsf{rs},\mathsf{gs},\mathsf{gd},\mathsf{yd},\mathsf{ys}}
= \varnothing$.

\begin{figure}[t]
  \centering
  {\footnotesize
    \begin{subfigure}[b]{0.37\textwidth}
      \centering
      \begin{tikzpicture}[xscale=1.2,yscale=1.2]
        \node (rs) at (0,1) {$\mathsf{rs}$};
        \node (gs) at (1,0.6) {$\mathsf{gs}$};
        \node (ys) at (2,1) {$\mathsf{ys}$};
        \node (rd) at (0,0) {$\mathsf{rd}$};
        \node (gd) at (1,0) {$\mathsf{gd}$};
        \node (yd) at (2,0) {$\mathsf{yd}$};

        \draw[->,brown,line width=0.7pt] (-0.5,1) -- (rs);
        \draw[->] (rs) to (gs);
        \draw[->] (gs) to (yd);
        \draw[->] (yd) to (ys);
        \draw[->] (rd) -- (rs);
        \draw[->] (rd) -- (gs);
        \draw[->] (gs) -- (gd);
        \draw[->] (ys) to[out=160,in=20] (rs);
        \draw[->] (gd) -- (yd);
        \draw[->] (gd) to[out=240, in=300, loop] ();
        \draw[->] (ys) to[out=30, in=90, loop] ();
        \draw[->] (rs) to[out=90, in=150, loop] ();
      \end{tikzpicture}
      \caption{Behaviour of cars at US traffic light.}
      \label{fig:semaphore-car}
  \end{subfigure}}
  \quad
  {\footnotesize
    \begin{subfigure}[b]{0.32\textwidth}
      \centering
      \begin{tikzpicture}[xscale=1.2,yscale=1.2]
        \node (rs) at (0,1) {$\mathsf{rs}$};
        \node (gs) at (1,0.6) {$\mathsf{gs}$};
        \node (ys) at (2,1) {$\mathsf{ys}$};
        \node (rd) at (0,0) {$\mathsf{rd}$};
        \node (gd) at (1,0) {$\mathsf{gd}$};
        \node (yd) at (2,0) {$\mathsf{yd}$};

        \draw[->,brown,line width=0.7pt] (-0.5,1) -- (rs);
        \draw[->] (rs) to (gs);
        \draw[->] (gs) to (yd);
        \draw[->] (yd) to (ys);
        \draw[->] (rd) -- (rs);
        \draw[->] (rd) -- (gs);
        \draw[->] (gs) -- (gd);
        \draw[->] (ys) to[out=160,in=20] (rs);
        \draw[->] (gd) -- (yd);
        \draw[->] (gd) to[out=240, in=300, loop] ();
        \draw[->] (ys) to[out=30, in=90, loop] ();
        \draw[->] (rs) to[out=90, in=150, loop] ();

        \node[draw=red, thick, densely dotted, rounded corners,
          fit=(rs) (ys), inner
        sep=6pt, xscale=0.9, yscale=0.5] (boxA) {};
        \node[below left=0pt and 0pt of boxA]  {\textcolor{red}{$a$}};

        \node[draw=blue, thick, densely dotted, rounded corners,
          fit=(gs) (ys) (gd) (yd),
        inner sep=5pt, xscale=0.95, yscale=1.05] (boxC) {};
        \node[above=0pt of boxC]  {\textcolor{blue}{$c$}};

        \node[draw=darkgreen, thick, densely dotted, rounded corners,
          fit=(rd) (gd) (yd), inner
        sep=6pt, xscale=0.9, yscale=0.5] (boxB) {};
        \node[below left=0pt and 0pt of boxB]  {\textcolor{darkgreen}{$b$}};

      \end{tikzpicture}
      \caption{Abstract system.}
      \label{fig:composed-abstract}
  \end{subfigure}}
  \quad
  {
    \begin{subfigure}[b]{0.2\textwidth}
      \centering
      \begin{tikzpicture}[xscale=0.7,yscale=0.6, every
        node/.style={font=\footnotesize}]
        \node (bot) at (2,0) {$\bot$};
        \node (ys) at (1,1) {$a\wedge c$};
        \node (ydgd) at (3,1) {$b\wedge c$};
        \node (rsys) at (0,2) {$\textcolor{red}{a}$};
        \node (ysgsydgd) at (2,2) {$\textcolor{blue}{c}$};
        \node (rdydgd) at (4,2) {$\textcolor{darkgreen}{b}$};
        \node (rsysgsydgd) at (1,3) {$a\vee c$};
        \node (rdydgdysgs) at (3,3) {$b\vee c$};
        \node (top) at (2,4) {$\top$};

        \draw[-] (bot) to (ys);
        \draw[-] (bot) to (ydgd);
        \draw[-] (ys) to (rsys);
        \draw[-] (ys) to (ysgsydgd);
        \draw[-] (ydgd) to (ysgsydgd);
        \draw[-] (ydgd) to (rdydgd);
        \draw[-] (rsys) to (rsysgsydgd);
        \draw[-] (ysgsydgd) to (rsysgsydgd);
        \draw[-] (ysgsydgd) to (rdydgdysgs);
        \draw[-] (rdydgd) to (rdydgdysgs);
        \draw[-] (rsysgsydgd) to (top);
        \draw[-] (rdydgdysgs) to (top);
      \end{tikzpicture}
      \caption{Abstract domain.}
      \label{fig:semaphore-hasse}
  \end{subfigure}}
  \caption{Variations on the traffic light example
  from~\cite{DBLP:journals/jacm/ClarkeGJLV03}.}
  \label{fig:abs-semaphore}
\end{figure}
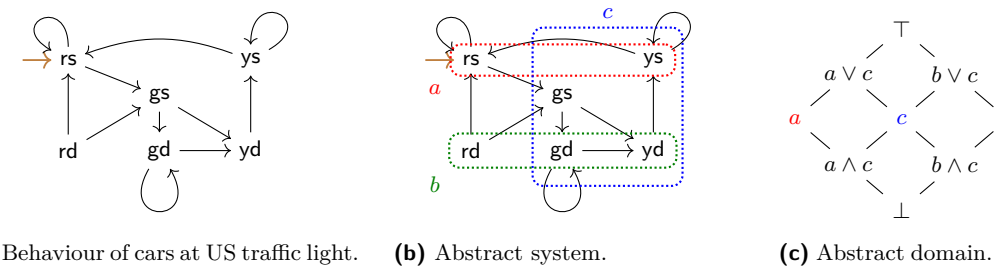

\subparagraph{How the abstraction works.}
We show how any sound abstraction of state properties in $\wp(\Sigma)$ of the
original system can be lifted to a range of sound
abstractions on stacks with varying degrees of precision, in a way
that $\MOKA$ programs can be analysed by an abstract interpreter,
denoted by $\sem[]{\cdot}^{\sharp}$.
The analysis output is always sound, meaning that the set of counterexamples
is over-approximated.
In particular, letting $\alpha$ denote the abstraction map to an
abstract lattice,
and $\bot$ the bottom element of the corresponding abstract computational domain,
if $\sem[]{\PtoR{\varphi}}^{\sharp} \alpha(\set{\sigma})=\bot$,
then $\sigma \models \varphi$ holds.
Vice versa, if $\sem[]{\PtoR{\varphi}}^{\sharp}\alpha(\set{\sigma})\neq\bot$,
then the abstract analysis might raise a false alarm because, in general,
abstractions are not complete.

Back to the previous example, let us consider the
(non-partitioning) abstract domain
$A\triangleq\set{
  \bot_A,
  a\wedge c,
  b\wedge c,
  \textcolor{red}{a},
  \textcolor{darkgreen}{b},
  \textcolor{blue}{c},
  a\vee c,
  b\vee c,
\top_A}$
in Figure~\ref{fig:semaphore-hasse}, induced by the abstract properties
$\textcolor{red}{a}$,
$\textcolor{darkgreen}{b}$ and
$\textcolor{blue}{c}$, whose concretizations are represented as
dotted boxes in Figure~\ref{fig:composed-abstract}.
The
abstract interpreter computes
$\sem[]{\PtoR{\AG(\mathsf{\neg rd})}}^{\sharp}\alpha(\set{\mathsf{rs}})=
\sem[]{\textcolor{lipicsLineGray}{\push;}\,
\nexte^*;\mathsf{rd}?\textcolor{lipicsLineGray}{; \pop}}^{\sharp} \textcolor{red}{a}= \bot$, thus
allowing us to conclude that $\mathsf{rs} \models \AG(\mathsf{\neg rd})$.
Since $\mathsf{rs}$ has a transition to $\mathsf{gs}$, we have
$\sem[]{\mathsf{rd}?}^{\sharp}\sem[]{\nexte^*}^{\sharp}\textcolor{red}{a}
=
\sem[]{\mathsf{rd}?}^{\sharp}(a\vee c) = \bot$.
It is worth remarking that in the abstract domain $A$, computing
$\sem[]{\nexte^*}^{\sharp}\textcolor{red}{a}$ requires
two iterations instead of
four, as it happens in the concrete computation. Consider now the formula $\psi = \AG (\mathsf{g} \to \AX
\mathsf{d})$, namely, ``whenever the
semaphore is green, then at the next state the car is driving''.
Spelling out the $\MOKA$ encoding, we have:
$\PtoR{\psi} = \textcolor{lipicsLineGray}{\push;}\, \nexte^*;
\mathsf{g?}; \textcolor{lipicsLineGray}{\push;}\, \nexte;
\mathsf{\neg d?}; \textcolor{lipicsLineGray}{\pop; \pop}$.
Again, $\psi$ holds for the concrete system in
Figure~\ref{fig:semaphore-car}, but here the abstract interpretation is
imprecise, because 
$\sem[]{\textcolor{lipicsLineGray}{\push; }\, \nexte^*;
\mathsf{g?}}^{\sharp}\alpha(\set{\mathsf{rs}})=
\textcolor{blue}{c}\concat \textcolor{red}{a}$
and
$\sem[]{\textcolor{lipicsLineGray}{\push;}\, \nexte;
\mathsf{\neg d?}}^{\sharp}(\textcolor{blue}{c}\concat
\textcolor{red}{a})= (a\vee c)\concat
\textcolor{blue}{c}\concat \textcolor{red}{a}$, so that
$\sem[]{\PtoR{\psi}}^{\sharp}\alpha(\set{\mathsf{rs}})
=\textcolor{red}{a} \neq \bot$, thus
raising a false alarm.
Note that for nested formulae the stack structure of domain elements, managed by $\push$/$\pop$ operators, allows us to recover  the (abstract) states where the computation started. These are the states that should be returned as counterexamples when the property fails. For instance, in the example above, the stack $(a\vee c)\concat
\textcolor{blue}{c}\concat \textcolor{red}{a}$ is produced, representing an abstract trace where the property fails, and, from this, two $\pop$ operations extract the starting abstract state $\textcolor{red}{a}$.

To eliminate false alarms, we leverage the concept of local completeness in abstract interpretation.
Roughly, the idea consists in focusing on the (abstract) computation
path produced by some input of interest, and then refining the
abstraction only when needed to make it complete \emph{locally} to
such computation. More precisely, we apply a variation of local
completeness logic (LCL)~\cite{DBLP:journals/jacm/BruniGGR23,BruniGGR21}, which is here
extended to deal with fixpoint operators.
The LCL proof system, parametrised by a generic state abstraction $A$,
works with O'Hearn-like judgements $\vdash_A [P]\ \r \ [Q]$, where
$\r$ is a program and $P, Q$ denote state properties or,
equivalently, the underlying sets of states satisfying those properties.
The triple $\vdash_A [P]\ \r \ [Q]$ is valid when $Q$ is an under-approximation of the states
reachable by $\r$ from $P$ while the abstraction of $Q$ over-approximates such reachable states,
and the abstract computation of $\r$ on the precondition $P$ is locally complete.
Roughly, this can be expressed as $Q \subseteq \sem[]{\r} P \subseteq \gamma\circ\alpha(Q)$.
In our setting, where the program  $\PtoR{\varphi}$ associated with a
formula $\varphi$ ``returns'' all the counterexamples to the validity of
$\varphi$, an inference of $\vdash_A [P]\ \PtoR{\varphi} \ [Q]$ shows
that $Q\subseteq \set{\sigma\in P \mid \sigma\not\models \varphi}\subseteq \gamma\circ\alpha(Q)$, thus bounding the set of counterexamples to
$\varphi$ in $P$.
Hence, if $Q\neq \varnothing$ then $\varphi$ does not hold for each
$\sigma\in Q\subseteq P$, while $\varphi$  holds for all states $\sigma\in P\setminus (\gamma\circ\alpha(Q))$.
If, instead, $Q = \varnothing$, then $\alpha(Q)=\bot$, and 
$\varphi$ holds in the whole $P$.
As LCL derivations cannot succeed with locally incomplete
abstractions, if some proof obligation fails,
the abstract domain needs to be ``fixed'' to
achieve local completeness. This can be accomplished, e.g., 
by applying the abstraction repair techniques defined in~\cite{DBLP:conf/pldi/BruniGGR22}.

For the toy traffic light example, we can derive in LCL the program
triple $\vdash_A [\set{\mathsf{rs}}]\  \PtoR{\varphi}\  [\varnothing]$,
that confirms the validity of the formula $\varphi$.
Vice versa, an attempt to derive the triple
$\vdash_A [\set{\mathsf{rs}}]\  \PtoR{\psi}\  [\varnothing]$
fails because of local incompleteness.
Roughly, we can successfully derive
$\vdash_A [\set{\mathsf{rs}}] \
\nexte^*;\mathsf{g?}; \
[\set{\mathsf{gs},\mathsf{gd}}]$,
but then the execution of $\nexte$ on $\set{\mathsf{gs},\mathsf{gd}}$
is not locally complete, because
$\alpha(\sem[]{\nexte}\set{\mathsf{gs},\mathsf{gd}}) =
\alpha(\set{\mathsf{gd},\mathsf{yd}}) = b\wedge c$,
while
$\sem[]{\nexte}^{\sharp}\alpha(\set{\mathsf{gs},\mathsf{gd}}) =
\sem[]{\nexte}^{\sharp} \textcolor{blue}{c} = a\vee c$.
The abstraction repair procedure of~\cite{DBLP:conf/pldi/BruniGGR22} would
then lead to refine the abstract domain by adding a new abstract
element $c_1$ to represent the concrete set
$\gamma(c_1)=\set{\mathsf{gs},\mathsf{gd}}$.
Since abstract domains must be closed under meets, the addition of
$c_1$ will also entail the addition of $c_1\wedge
b$
such that $\gamma(c_1\wedge b)=\set{\mathsf{gd}}$.
Letting $A_1 \triangleq A\cup \set{c_1\wedge b,c_1}$, we can
still derive
${\vdash_{A_1} [\set{\mathsf{rs}}]
\ \nexte^*;\mathsf{g?} \
[\set{\mathsf{gs},\mathsf{gd}}]}$, then
$\vdash_{A_1} [\set{\mathsf{gs},\mathsf{gd}}]
\ \nexte\
[\set{\mathsf{gd},\mathsf{yd}}]$, and finally
$\vdash_{A_1} [\set{\mathsf{gd},\mathsf{yd}}]
\ \neg\mathsf{d?}\
[\varnothing]$, which can be composed together to conclude
$\vdash_{A_1} [\set{\mathsf{rs}}]\  \PtoR{\psi}\  [\varnothing]$.
In fact, in $A_1$ we just have
$\sem[]{\PtoR{\psi}}^{\sharp}\alpha_1(\set{\mathsf{rs}}) = \bot$.

In summary, our main contribution is a theoretical
  framework for systematically reducing model checking of temporal
  logics to program verification, that enables to re-use, directly or
  with little effort, abstract interpretation techniques and verifiers.
  This framework consists of:
  \begin{itemize}
  \item a meta-programming language $\MOKA$, where a number of temporal
    logics can be encoded in such a way that the $\MOKA$ encoding of a 
    formula $\varphi$ computes all and only the counterexamples to $\varphi$;
  
  \item a systematic technique for lifting abstractions on state properties
    to abstractions suited for analysing $\MOKA$ programs;
    
  \item an extension of the LCL program logic for dealing with least
    fixpoints, which enables a ``false alarm-guided'' abstraction refinement loop
     in the analysis of $\MOKA$ programs.
  \end{itemize}

\subparagraph{Synopsis.}
In~\S~\ref{sec:back} we provide some basics about
abstract interpretation and
the (fragments) of temporal logics considered
in the paper.
In~\S~\ref{sec:mocha} we introduce the language
$\MOKA$, and then prove in~\S~\ref{sec:s4} the key results that relate formula
satisfaction with program execution.
In~\S~\ref{sec:stackabs} we define a general technique for deriving
abstract domains for the static analysis of $\MOKA$ programs.
In~\S~\ref{sec:lcl} we showcase how local completeness logic reasoning can be
exploited in our framework.
In~\S~\ref{sec:rw} we discuss related work.
Finally, in~\S~\ref{sec:conc} we draw some conclusions and
sketch future avenues of research.

\section{Background}
\label{sec:back}

A complete lattice is a poset $(L, \leq_L)$ where 
every subset $X \subseteq L$ has both least upper bound (lub) and greatest
lower bound (glb), denoted by $\bigvee_L X$ and $\bigwedge_L X$,
respectively, with $\bot_L\triangleq \bigvee_L \varnothing$ and
$\top_L\triangleq \bigwedge_L \varnothing$.
When no ambiguities can arise, a lattice will be denoted as
$L$ and subscripts will be omitted.

Given two complete lattices $L_1$ and $L_2$, a
function $f: L_1 \to L_2$, is monotone if
$x \leq_1 y$ implies $f(x) \leq_2 f(y)$, and 
additive
(resp., co-additive)
if it preserves arbitrary lub
(resp., glb).
Any monotone function $f:
L \to L$ on a complete lattice has both a least and a greatest
fixpoint, denoted by $\operatorname{lfp}(f)$ and
$\operatorname{gfp}(f)$, respectively.
The set of functions $f : S \to L$ from a set $S$ to a complete
lattice $L$, denoted
$L^S$, forms a complete lattice when endowed with the pointwise order
s.t.\ $f \leq g$ if for all $s \in S$, $f(s) \leq_L g(s)$. If $L_1, L_2$
are complete lattices, then $L_1 \times L_2$ is their product
lattice endowed with the componentwise order s.t.\
$(x_1,x_2) \leq (y_1,y_2)$ if $x_1 \leq_1 y_1$ and $x_2 \leq_2 y_2$.
We write $L^n$ for the $n$-ary product of $L$ with itself and
$\str{L} \triangleq \set{ \bot, \top} \cup \bigcup_{n \geq 1}
L^n$ for
the complete lattice of non-empty finite sequences in $L$, ordered by
$x_1 \ldots x_n\leq y_1 \ldots y_m$ if $n=m$ and $x_i \leq_L y_i$ for
all $i \in [1,n]$, with top $\top$ and bottom~$\bot$.

\subsection{Abstract Interpretation}

Let us recall the basics of abstract
interpretation~\cite{DBLP:conf/popl/CousotC77} (see~\cite{cousot21} for a thorough account).
Given two complete lattices $C$ and $A$, called the
\emph{concrete} and the \emph{abstract}
domain, respectively, a \emph{Galois connection (GC)}
$\langle \alpha,\gamma \rangle: C \rightleftarrows A$ is a pair of
functions $\alpha : C \to A$ and $\gamma : A \to C$ s.t.
$\alpha(c) \leq_A a$
$\Leftrightarrow$ $c \leq_C \gamma(a)$ for any $c\in C$ and $a\in A$.

The function $\alpha$ is referred to as abstraction map and turns out to be
additive, while $\gamma$ is the concretization map which is always co-additive.
Intuitively, any abstract element $a\in A$ such that $c\leq
\gamma(a)$ is a sound over-approximation for the concrete value $c$, while the
abstraction $\alpha(c)$ is the most precise over-approximation of
$c$ in the abstract domain $A$, i.e.,
$\alpha(c) = \bigwedge_C \set{a \mid c \leq_C \gamma(a)}$ holds.
The notation $A_{\alpha,\gamma}$ denotes an
abstract domain endowed with its underlying
GC, and we will omit subscripts when $\alpha$ and $\gamma$ are
clear from the context.

\begin{example}[Image adjunction]
  \label{ex:image}
  Given any function $f : X \to Y$, 
  let $\dirIm{f}$ and $\invIm{f}$ denote the direct and inverse image of $f$,
  respectively.
  The pair
  $\langle \dirIm{f}, \invIm{f} \rangle : \pow{X} \to \pow{Y}$
  is a GC that we refer to as the
  \emph{image adjunction} (this is an instance of~\cite[Exercise
  7.18]{dp:lattices-order}).
  \qed
\end{example}

The class
of abstract domains
on $C$, denoted by
$\operatorname{Abs}(C) \triangleq \set{A_{\alpha, \gamma} \mid \langle
\alpha,\gamma \rangle: C \rightleftarrows A}$, can be preordered by
the domain refinement relation:
$A' \sqsubseteq
A$ when $\gamma_{A}(A)\subseteq \gamma_{A'}(A')$.

\begin{example}[Control flow graphs and predicate abstraction]\label{ex:cfg}
  Any program can be represented by its control flow
  graph (CFG). Let $\TVar$ be a set of variables valued in
  $\mathbb{V}$ and denote by
  $\mathit{Env} = \mathbb{V}^\TVar$ the set of environments.
  A CFG is a graph $(N,E,s,e)$, where $N$ is a
  finite set of nodes, representing program points,
  $s, e \in N$ are the start and end nodes, respectively,
  and $E \subseteq N \times F \times N$ is a set of edges, labelled
  over a set of transfer (additive) functions
  $F \subseteq {\pow{\mathit{Env}}}^{\pow{\mathit{Env}}}$.
  For example, the program {\tt{c}} in
  Figure~\ref{fig:prg}
  is
  decorated with program points $n \in N = \set{s, 1, 2, 3, e}$, has variables
  $\mathit{Var} = \set{x, y, z, w}$, and values  ranging in the finite domain
  $\mathbb{V} = \ev$ of integers modulo a given $k>0$. The
  CFG of {\tt{c}} is depicted in Figure~\ref{fig:cfg}. 

  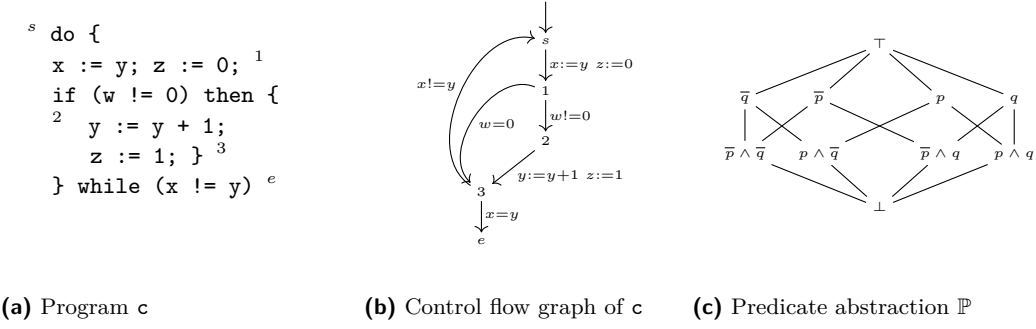
\begin{figure}[t]
    \centering
    \begin{subfigure}[b]{0.25\textwidth}
      \centering
      \small
      \begin{alltt}{
 \( \sp{s} \)do \{
    x := y; z := 0; \(\sp{1} \)
    if (w != 0) then \{
    \(\sp{2} \) y := y + 1;
       z := 1; \} \(\sp{3} \)
    \} while (x != y) \(\sp{e} \)
        }
      \end{alltt}
      \caption{Program {\tt{c}}}
      \label{fig:prg}
    \end{subfigure}
    \hfill
    \begin{subfigure}[b]{0.3\textwidth}
      \centering
      \tiny
      \begin{tikzcd}
        & \ar[d] \\
        & s \ar[d, "x := y \ z := 0"] \\
        & 1 \ar[ddl, bend right = 80, " w = 0"] \ar[d, "w != 0"]\\
        & 2 \ar[dl, "y := y + 1 \ z := 1"] \\
        3 \ar[d, "x = y"] \ar[uuur, bend left = 80, "x != y"]& \\
        e %
        &  \\
      \end{tikzcd}
      \caption{Control flow graph of {\tt{c}}}
      \label{fig:cfg}
    \end{subfigure}
    \begin{subfigure}[b]{0.35\textwidth}
      \centering
      \tiny
      \begin{tikzcd}[column sep=small]
        & & \top \ar[dll, dash] \ar[dl, dash] \ar[dr, dash]
        \ar[drr, dash]& & \\
        \overline{q} \ar[d, dash] \ar[dr, dash] & \overline{p}
        \ar[dl, dash] \ar[drr, dash]& & p \ar[dll, dash] \ar[dr,
        dash]& q \ar[dl, dash] \ar[d, dash]\\
        \overline{p} \land \overline{q} \ar[drr, dash]& p \land
        \overline{q} \ar[dr, dash]& & \overline{p} \land q \ar[dl,
        dash] & p \land q \ar[dll, dash]\\
        && \bot && \\ \\
      \end{tikzcd}
      \caption{Predicate abstraction $\mathbb{P}$}
      \label{fig:predicate}
    \end{subfigure}
    \caption{Program from~\cite[Figure
      1]{DBLP:conf/tacas/BallPR01}, with control flow graph and
    predicate abstraction domain.}
    \label{fig:example}
  \end{figure}

  Predicate abstraction allows to approximate program
  invariants~\cite{DBLP:journals/csur/JhalaM09}.  Given a set of
  predicates $\mathit{Pred} \subseteq \pow{\mathit{Env}}$ (where any
  $p\in \mathit{Pred}$ has a representation
  $p = \set{\rho \in \mathit{Env} \mid \rho \models p}$), the
  predicate abstraction domain $\mathbb{P}$ is defined by adding to
  $\mathit{Pred}$ the complement predicates
  $\overline{\mathit{Pred}}\triangleq \set{\overline{p}\mid p\in
    \mathit{Pred}}$, and then by closing
  $\mathit{Pred}\cup \overline{\mathit{Pred}}$ under logical
  conjunction.  We define a function
  $\pi: \pow{\mathit{Env}} \to \mathbb{P}$ that associates to each set
  of environments the strongest predicate it satisfies.  For example,
  given the two predicates $p \triangleq (z = 0)$ and
  $q \triangleq (x = y)$, the predicate abstraction domain
  $\mathbb{P}$ induced by the set $\mathit{Pred} =\set{ p, q}$ is
  depicted in Figure~\ref{fig:predicate}. Correspondingly, the product
  abstraction $\mathbb{P}^N$ allows us to represent the abstract state
  of the program as a function that associates to each program point
  $n$ the strongest predicate in $\mathbb{P}$ that holds at
  $n$. Hence, for instance, the set of all possible initial states of
  {\tt{c}}, which is $\set{(s, xyzw) \mid x, y, z, w \in \mathbb{V}}$,
  is represented by the function
  $(s \mapsto \top, 1 \mapsto \bot, 2 \mapsto \bot, 3 \mapsto \bot, e
  \mapsto \bot)$, often written $(s \mapsto \top)$, omitting the
  program points which are mapped to $\bot$.  \qed
\end{example}

An abstract interpreter computes in the underlying abstract domain through
correct (and effective) abstract approximations of concrete functions.
Given
$A_{\alpha, \gamma} \in \operatorname{Abs}(C)$ and a function $f: C
\to C$, an abstract function $f^{\sharp}: A \to A$ is a
correct approximation of $f$ if $\alpha \circ f \leq
f^{\sharp}\circ \alpha$, and it is
a \emph{complete} approximation of $f$ when $\alpha
\circ f = f^{\sharp}\circ \alpha$.
The \emph{best correct approximation}
(bca) of $f$ in $A$,
is defined as $f^{A}
\triangleq \alpha \circ f \circ \gamma: A \to A$,
and turns out to be the most precise correct abstraction, i.e., $f^{A}
\leq f^{\sharp}$ for any other correct approximation $f^{\sharp}$ of $f$.
When $f^{\sharp}$ is complete, then
$f^{\sharp} = f^A$,
thus making completeness an abstract domain property defined by the equation
$\alpha \circ f=\alpha \circ f \circ \gamma \circ \alpha$.
In program
analysis, abstract domains are commonly endowed with correct but
incomplete abstract transfer functions. When
completeness holds, the abstract interpreter is
as precise as possible for the given abstract domain and cannot raise
false alarms when verifying properties
that are expressible in the abstract domain
\cite{DBLP:journals/jacm/GiacobazziRS00}.

\subsection{Transition Systems and Logics}

Temporal logics are typically interpreted
over unlabeled, finite, directed
graphs, whose nodes and edges model states and
transitions between them (i.e., state changes), respectively.

A \emph{transition system} $T$ is a
tuple $(\Sigma, \operatorname{I}, \mathbf{P}, \trel, \vdash)$, where
$\Sigma$ is a finite set of states ranged over by $\sigma$,
$\operatorname{I} \subseteq \Sigma$ is the set of initial
states, $\trel \subseteq \Sigma \times \Sigma$ is the
transition relation, $\mathbf{P}$ is a (finite) set of atomic
propositions, ranged over by $p$, and $\mathbin{\vdash} \subseteq \Sigma
\times \mathbf{P}$ is the
satisfaction relation.
We assume that for any $p \in
\mathbf{P}$ also its negation $\neg p$ is in $\mathbf{P}$,
such that for $\sigma \in \Sigma$
either $\sigma \vdash p$ or $\sigma
\vdash \neg p$ holds (also written $\sigma \centernot \vdash p$), and that
$\mathsf{tt}\in\mathbf{P}$ such that $\sigma \vdash \mathsf{tt}$ for
all $\sigma\in\Sigma$.
We write $\sigma \trel \sigma'$ instead of $(\sigma, \sigma') \in
\trel$ and let $\operatorname{post}(\sigma) = \set{\sigma' \mid \sigma \trel
\sigma'}$ denote
the set of direct successors of $\sigma \in
\Sigma$.
As common in model checking~\cite{CGKPV:MC-book}, we consider systems whose transition
relation is total, %
i.e., for all $\sigma \in \Sigma$ there exists $\sigma' \in \Sigma$
such that $\sigma \trel \sigma'$.

A path is an infinite denumerable state sequence
$(\sigma_i)_{i \in \N}$ such that $\sigma_i \trel \sigma_{i+1}$
for all $i \in \N$.

\begin{example}[Control flow graphs as transition systems]
  \label{ex:lts}
  A CFG $(N,E,s,e)$ can be viewed as a transition system with states
  $\Sigma = N \times \mathit{Env}$ and transition relation defined by
  $(n,\rho) \trel (n', \rho')$ if there is an edge $(n, f, n') \in E$
  such that $\rho' \in f(\set{\rho})$. Additionally, in order to make the
    transition relation total we add self-loops to all the states
    $(e,\rho)$ involving the end node $e$.
  For instance, by using the
  shorthand $(n,xyzw)$ for the states in
  Example~\ref{ex:cfg}, the transition $(s,0111) \trel (1,1101)$ is
  induced by the edge $(s, x:= y  \ z := 0, 1)$ in
  Figure~\ref{fig:cfg}, while $(3,1111) \trel
  (e,1111)$ by $(3, x = y, e)$ in Figure~\ref{fig:cfg}.

  Let us point out that the predicate abstraction  in
  Example~\ref{ex:cfg} naturally lifts to the powerset of states with
  codomain $A= \mathbb{P}^N$, with the abstraction map
  $\alpha(X)(n) \triangleq \pi(\set{\rho \mid (n, \rho) \in X})$ for
  $X \in \pow{\Sigma} = \pow{N \times \mathit{Env}}$.
  Given $\sigma^{\sharp} \in A$, we define
  $\supp{\sigma^{\sharp}} \triangleq \set{n \in N \mid \sigma^{\sharp}(n) \neq \bot}$. 
  \qed
\end{example}

\subparagraph{ACTL.}
ACTL is the fragment of CTL whose temporal formulae are universally
quantified over all
paths leaving the current state. Thus,
given a set of atomic propositions $p\in \mathbf{P}$: 
\begin{center}
  $\text{ACTL}\ni \varphi \ ::=  \ p \mid
  \neg p \mid
  \varphi_1 \lor \varphi_2 \mid
  \varphi_1 \land \varphi_2 \mid
  \AX \varphi_1 \mid
  \AF \varphi_1 \mid
  \AG \varphi_1 \mid
  \varphi_1 \AU \varphi_2$
\end{center}

\begin{definition}[ACTL semantics]\label{ACTL-semantics}
  Given a transition system
  $T = (\Sigma, \operatorname{I}, \mathbf{P}, \trel, \vdash)$, the
  semantics %
  $\sem[]{\varphi} \subseteq \Sigma$ of
  {\rm \text{ACTL}} formulae over $T$ is as follows:
  \begin{center}
    \begin{tabular}{r@{\hspace{2pt}}c@{\hspace{2pt}}l@{\hspace{-130pt}}r@{\hspace{2pt}}c@{\hspace{2pt}}l}
      $\sem[]{p}$ & $\triangleq$ & $\set{\sigma \in \Sigma \mid
      \sigma \vdash p}$ & 
      $\sem[]{\neg p}$ & $\triangleq$ &$\set{\sigma \in \Sigma \mid
      \sigma \vdash \neg p}$ \\
      $\sem[]{\varphi_1 \lor \varphi_2}$  & $\triangleq$ &
      $\sem[]{\varphi_1} \cup \sem[]{\varphi_2}$ &
      $\sem[]{\varphi_1 \land \varphi_2}$ & $\triangleq$ &
      $\sem[]{\varphi_1} \cap \sem[]{\varphi_2}$ \\
      $\sem[]{\AX \varphi_1}$ & $\triangleq$ & $\set{\sigma \in \Sigma
        \mid \forall \sigma' . \sigma \trel
      \sigma' \Rightarrow \sigma' \in \sem[]{\varphi_1}}$ \\
      $\sem[]{\AF \varphi_1}$ & $\triangleq$ & $\set{\sigma_0 \in \Sigma \mid
        \textrm{for all path}\ (\sigma_i)_{i \in \N}~\exists
      k\in\mathbb{N}.~ \sigma_k\in \sem[]{\varphi_1}}$ \\
      $\sem[]{\AG \varphi_1}$ & $\triangleq$ & $\set{\sigma_0 \in \Sigma \mid
        \textrm{for all path}\ (\sigma_i)_{i \in \N}~\forall
      j\in\mathbb{N}.~ \sigma_j\in \sem[]{\varphi_1}}$ \\
      $\sem[]{\varphi_1 \AU \varphi_2}$ & $\triangleq$ &
      $\set{\sigma_0 \in \Sigma \mid
        \textrm{for all path}\ (\sigma_i)_{i \in \N}~ \exists k \in
        \mathbb{N} . ~ (\sigma_k \in \sem[]{\varphi_2} \land \forall j <
      k . ~ \sigma_j \in \sem[]{\varphi_1})}$ \\
    \end{tabular}
  \end{center}
\end{definition}

\subparagraph{Universal fragment of single variable
$\mu$-calculus.}
The modal $\mu$-calculus is a well known extension of propositional modal logic
with least and greatest fixed point operators.  We will focus on its universal
fragment only allowing the $\Box$ modal operator that
quantifies over all transitions. Moreover, for the sake of
simplicity, we restrict to
the single variable fragment where, roughly speaking,
nested fixpoints cannot have mutual dependencies.

Given a set of atomic propositions $p\in \mathbf{P}$, the
$\mu_\Box$-calculus is defined as follows: 
\begin{center}
  $\mu_\Box \ni \varphi \ ::= \ p \mid \neg p \mid \varphi_1 \lor
  \varphi_2 \mid \varphi_1
  \land \varphi_2
  \mid \Box \varphi_1  \mid x \mid \mu x . \varphi_x \mid \nu x. \varphi_x$
\end{center}

\begin{definition}[$\mu_\Box$-calculus semantics]
  \label{mu-semantics}
  Given
  $T = (\Sigma, \operatorname{I}, \mathbf{P}, \trel, \vdash)$,
  and a valuation $\mathcal{V}: \LVar \to \pow{\Sigma}$,
  the semantics
  $\sem[\mathcal{V}]{\varphi}
  \subseteq \Sigma$ of
  $\mu_\Box$-calculus formulae over $T$ 
  is as follows:
  \begin{center}
    \begin{tabular}{r@{\hspace{2pt}}c@{\hspace{2pt}}l
      r@{\hspace{2pt}}c@{\hspace{2pt}}l}
      $\sem[\mathcal{V}]{p}$ & $\triangleq$ & $\set{\sigma \in \Sigma
      \mid \sigma \vdash p}$  &  $\sem[\mathcal{V}]{\neg p}$ & $\triangleq$ & $\set{\sigma \in \Sigma
      \mid \sigma \vdash \neg p}$  \\
      $\sem[\mathcal{V}]{\varphi_1 \lor \varphi_2}$ & $\triangleq$ & 
      $\sem[\mathcal{V}]{\varphi_1} \cup
      \sem[\mathcal{V}]{\varphi_2}$ &
      $\sem[\mathcal{V}]{\varphi_1
      \land \varphi_2}$ & $\triangleq$ &
      $\sem[\mathcal{V}]{\varphi_1} \cap \sem[\mathcal{V}]{\varphi_2}$  \\
      $\sem[\mathcal{V}]{\Box \varphi_1}$ & $\triangleq$ & $\set{\sigma
        \in \Sigma \mid \forall \sigma' . \sigma \trel \sigma'
      \Rightarrow \sigma' \in \sem[\mathcal{V}]{\varphi_1}}$ &
      $\sem[\mathcal{V}]{x}$ & $\triangleq$ & $\mathcal{V}(x)$ \\
      $\sem[\mathcal{V}]{\mu x . \varphi_x}$ & $\triangleq$ &
      $\operatorname{lfp}(\lambda S. \llbracket \varphi_x
      \rrbracket_{\mathcal{V}[x \mapsto S]})$ &
      $\sem[\mathcal{V}]{\nu x . \varphi_x}$ & $\triangleq$ &
      $\operatorname{gfp}(\lambda S. \llbracket \varphi_x
      \rrbracket_{\mathcal{V}[x \mapsto S]})$ \\
    \end{tabular}
  \end{center}
  \noindent
  where $\mathcal{V}[x \mapsto S]$ is the usual notation for function update.
\end{definition}

\noindent
  We write $\sem[]{\varphi}$ instead of
  $\sem[\mathcal{V}]{\varphi}$ when the valuation is inessential,
and $\sigma\models\varphi$ when $\sigma\in\sem[]{\varphi}$. 

\section{The Language {$\MOKA$}}
\label{sec:mocha}

We define a meta-language, called $\MOKA$ (for $\mathsf{MO}$del
checking as abstract interpretation of $\mathsf{K}$leene $\mathsf{A}$lgebras), as a
(generalised) Kleene Algebra with a
set
of basic expressions suited for identifying counterexamples to the validity of
temporal formulae.

\subparagraph{$\kaf$.}
We rely on Kozen's Kleene Algebra with
tests~\cite{DBLP:journals/toplas/Kozen97}
with
a Fixpoint operator~\cite{DBLP:conf/csl/Leiss91},
$\kaf$ for short.
Given a set $\mathsf{Exp}$
of basic expressions $\mathsf{e}$, $\kaf$ is defined below:
\begin{center}
  $\kaf \ni \r  :: = \one \mid \zero \mid \mathsf{e} \mid
  \r_1;\r_2 \mid
  \r_1 \oplus \r_2 \mid \r_1^* \mid \X \mid \upmu \X. \r_1$
\end{center}
The term $\one$ represents the identity, i.e.\ no
action,  $\zero$ represents divergence,
$\r_1; \r_2$
represents sequential composition,
$\r_1 \oplus \r_2$ represents non-deterministic choice,
$\r_1^{*}$ represents the Kleene
iteration of $\r_1$, i.e.\ $\r_1$
performed zero or any finite number of times,
$\X$ is a variable ranging in a set $\RVar$, and
$\upmu \X. \r_1$
represents the least fixpoint
operator with respect to variable $\X$.

Commands are interpreted as functions over a complete
lattice $C$.  Given a semantics
$\bsem{\cdot}: \mathsf{Exp} \to C \to C$ for basic
expressions,
the semantics of regular
expressions $\sem{\cdot}: \kaf \to C \to
C$ is
inductively defined as follows, where $\eta: \RVar \to
C \to C$ is an
environment:
\begin{center}
  \begin{tabular}{r c l c c r c l}
    $\sem{\one}$ & $\triangleq$ & $\lambda x.\, x$ & & &
    $\sem{\zero}$ & $\triangleq$ &
    $\lambda x.\, \bot$ \\
    $\sem{\e}$ & $\triangleq$ & $\bsem{\e}$ & & &
    $\sem{\mathsf{r_1; r_2}}$ & $\triangleq$ &
    $\sem{\mathsf{r_2}}\circ \sem{\mathsf{r_1}}$ \\
    $\sem{\mathsf{r_1 \oplus r_2}}$ & $\triangleq$ & $\sem{\mathsf{r_1}}
    \lor \sem{\mathsf{r_2}}$ & & &  $\sem{\mathsf{r_1^*}}$ &
    $\triangleq$ & $ \bigvee
    \set{\sem{\mathsf{r_1}}^k \mid k \in \N}$ \\
    $\sem{\X}$ & $\triangleq$ & $\eta(\X)$ & & &  $\sem{\upmu \X.
    \mathsf{r_1}}$ &
    $\triangleq$ & $\operatorname{lfp}(\lambda
    f: C \to C. \sem[{\eta[\X \mapsto f]}]{\r_1})$ \\
  \end{tabular}
\end{center}

It can be seen that the Kleene star $\r^*$ can be encoded as a
fixpoint  $\upmu \X . (\one \oplus \r; \X)$, 
and, when the semantics of basic expressions
$\bsem{\cdot}$ is additive, also as $\r^* = \upmu \X . (\one
\oplus \X; \r)$
(Lemma~\ref{le:encoding-fix2})
while this fails in
general for non-additive semantics (Example~\ref{ex:non-additive-fix}).
Even if redundant, we include the Kleene star in our language as it allows to simplify some encodings  whenever the full expressiveness of least fixpoint calculation is not required (see the remark at the end of Section~\ref{sec:s4}). The same applies to the proof logic: as discussed in Section~\ref{sec:lcl}, basic LCL suffices for the fragment of $\kaf$ without least fixpoint operator, while the extension $\upmu$LCL is needed for dealing with the full language.
For closed $\kaf$ terms the environment is inessential
and we write just $\sem[]{\mathsf{r}}$ instead of $\sem{\mathsf{r}}$.
\begin{toappendix}
  \begin{lemma}
    \label{le:additivity-regular-expressions}
    If the semantics of the basic expressions
    is additive and $\eta$ maps each variable to an
    additive function, then $\sem{\r}$ is additive for any $\kaf$ term $\r$.
  \end{lemma}

  \begin{proof}
   
    We want to prove that, for any $\kaf$ term $\r$ and any $\a, \b \in C$, 
    it holds $\sem{\r}(\a \lor \b) = \sem{\r}\a \lor \sem{\r}\b$.
    The proof is by structural induction on the term $\r$.
    Let $\a, \b \in C$.
    
    \subparagraph*{($\r = \e$)} $\sem{\e} = \bsem{\e}$ is additive by
    hypothesis.

    \subparagraph*{($\r = \r_1; \r_2$)} By inductive hypothesis $\sem{\r_1},
    \sem{\r_2}$ are additive. Since the composition of additive functions is additive, so is 
    $\sem{\r_1; \r_2}=\sem{\r_2}\circ \sem{\r_1}$.

    \subparagraph*{($\r = \r_1 \oplus \r_2$)} By inductive hypothesis $\sem{\r_1},
    \sem{\r_2}$ are additive. Since the join of additive functions is additive, so is 
    $\sem{\r_1\oplus \r_2}=\sem{\r_1}\vee \sem{\r_2}$.

    \subparagraph*{($\r = \r_1^*$)} 
    By definition $\sem{\r_1^*} = \bigvee_k \set{\sem{\r_1}^k \mid k
    \in \N}$, where $\sem{\r_1}$ is additive 
    by inductive hypothesis.
    First, we observe that each $\sem{\r_1}^k$ is additive is additive because it is the composition of $k$ additive functions.
    And then we can conclude because the join of additive
    functions is additive.

    \subparagraph*{($\r = \X$)} $\sem{\X} = \eta(\X)$
    is additive by the hypothesis on $\eta$.  

    \subparagraph*{($\r = \upmu \X. \r_1$)} Let $G : C^C \to C^C$ be defined by
    $G(f) = \sem[{\eta[\X \mapsto f]}]{\upmu \X. \r_1}$. Then by Kleene
    fixpoint theorem we have $\sem{\r} = \operatorname{lfp}(G) =
    G^\alpha(\bot_\to)$ for some
    ordinal $\alpha$, where $\bot_\to = \lambda c.\, \bot_{C}$. We prove
    the statement by transfinite induction.

    \begin{itemize}
      \item If $\alpha = 0$ we have $G^0(\bot_\to) = \bot_\to$
        which is additive.
      \item If $\alpha = n + 1$ (non-limit ordinal) we have
        \begin{align*}
          G^{n+1}(\bot_\to)(\a \lor \b)
          & = G(G^n(\bot_\to)(\a \lor \b))
          & [\mbox{definition of $G$}] \\
          & = G ( G^n(\bot_\to)(\a) \lor
          G^n(\bot_\to)(\b) ) & [\mbox{inductive hypothesis}] \\
          & = G ( G^n(\bot_\to)(\a) ) \lor
          G  ( G^n(\bot_\to)(\b)  )
          & [\mbox{inductive hypothesis}] \\
          & =  G^{n+1}(\bot_\to)(\a) \lor
          G^{n+1}(\bot_\to)(\b)
          & [\mbox{definition of  $G^{n+1}$}] \\
        \end{align*}
      \item If $\alpha$ is a limit ordinal we have
        \begin{equation*}
          G^{\alpha}(\bot_\to) = \bigvee_{\beta \leq \alpha}
          G^{\beta}(\bot_\to)
        \end{equation*}
        which is additive since it is the join of additive functions
        (by inductive hypothesis).
    \end{itemize}
  \end{proof}

  \begin{lemma}
    \label{le:encoding-fix}
    The encoding of $\r^*$ as $\upmu \X. (\one \oplus \r; \X)$ is correct.
  \end{lemma}
  \begin{proof}
    Observe that $\sem{\upmu \X. (\one \oplus \r; \X)} =
    \operatorname{lfp}(G)$ where $G: C^C \to C^C$  
    is defined by
    \begin{center}
      $G(f) = \sem[{\eta[\X \mapsto f]}]{\one \oplus \r; \X}
      = \lambda c.\, c \lor f(\sem{\r}c)$
    \end{center}
    Let $\bot_\to = \lambda c. \bot_{C}$  be the bottom element of
    $C^{C}$. One can show that $G^0(\bot_\to) = \bot_\to$ and, inductively,

    \[
      G^n(\bot_\to) = \lambda c.\, \bigvee_{i \leq n-1} \sem{\r}^i(c).
    \]
    Therefore
    $\sem{\upmu \X. (\one \oplus \r; \X)} = \operatorname{lfp}(G) =
    G^\omega(\bot_\to) = \lambda c.\, \bigvee_{i < \omega}
    \sem{\r}^i(c) = \sem{\r^*}$.
  \end{proof}

  \begin{lemma}
    \label{le:encoding-fix2}
    When the semantics is additive $\r^*$ can be also encoded as
    $\upmu \X. (\one \oplus \X;r)$
  \end{lemma}
  \begin{proof}
    Observe that $\sem{\upmu \X. (\one \oplus \X;\r)} =
    \operatorname{lfp}(G)$ where $G: C^C \to C^C$ is defined by
   \begin{center}
      $G(f) = \sem[{\eta[\X \mapsto f]}]{\one \oplus \X;\r}
      = \lambda c.\, c \lor \sem{\r}(f(c))$
    \end{center}
    Let $\bot_\to = \lambda c. \bot_C$  be the bottom element of
    $C^C$. One can show that $G^0(\bot_\to) = \lambda c.\, c$
    and, inductively,
    \[
      G^n(\bot_\to) = \lambda c.\, \bigvee_{i \leq n} \sem{\r}^i(c).
    \]
    In fact,
    $G^{n+1}(\bot_\to) = G(G^n(\bot_\to)) = \lambda c.\, c \lor
    \sem{\r}(\bigvee_{i \leq n} \sem{\r}^i(c)) = \lambda c.\, c \lor
    \bigvee_{i \leq n} \sem{\r}^{i+1}(c) = \lambda c.\, \bigvee_{i
    \leq n+1} \sem{\r}^i(c)$, where the second last passage is
    allowed by additivity of $\sem{\r}$.
    Hence we conclude as in Lemma~\ref{le:encoding-fix}.
  \end{proof}

  \begin{example}[Non-additive $f$ fixpoint encoding]
    \label{ex:non-additive-fix}
    The equation $\r^* = \upmu \X . (\one \oplus \r;\X)$ does not
    generally hold if the semantics is not additive. Consider for
    example $U=\set{a, b, c}$ and the function
    $u: \pow{U} \to \pow{U}$ such that
    $u(\set{a,b}) = u(U) = U$, $u(S) = \set{b}$
    for any other $S \in \pow{U}$ and let $\mathsf{u}$ be a basic
    expression such that $\bsem{\mathsf{u}}=u$. Then
    \begin{align*}
      \sem{\e}^*(\set{a})
      & = \set{a} \cup u^1(\set{a})
      \cup u^2(\set{a}) \cup \dots \\
      & = \set{a} \cup \set{b} \cup u(\set{b}) = \set{a, b}
    \end{align*}
    while instead
    $\sem{\upmu \X . (\one \oplus \X;\mathsf{u})}= \operatorname{lfp}(G)$
    where, for $f : \pow{\set{a,b,c}} \to  \pow{\set{a,b,c}}$ it holds
    \begin{center}
      $G(f) = \sem[{[X \mapsto \beta]}]{\one \oplus \X;\mathsf{u}} = \lambda
      S.\, (S \cup u(f(S)))$
    \end{center}
    If we denote by $\bot = \lambda S.\, \varnothing$ we have
    \begin{align*}
      G^1(\bot) & = \lambda S.\ S \cup \set{b} \\
      G^2(\bot) & = \lambda S.\ S \cup u(S \cup \set{b})
    \end{align*}
    One can see that $G^2(\bot)(S) = U$ if $a \in S$, 
      while $G^2(\bot)(S)=S \cup \set{b}$ otherwise. 
      Therefore,      
      $G^2(\bot) = \operatorname{lfp}(G)$ and
      thus $\sem{\upmu \X . (\one \oplus \X;\mathsf{u})}(\set{a}) =
      G^2(\bot)(\set{a}) = U \supsetneq \sem{\mathsf{u}}^*(\set{a})$.
      \qed
    \end{example}
  \end{toappendix}

  \begin{toappendix}
    We write $\subst{\r}{\X}{\s}$ for the (capture
    free) substitution of the free occurrences of $\X$ for $\s$ in $\r$.

    \begin{lemma}[Substitution lemma]
      \label{le:substitution-lemma}
      Let $\eta$ be an environment, $\r, \s$ be $\kaf$ terms,
      $\X \in \RVar$. Then we have
      $\llbracket \r
        \rrbracket_{\eta[\X\mapsto\sem{\s}]} =
        \sem{\subst{\r}{\X}{\s}}$.
       \end{lemma}

    \begin{proof}
      The proof goes by structural induction on $\r$.
      
      \subparagraph*{($\r = \e$)}
      In this case $\llbracket
      \e\rrbracket_{\eta[\X \mapsto \sem{\s}]} =
      \sem{\e} = \sem{\subst{\e}{\X}{\s}}$.

      \subparagraph*{($\r = \Y$)}
      If $\Y=\X$ then $\llbracket \X\rrbracket_{\eta[\X
      \mapsto \sem{\s}]} =
      \eta[\X \mapsto \sem{\s}](\X) = \sem{\s} =
      \sem{\subst{\X}{\X}{\s}}$. Otherwise
      $\llbracket \Y\rrbracket_{\eta[\X
      \mapsto \sem{\s}]} =
      \eta[\X \mapsto \sem{\s}](\Y) = \eta(\Y) =
      \sem{\subst{\Y}{\X}{\s}}$

      \subparagraph*{($\r = \ \r_1 \oplus \r_2$)} 
      Let $c \in C$. We have
      \begin{align*}
        \llbracket \r_1 \oplus \r_2 \rrbracket_{\eta[\X \mapsto
        \sem{\s}]} (c)
        & =
        \llbracket \r_1\rrbracket_{\eta[\X \mapsto \sem{\s}]} (c)
        \vee \llbracket \r_2\rrbracket_{\eta[\X \mapsto
        \sem{\s}]} (c) \\
        & \overset{\clubsuit}{=} \sem{\subst{\r_1}{\X}{\s}}c \vee
        \sem{\subst{\r_2}{\X}{\s}}c \\
        & = \sem{\subst{\r_1}{\X}{\s} \oplus
        \subst{\r_2}{\X}{\s}}c \\
        & = \sem{\subst{(\r_1 \oplus \r_2)}{\X}{\s}}c
      \end{align*}
      where at $(\clubsuit)$ we used the inductive hypothesis.
      \subparagraph*{($\r = \ \r_1; \r_2$)} 
      \begin{align*}
        \llbracket \r_1 ; \r_2 \rrbracket_{\eta[\X \mapsto \sem{\s}]}
        & = \llbracket \r_2 \rrbracket_{\eta[\X \mapsto \sem{\s}]}
        \circ \llbracket \r_1 \rrbracket_{\eta[\X \mapsto
        \sem{\s}]} \\
        & \overset{\clubsuit}{=} \sem{\subst{\r_2}{\X}{\s}} \circ
        \sem{\subst{\r_1}{\X}{\s}}  \\
        & = \sem{\subst{\r_1}{\X}{\s} ; \subst{\r_1}{\X}{\s}} \\
        & = \sem{\subst{(\r_1; \r_2)}{\X}{\s}}
      \end{align*}
      where at $(\clubsuit)$ we used the inductive hypothesis.
      \subparagraph*{($\r = \upmu \Y. \r_1$)} 
      If $\Y=\X$, i.e., $\r = \upmu \X. \r_1$ then we conclude trivially
      since $\subst{(\upmu \X . \r_1)}{\X}{\s} = \upmu \X .
      \r_1$, and thus:
      \begin{align*}
        \llbracket \upmu \X . \r_1 \rrbracket_{\eta[\X
        \mapsto \sem{\s}]}
        & = \sem{\upmu \X . \r_1} 
        = \sem{\subst{(\upmu \X . \r_1)}{\X}{\s}}
      \end{align*}
      as desired, where, in the first equality, we used the fact that if $\eta(\Y) =
      \eta'(\Y)$ for all
      $\Y \in \mathsf{fv}(\r)$ then $\sem{\r} = \sem[\eta']{\r}$.

      \noindent
      Let instead $\r = \upmu \Y. \r_1$ with $\Y \neq \X$. By alpha-renaming, we can assume without
      loss of generality $\Y \notin \mathsf{fv}(\s)$. Then, letting 
      $F(g) = \llbracket \r_1  \rrbracket_{\eta[\X \mapsto
        \sem{\s}][\Y \mapsto g]}$,
      by definition of least fixpoint we
      have
      \begin{align*}
        \llbracket\upmu \Y. \r_1  \rrbracket_{\eta[\X \mapsto \sem{\s}]}
        & = \bigwedge \set{g \mid F(g) \leq g} \\
        & = \bigwedge \set{g \mid \llbracket \r_1
          \rrbracket_{\eta[\X \mapsto \sem{\s}][\Y \mapsto g]}
        \leq g} \\
        & = \bigwedge \set{g \mid \llbracket \r_1
          \rrbracket_{\eta[\Y \mapsto g][\X \mapsto \sem{\s}]}
        \leq g} \\
        & \overset{\clubsuit}{=} \bigwedge \set{\llbracket \subst{\r_1}{\X}{\s}
        \rrbracket_{\eta[\Y \mapsto g]} \leq g} \\
        & = \sem{\upmu \Y . (\subst{\r_1}{\X}{\s})} \\
        & = \sem{\subst{(\upmu \Y . \r_1)}{\X}{\s}}
      \end{align*}
      where at $(\clubsuit)$ we used the inductive hypothesis.
      \hfill
    \end{proof}

    \begin{remark}
      Notice that
      $\sem{\upmu \X . \r} = \sem{\subst{\r}{\X}{\upmu \X.\r}} =
      \llbracket \r
      \rrbracket_{\eta[\X\mapsto\sem{\upmu \X . \r}]}$ where the first
      equality follows from the substitution lemma
      (Lemma~\ref{le:substitution-lemma}), while
      $ \sem{\upmu \X . \r} = \llbracket \r
      \rrbracket_{\eta[\X\mapsto\sem{\upmu \X . \r}]}$ since
      if we define
      $F_{\r} : C^C \to C^C$ 
      to be such that
      $F_{\r}(f) = \llbracket \r \rrbracket_{\eta[\X
      \mapsto f]}$ then
      $\sem{\upmu \X . \r} = \operatorname{lfp}(F_{\r})$ by definition,
      hence
      $\sem{\upmu \X . \r} = F_{\r} \Bigl ( \sem{\upmu \X . \r} \Bigr)=
      \llbracket \r \rrbracket_{\eta[\X\mapsto\sem{\upmu \X . \r}]}$.
    \end{remark}

  \end{toappendix}

  \subparagraph{The Language $\MOKA$.} Given a transition system
  $T = (\Sigma, \operatorname{I}, \mathbf{P}, \trel, \vdash)$, 
  $\MOKA$ is an instance of $\kaf$ interpreted
  over stacks of frames,
  each frame representing a computation path. %

  \begin{definition}[Frame, stack]
    A \emph{frame} is a pair
    $\frame{\sigma, \Delta} \in \Sigma \times \pow{\Sigma}$. We
    denote by $\operatorname{F}_{\Sigma}$ the set of frames. A
    \emph{stack} is a finite non-empty sequence of frames, i.e.
    an element of
    $\str{\operatorname{F}_{\Sigma}}$,
    denoted by $\frame{\sigma, \Delta}\concat S$, where $S$
    is a stack or the empty sequence $\varepsilon$.
  \end{definition}

  A frame $\frame{\sigma, \Delta}$ represents
  a computation in $T$ where $\sigma$ is the current state and $\Delta$
  is the set of
  traversed states, used for loop-detection.   The
  order of the traversed states and their possible multiple occurrences
  are abstracted
  away as they are irrelevant when checking the satisfaction of a
  formula. Frames are stacked to deal with formulae with
  nested operators.

  The language ${\MOKA}$ is defined as instance of $\kaf$ with basic expressions
  for extending and filtering
  frames, for (de)constructing
  stacks, and to keep
  track of fixed-point equations. \label{ref:mocha-b.e.}

  \begin{definition}[$\MOKA$ language]
    The language $\MOKA$ is an instance of
    $\kaf$ with the
    following basic expressions, where $\ps$ ranges over atomic
    propositions:
    \begin{center}
      $\e :: = \p \mid \np \mid \loops \mid \nexte \mid \add \mid \reset
      \mid  
      \push \mid \pop$
    \end{center}
  \end{definition}

  The (concrete) semantics of $\MOKA$ is given over the powerset of the
  set of stacks, $C = \pow{\str{\operatorname{F}_{\Sigma}}}$,
  ordered by subset inclusion.
  Although we could focus on stacks of uniform length, we use a larger domain to simplify the notation.

  The semantics of $\MOKA$ commands
  $\sem{\r} : \pow{\str{\operatorname{F}_{\Sigma}}} \to
  \pow{\str{\operatorname{F}_{\Sigma}}}$ follows from the general
  definition for $\kaf$, once we specify the
  semantics of its basic expressions.

  \begin{definition}[Basic expression semantics]
  Given 
  $T = (\Sigma, \operatorname{I}, \mathbf{P}, \trel, \vdash)$,
    the semantics of $\MOKA$ basic expressions 
     is
    the additive extension of the functions below, where
    $\frame{\sigma, \Delta}\concat S \in \str{\operatorname{F}_{\Sigma}}$:
    \begin{center}
      \begin{tabular}{r@{\hspace{2pt}}c@{\hspace{2pt}}l
        r@{\hspace{2pt}}c@{\hspace{2pt}}l}
        $\bsem{\p}\set{\frame{\sigma, \Delta} \concat S}$ & $\triangleq$
        & $\set{\frame{\sigma, \Delta}\concat S \mid \sigma \vdash p }$
        & $\bsem{\np}\set{\frame{\sigma, \Delta} \concat S}$ & $\triangleq$
        & $\set{\frame{\sigma, \Delta}\concat S \mid \sigma \vdash \neg p }$
        \\ 
        $\bsem{\loops}\set{\frame{\sigma, \Delta} \concat S}$ & $\triangleq$
        & $\set{\frame{\sigma, \Delta}\concat S \mid \sigma \in \Delta }$ &
        $\bsem{\nexte}\set{\frame{\sigma, \Delta} \concat S}$ & $\triangleq$
        & $\set{\frame{\sigma', \Delta}\concat S \mid \sigma
        \trel \sigma' }$ \\
        $\bsem{\add}\set{\frame{\sigma, \Delta} \concat S}$ & $\triangleq$
        & $\set{\frame{\sigma, \Delta \cup \set{\sigma}}\concat S }$
        & $\bsem{\reset}\set{\frame{\sigma, \Delta} \concat S}$ & $\triangleq$
        & $\set{\frame{\sigma, \varnothing}\concat S }$ \\
        $\bsem{\push}\set{\frame{\sigma, \Delta} \concat S}$ & $\triangleq$
        & $\set{ \frame{\sigma, \Delta}\concat\, \frame{\sigma,
        \Delta}\concat S }$
        & $\bsem{\pop}\set{\frame{\sigma, \Delta} \concat S}$ &
        $\triangleq$ & $\set{S \mid S \neq \varepsilon}$
      \end{tabular}
    \end{center}
  \end{definition}
  
  The basic expressions $\p$, $\np$, $\loops$, $\nexte$,
  $\add$, $\reset$  %
  operate on the top frame.
  The filter $\p$ (resp. $\np$) checks the validity of the
  proposition $p$ (resp. $\neg p$),  $\loops$
  checks if the
  current state loops
  back
  to some state in
  the trace, $\nexte$ extends the trace by
  one step in all possible ways, $\add$ adds
  the current state to the trace and  $\reset$ empties
  the trace.
  Instead,  $\push$ and
  $\pop$ change the stack length:
  $\push$ extends the  stack to start a new
  trace, while 
  $\pop$ restores the previous trace
  from the stack.
  Given a set of states $X \subseteq \Sigma$, we  write
  $\sem{\r}X$ as a shorthand for
  $\sem{\r}(\set{\frame{\sigma, \varnothing} \mid \sigma \in X})$.

  \section{Formulae as $\MOKA$ programs}
  \label{sec:s4}

  We show that ACTL and $\mu_\Box$-calculus formulae $\varphi$
  can be transformed into
  $\MOKA$ programs $\PtoR{\varphi}$
  so that for any transition system $T$, the semantics
  of $\PtoR{\varphi}$ over $T$ consists exactly of
  the counterexamples to the validity of $\varphi$, that is, the states
  which do not satisfy $\varphi$.

  \begin{theorem}[Model checking as program verification]
    \label{th:counterexamples-general}
    Given 
  $T = (\Sigma, \operatorname{I}, \mathbf{P}, \trel, \vdash)$,
    for any ACTL or $\mu_{\Box}$ formula 
    $\varphi$ 
    and set of stacks $P \subseteq \str{\operatorname{F}_{\Sigma}}$,
    it holds
    $\sem{\PtoR{\varphi}} P = \set{\frame{\sigma,\Delta}\concat S \in
      P \mid \sigma \not\models \varphi}.$
  \end{theorem}

\noindent
It follows that $\sem{\PtoR{\varphi}} P\subseteq P$ and that $\sem{\PtoR{\varphi}}\{\sigma\}  = \varnothing$ iff $\sigma \models \varphi$. 
It is also worth noting that
the semantics of $\MOKA$
programs encoding formulae is a lower closure on sets of states, being
monotone, reductive and idempotent (and it preserves arbitrary
unions)~\cite[Section~3.2.3]{mine17}.
In this respect, it behaves analogously to the collecting semantics of
Boolean tests which filter out memory states, keeping only those that
satisfy the condition. As a consequence, similarly to what happens
when abstracting Boolean tests, in any abstract domain $A$
approximating sets of states the identity $\lambda a\in A.\, a$, is a
correct (over-)approximation of $\sem{\PtoR{\varphi}}$.

  \subparagraph{ACTL as $\MOKA$.}
  To each ACTL formula $\varphi$ we assign the $\MOKA$ program
  $\PtoR{\varphi}$
  as follows:
  \begin{center}
    \begin{tabular}{r@{\hspace{2pt}}c@{\hspace{2pt}}l@{\hspace{-130pt}}r@{\hspace{2pt}}c@{\hspace{2pt}}l}
      $\PtoR{\,p}$ & $\triangleq$ & $\np$ & $\PtoR{\,\neg p}$ & $\triangleq$ & $\p$
      \\
      $\PtoR{\varphi_1\land \varphi_2}$ & $\triangleq$ &
      $\PtoR{\varphi_1} \oplus \PtoR{\varphi_2}$
      &
      $\PtoR{\varphi_1\lor \varphi_2}$ & $\triangleq$ &
      $\PtoR{\varphi_1} ; \PtoR{\varphi_2}$ \\
      $\PtoR{\AX \varphi}$ & $\triangleq$ & $\push; 
      \nexte;\PtoR{\varphi}; \pop$ \\
      $\PtoR{\AF \varphi}$ & $\triangleq$ & $\PtoR{\varphi};
      \push; \reset; \left (\add; \nexte;
      \PtoR{\varphi} \right )^*; \loops; \pop$ \\
      $\PtoR{\AG \varphi}$ & $\triangleq$ & $\push;
      \nexte^*;\PtoR{\varphi}; \pop$ \\
      $\PtoR{\varphi_1 \AU \varphi_2}$ & $\triangleq$ &
      $\PtoR{\varphi_2}; \push; \reset; \left
      (\add; \nexte;
      \PtoR{\varphi_2} \right )^*; \left ( \loops \oplus
      \PtoR{\varphi_1} \right ); \pop$ \\
    \end{tabular}
  \end{center}

  The proof of the correctness of this transform can be found in the appendix as Theorem~\ref{th:actl-program}.
  The intuition is that  $\MOKA$ programs $\PtoR{\varphi}$ act
  as (negative) filters: applied to a set of
  frames $\frame{\sigma, \Delta}\concat S$, each
  representing a computation that reached state $\sigma$, they
  filter out states where $\varphi$ is satisfied, keeping only candidate
  counterexamples to the validity of $\varphi$.
  In general, when the check requires exploring
  the future of the current state, a new frame 
  is started with $\push$;
  if the search is successful, a
  closing $\pop$ recovers the starting state violating the formula.
  We explain the main clauses defining 
  $\PtoR{\varphi}$.
  It turns out 
  that $\frame{\sigma, \Delta}\concat S$ is a counterexample for
  \begin{itemize}

    \item $\AX \varphi$ if $\varphi$ is violated in at least
      one successor state of $\sigma$ (as computed by
      $\nexte;\PtoR{\varphi}$). %

    \item $\AF \varphi$ if there is an infinite trace where
      $\varphi$ never holds. This is encoded by saying that
      $\varphi$ does not hold in the current state,
      i.e.\ $\PtoR{\varphi}$,
      and after that, there is an infinite trace traversing only states
      (collected through the command
      $(\add; \nexte;\PtoR{\varphi})^{*}$) that do not satisfy
      $\varphi$. %
      Since we deal with finite state
      systems, infinite
      traces can be identified with looping traces (whence the check
      $\loops$); for this to work, after the $\push$ we $\reset$ the past.

    \item $\AG\varphi$ if we can reach, by repeatedly applying $\nexte$,
      a counterexample to $\varphi$.

    \item $\varphi_1\AU\varphi_2$ if $\varphi_2$ does not hold in the
      current state, i.e.\ $\PtoR{\varphi_2}$, and progressing through
      states that do not satisfy $\varphi_2$ with
      $(\add; \nexte;\PtoR{\varphi_2})^{*}$, either we detect a
      maximal (looping) trace or a state where $\varphi_1$ does not
      hold.
  \end{itemize}

  \begin{example}[Control flow graphs and ACTL encoding]\label{ex:actl} 
    In our running example (see Examples~\ref{ex:cfg}
    and~\ref{ex:lts}), we check the property stating that if
    the program {\tt{c}} terminates then
    the variable $z$ is zero, which in the CFG means
    \emph{``when the exit node $e$ is reached, $z = 0$''} holds. This
    is expressed by
    the \ACTL formula  
    $\varphi = \AG(n = e \to z = 0)$, where we use $p \to \varphi'$ as syntactic sugar for $\neg p \lor \varphi'$.
    The corresponding $\MOKA$ program is:
    $\mathsf{m}_{\varphi} \triangleq \PtoR{\varphi} = \PtoR{\AG (n = e \to z = 0)} = \push; \nexte^*;
    \mathsf{n = e ?} ; \mathsf{z \neq 0?}; \pop$. \qed
  \end{example}

  \begin{toappendix}
    \subparagraph*{ACTL as $\MOKA$}
    \begin{theorem}[Formula satisfaction as program verification
      - ACTL version]
      \label{th:actl-program}
      Given 
      $T = (\Sigma, \operatorname{I}, \mathbf{P}, \trel, \vdash)$, for all frames $\frame{\sigma, \Delta}\concat S \in
      \str{\operatorname{F}_{\Sigma}}$,
      for any \emph{ACTL} formula $\varphi$  and any
      environment $\eta$ it holds: 
      \begin{equation*}
        \sem[]{\PtoR{\varphi}} \set{\frame{\sigma, \Delta}\concat S} =
        \set{\frame{\sigma, \Delta}\concat S \mid \sigma \notin
        \sem[]{\varphi}}
      \end{equation*}
    \end{theorem}

    \begin{proof}
      We proceed by structural induction on $\varphi$. 

      In the proof we use the shorthand $\sigma\trel^n \sigma_n$ in place of
      $\sigma = \sigma_0 \trel \sigma_1 \trel \dots \trel \sigma_n$ leaving implicit the states
      $\sigma_0,...\sigma_{n-1}$. Similarly, we write $\sigma\trel_\exists^n \sigma_n$ 
      for $\exists \sigma_0,\ldots,\sigma_{n-1}.~\sigma\trel^n \sigma_n$.

      \subparagraph*{($\varphi=p$) or ($\varphi = \neg p$)}
      Trivial, by definition.
      \subparagraph*{($\varphi= \varphi_1 \land \varphi_2$ or
      $\varphi= \varphi_1 \lor \varphi_2$)} 
      Just follows using De Morgan's law assuming the
      inductive hypothesis (for $i\in \set{1, 2}$):
      \begin{center}
        $\forall \sigma,\Delta,S.~\sem[]{\PtoR{\varphi_i}}
        \set{\frame{\sigma,
        \Delta}\concat S} = \set{\frame{\sigma, \Delta}\concat S \mid
        \sigma \notin \sem[]{\varphi_i}}$
      \end{center}
      \subparagraph*{($\varphi=\AX\varphi_1$)} 
      \begin{align*}
        \sem[]{\PtoR{\AX\varphi_1}}\set{\frame{\sigma, \Delta}\concat S}
        & =
        \sem[]{\push;\nexte;\PtoR{\varphi_1};\pop}\set{\frame{\sigma,
        \Delta}\concat S} \\
        & =
        \sem[]{\nexte;\PtoR{\varphi_1};\pop}\set{
        \frame{\sigma,\Delta}\concat\,
        \frame{\sigma,\Delta}\concat 
        S} \\
        & = \sem[]{\PtoR{\varphi_1};\pop}\set{
               \frame{\sigma',\Delta}\concat\,
               \frame{\sigma, \Delta}\concat 
               S
          \mid \sigma
        \trel \sigma'}\\
        & \overset{\clubsuit}{=} \sem[]{\pop}\set{
        \frame{\sigma', \Delta}\concat\,
          \frame{\sigma, \Delta}\concat 
          S \mid \sigma \trel
          \sigma', \ \sigma'
        \notin \sem[]{\varphi_1}}\\
        & = \set{\frame{\sigma, \Delta}\concat S \mid \exists \sigma' \notin \sem[]{\varphi_1}.~
          \sigma \trel
        \sigma'} \\
        & = \set{\frame{\sigma, \Delta}\concat S \mid \sigma\notin
        \sem[]{\AX \varphi_1}}
      \end{align*}
      where at $(\clubsuit)$ we used the inductive hypothesis about $\varphi_1$.

      \subparagraph*{($\varphi=\AG\varphi_1$)} 
      \begin{align*}
        \sem[]{\PtoR{\AG\varphi_1}} \set{\frame{\sigma,
        \Delta}\concat S} & =
        \sem[]{\push;\nexte^{*};\PtoR{\varphi_1};\pop}
        \set{\frame{\sigma, \Delta}\concat S}\\
        & = \sem[]{\nexte^{*};\PtoR{\varphi_1};\pop}\set{\frame{\sigma,
        \Delta}\concat\, \frame{\sigma, \Delta}\concat S}\\
        & = \sem[]{\PtoR{\varphi_1};\pop}\left (\sem[]{\nexte^{*}}
          \set{\frame{\sigma, \Delta}\concat\,
            \frame{\sigma, \Delta} \concat
        S} \right )\\
        & = \sem[]{\PtoR{\varphi_1};\pop} \left ( \bigcup_{n \in \N}
          (\sem[]{\nexte})^n\set{\frame{\sigma, \Delta}\concat\,
        \frame{\sigma, \Delta}\concat S} \right )\\
        & = \sem[]{\PtoR{\varphi_1};\pop}\left( \bigcup_{n \in \N}
          \set{\frame{\sigma_n, \Delta} \concat\,
            \frame{\sigma, \Delta}
            \concat S \mid \sigma {\textstyle \trel_\exists^n} \sigma_{n}
            } \right ) \\
        & \overset{\clubsuit}{=} \sem[]{\pop}\left ( \bigcup_{n \in \N}
          \set{\frame{\sigma_n, \Delta} \concat\,
            \frame{\sigma, \Delta}
            \concat S \mid \sigma {\textstyle \trel_\exists^n} \sigma_n,
        \ \sigma_n \notin \sem[]{\varphi_1}}\right ) \\
        & = \set{\frame{\sigma, \Delta}\concat S \mid 
        \exists n,\sigma_n.~\sigma {\textstyle \trel_\exists^n} \sigma_n,
        \ \sigma_n \notin
        \sem[]{\varphi_1}} \\
        & = \set{\frame{\sigma, \Delta}\concat S \mid \sigma \notin
        \sem[]{\AG \varphi_1}}
      \end{align*}
      where at $(\clubsuit)$ we used the inductive hypothesis about $\varphi_1$.
      
      \subparagraph*{($\varphi=\varphi_1\AU\varphi_2$)} 
      \begin{align*}
        & \sem[]{\PtoR{\varphi_1\AU\varphi_2}}
        \set{\frame{\sigma, \Delta}\concat S} \\
        & \qquad =
        \sem[]{\PtoR{\varphi_2};\push;\reset;(\add;\nexte;\PtoR{\varphi_2})^{*};(\PtoR{\varphi_1}
        \oplus \loops);\pop} \set{\frame{\sigma, \Delta}\concat S}\\
        & \qquad \overset{\clubsuit}{=}
        \sem[]{\push;\reset;(\add;\nexte;
          \PtoR{\varphi_2})^{*};(\PtoR{\varphi_1} \oplus
        \loops);\pop} \set{\frame{\sigma, \Delta}\concat S \mid \sigma
        \notin \sem[]{\varphi_2}}\\
        & \qquad =  \sem[]{(\add;
          \nexte;\PtoR{\varphi_2})^{*};(\PtoR{\varphi_1}
        \oplus \loops);\pop} \set{\frame{\sigma, \varnothing} \concat\,
          \frame{\sigma, \Delta}\concat S \mid \sigma \notin
        \sem[]{\varphi_2}}\\
        & \qquad =  \sem[]{(\PtoR{\varphi_1} \oplus \loops);\pop} \left (
          \sem[]{(\add; \nexte;\PtoR{\varphi_2})^{*}}\set{\frame{\sigma,
            \varnothing}\concat\, \frame{\sigma, \Delta}\concat
            S \mid \sigma
        \notin \sem[]{\varphi_2}} \right )\\
        & \qquad =  \sem[]{(\PtoR{\varphi_1} \oplus \loops);\pop} \left (
          \bigcup_{n \in \N}(\sem[]{(\add;\nexte;
            \PtoR{\varphi_2})})^{n}\set{\frame{\sigma,
              \varnothing}\concat\, \frame{\sigma, \Delta}\concat S \mid \sigma
        \notin \sem[]{\varphi_2}}  \right )\\
        & \qquad \overset{\clubsuit}{=} \sem[]{(\PtoR{\varphi_1} \oplus \loops);\pop} \Biggl(
          \bigcup_{n \in \N} \left \{\frame{\sigma_n, \set{\sigma_0,
            \dots, \sigma_{n-1}}}\concat\, \frame{\sigma, \Delta} \concat
            S \mid \sigma {\textstyle \trel^n} \sigma_{n}, \right . \\
            & \hspace{150pt} \left .  \forall i \in [0, n].~ \sigma_i
        \notin \sem[]{\varphi_2} \right \}  \Biggr) \\
        & \qquad \overset{\heartsuit}{=}  \sem[]{\pop} \Biggl( \bigcup_{n \in \N} \left \{
            \frame{\sigma_n, \set{\sigma_0, \dots,
            \sigma_{n-1}}} \concat\,
            \frame{\sigma, \Delta}\concat S \mid \sigma {\textstyle \trel^n} \sigma_n, \right . \\ 
            & \hspace{100pt} \left . \forall i \in [0, n].~ \sigma_i
            \notin \sem[]{\varphi_2}, \ (\sigma_n \notin \sem[]{\varphi_1}
            \ \lor \ \exists i \in [0, n-1] \ \sigma_n = \sigma_i)
        \right \} \Biggr) \\
        & \qquad =  \left \{\frame{\sigma, \Delta}\concat S \mid
        \exists n,\sigma_n.~\sigma {\textstyle \trel_\exists^n} \sigma_n,
          \ \forall i \in [0,
          n]. \ \sigma_i \notin \sem[]{\varphi_2}, 
          \right . \\
          & \hspace{150pt} \left . 
          (\sigma_n \notin \sem[]{\varphi_1}
            \lor \exists i \in [0, n-1] .~
        \sigma_n = \sigma_i) \right \} \\
        & \qquad =  \set{\frame{\sigma, \Delta}\concat S \mid \sigma
        \notin \sem[]{\varphi_1 \AU \varphi_2}}
      \end{align*}
      where at $(\clubsuit)$ we used the inductive hypothesis about $\varphi_2$ and at $(\heartsuit)$ about $\varphi_1$.
      For the last step observe that finiteness of the transition system
      plays an essential role. In fact, in order to falsify
      $\varphi_1\AU\varphi_2$ a state $\sigma$ has to admit a finite
      trace $\sigma=\sigma_0 \trel \cdots \trel \sigma_n$ where $\varphi_2$ is
      never satisfied leading to a state which does not satisfy
      $\varphi_1$, or it can admit an infinite trace where
      $\varphi_2$ is
      never satisfied: since the transition system is finite such an
      infinite trace must eventually be a cycle.

     \subparagraph*{($\varphi=\AF\varphi_1$)}
      Analogous to the previous case, as we have:
      \begin{align*}
        & \sem[]{\PtoR{\AF \varphi_1}}
        \set{\frame{\sigma, \Delta}\concat S} \\
        & \qquad =
        \sem[]{\PtoR{\varphi_1};\push;\reset;(\add;\nexte;\PtoR{\varphi_1})^{*}; \loops;\pop} \set{\frame{\sigma, \Delta}\concat S}\\
        & \qquad \overset{\clubsuit}{=}
        \sem[]{\push;\reset;(\add;\nexte;\PtoR{\varphi_1})^{*};
        \loops;\pop} \set{\frame{\sigma, \Delta}\concat S \mid \sigma
        \notin \sem[]{\varphi_1}}\\
        & \qquad =  \sem[]{(\add;
          \nexte;\PtoR{\varphi_1})^{*}; \loops;\pop} \set{\frame{\sigma, \varnothing} \concat\,
          \frame{\sigma, \Delta}\concat S \mid \sigma \notin
        \sem[]{\varphi_1}}\\
        & \qquad =  \sem[]{\loops;\pop} \left (
          \sem[]{(\add; \nexte;\PtoR{\varphi_1})^{*}}\set{\frame{\sigma,
            \varnothing}\concat\, \frame{\sigma, \Delta}\concat
            S \mid \sigma
        \notin \sem[]{\varphi_1}} \right )\\
        & \qquad =  \sem[]{\loops;\pop} \left (
          \bigcup_{n \in \N} (\sem[]{(\add; \nexte;\PtoR{\varphi_1})})^{n}\set{\frame{\sigma,
              \varnothing}\concat\, \frame{\sigma, \Delta}\concat S \mid \sigma
        \notin \sem[]{\varphi_1}}  \right )\\
        & \qquad \overset{\clubsuit}{=} \sem[]{\loops;\pop} \Biggl(
          \bigcup_{n \in \N} \left \{\frame{\sigma_n, \set{\sigma_0,
            \dots, \sigma_{n-1}}}\concat\, \frame{\sigma, \Delta} \concat
            S \mid 
            \sigma\trel^n \sigma_n,
            \right . \\
            & \hspace{150pt} \left .  \forall i \in [0, n].~ \sigma_i
        \notin \sem[]{\varphi_1} \right \}  \Biggr) \\
        & \qquad =  \sem[]{\pop} \Biggl( \bigcup_{n \in \N} \left \{
            \frame{\sigma_n, \set{\sigma_0, \dots,
            \sigma_{n-1}}} \concat\,
            \frame{\sigma, \Delta}\concat S \mid 
            \sigma\trel^n \sigma_n,
            \right . \\
            & \hspace{150pt} \left . \forall i \in [0, n].~ \sigma_i
            \notin \sem[]{\varphi_1}, \ \exists i \in [0, n-1].~ \sigma_n = \sigma_i
        \right \} \Biggr) \\
        & \qquad =  \left \{\frame{\sigma, \Delta}\concat S \mid 
        \exists n,\sigma_n.~\sigma {\textstyle \trel_\exists^n} \sigma_n,
          \ \forall i \in [0,
          n].~ \sigma_i \notin \sem[]{\varphi_1}, 
          \exists i \in [0, n-1].~
        \sigma_n = \sigma_i \right \} \\
        & \qquad =  \set{\frame{\sigma, \Delta}\concat S \mid \sigma
          \notin \sem[]{\AF \varphi_1}}
      \end{align*}
      where at $(\clubsuit)$ we used the inductive hypothesis about $\varphi_1$.  As in
      the previous case we use the fact that infinite traces are
      identified as (finite) loops.
    \end{proof}
  \end{toappendix}

  \subparagraph*{$\mu_\Box$-calculus as $\MOKA$}

  To each $\mu_\Box$-formula $\varphi$ 
  we assign the $\MOKA$ program $\PtoR{\varphi}$
  as follows:
  \begin{center}
    \begin{tabular}{r@{\hspace{2pt}}c@{\hspace{2pt}}lr@{\hspace{2pt}}c@{\hspace{2pt}}l}
      $\PtoR{\,p}$ & $\triangleq$ & $\np$
      & $\PtoR{\,\neg p}$ & $\triangleq$ & $\p$
      \\
      $\PtoR{\varphi_1\land \varphi_2}$ & $\triangleq$ &
      $\PtoR{\varphi_1} \oplus \PtoR{\varphi_2}$ 
      &
      $\PtoR{\varphi_1\lor \varphi_2}$ & $\triangleq$ &
      $\PtoR{\varphi_1} ; \PtoR{\varphi_2}$ \\
      $\PtoR{\Box \varphi}$ & $\triangleq$ &
      $\push; \nexte;\PtoR{\varphi}; \pop$ & %
      $\PtoR{x}$ & $\triangleq$ & $\X$  \\
      $\PtoR{\mu x.\varphi_x}$ & $\triangleq$ & $\push; \reset; \upmu
      \X. \Bigl ( \loops \oplus (\add; \PtoR{\varphi_x}) \Bigr ); \pop$  \qquad\qquad %
      & $\PtoR{\nu x.\varphi_x}$ & $\triangleq$ & $\upmu
      \X. 
      \PtoR{\varphi_x} $ \\
    \end{tabular}
  \end{center}

  The second component of a frame $\frame{\sigma, \Delta}$ is still
  used to identify looping computations. This intervenes in the encoding
  of least fixpoints $\mu x. \varphi_x$: when checking the
  formula, the current state first logged to the current frame
  ($\add; \PtoR{\varphi_x}$ branch); a counterexample is found when we
  try to verify the fixpoint property in a state where the check has
  already been tried (filtered by $\loops$).
  The encoding of greatest fixpoints $\nu x. \varphi_x$ instead is
  simpler: searching for counterexamples 
  naturally translates to a least fixpoint, which is offered
  natively by $\MOKA$.

  The correctness of this transform (see Theorem~\ref{th:programs-mu} and Corollary~\ref{co:programs-mu} in the appendix) leverages
  a small variation of the tableau construction
  in~\cite{DBLP:conf/tapsoft/StirlingW89},
  with judgements roughly of the form $\sigma, \Delta \vdash \varphi$,
   meaning that the formula $\varphi$ holds in a state $\sigma$
  assuming that all the states in $\Delta$ 
 have been visited
  while checking the
  current fixpoint subformula. Then, we show that given a formula
  $\varphi$ and a stack $\frame{\sigma, \Delta}\concat S$,
  if there is no
  successful tableau for $\sigma, \Delta \vdash \varphi$ then
  $\sem{\PtoR{\varphi}} \set{\frame{\sigma, \Delta}\concat S} =
  \set{\frame{\sigma, \Delta}\concat S}$ holds,  while
  $\sem{\PtoR{\varphi}} \set{\frame{\sigma, \Delta}\concat
  S} = \varnothing$
  otherwise.

  \begin{toappendix}
    \subparagraph{$\mu_\Box$-calculus as $\MOKA$.}
    In order to prove an analogous result in the case of the
    $\mu_\Box$-calculus we first introduce an intermediate encoding in
    an extension of $\MOKA$ including an additional basic expression
    $\nloops$, the negation of $\loops$, which checks if the current
    state does not loop back to one of the states in the current
    trace. Hence its semantics is
    \begin{center}
      $\bsem{\nloops}\set{\frame{\sigma, \Delta} \concat S}$ $\triangleq$
      $\set{\frame{\sigma, \Delta}\concat S \mid \sigma
        \notin \Delta }$
    \end{center}
  
    Then, in the encoding of fixpoint formulae we continue the check
    only for states which have not been already considered, i.e.,
    where $\nloops$ succeeds:
    
    \begin{center}
    \begin{tabular}{r@{\hspace{2pt}}c@{\hspace{2pt}}l}
      $\PtoRn{\,p}$ & $\triangleq$ & $\np$ \qquad $\PtoRn{\,\neg p}$ $\triangleq$ $\p$
      \qquad
      $\PtoRn{\varphi_1\land \varphi_2}$ $\triangleq$
      $\PtoRn{\varphi_1} \oplus \PtoRn{\varphi_2}$
      \qquad
      $\PtoRn{\varphi_1\lor \varphi_2}$ $\triangleq$
      $\PtoRn{\varphi_1} ; \PtoRn{\varphi_2}$
       \\
      $\PtoRn{x}$ & $\triangleq$ & $\X$ \hspace{110.3pt} $\PtoRn{\Box \varphi}$ $\triangleq$ 
      $\push; \nexte;\PtoRn{\varphi}; \pop$ \\
      $\PtoRn{\mu x.\varphi_x}$ & $\triangleq$ & $\push; \reset; \upmu
      \X. \Bigl ( (\nloops;
      \add; \PtoRn{\varphi_x}) \oplus \loops \Bigr ); \pop$ \\
      $\PtoRn{\nu x.\varphi_x}$ & $\triangleq$ & $\push; \reset; \upmu
      \X. \Bigl ( \nloops; \add;
      \PtoRn{\varphi_x} \Bigr ); \pop$ \\
    \end{tabular}
  \end{center}

  \begin{itemize}
    \item for $\mu x. \varphi_x$ a counterexample is found when we try
      to verify the fixpoint property in a state where the check has
      been already tried (filtered by $\loops$). Instead, when
      checking on new states (filtered by $\nloops$), these
      are first logged by $\add$ to the current frame and
      then the check continues.

    \item for $\nu x. \varphi_x$ the situation is dual: a
      second check on
      a state already visited witnesses that the formula is satisfied,
      hence the corresponding state must be discharged as it cannot
      provide a counterexample. Instead, when checking on new states
      (filtered by $\nloops$ ) these are logged by $\add$
      to the current frame
      and the check continues.
  \end{itemize}

    \begin{theorem}[Formula satisfaction as program verification
      - $\mu$-calculus version]
      \label{th:programs-mu}
      Given 
      $T = (\Sigma, \operatorname{I}, \mathbf{P}, \trel, \vdash)$, for all frames $\frame{\sigma, \Delta}\concat S \in
      \str{\operatorname{F}_{\Sigma}}$,
      for any $\mu_\Box$-calculus formula $\varphi$,
      environment $\eta$ and valuation $\mathcal{V}$ it holds:
      \begin{equation*}
        \sem{\PtoRn{\varphi}} \set{\frame{\sigma, \Delta}\concat S} =
        \set{\frame{\sigma, \Delta}\concat S \mid \sigma \notin
        \sem[]{\varphi}}
      \end{equation*}
    \end{theorem}

    \begin{proof}
      We assume, w.l.o.g., that some formula $\varphi$ is fixed
      in which all variables are distinct in a way that we can
      map each variable $y$ uniquely to the corresponding
      formula $\varphi_y$, where $\mu y. \varphi_y$ is the
      subformula in $\varphi$, and for any subformula of
      $\varphi$ we can identify the variable $\bvar{\varphi}$ quantified by the
      enclosing fixpoint.

      We consider the tableau construction
      of~\cite{DBLP:conf/tapsoft/StirlingW89}, with small
      modifications to adapt it to our context. In particular
      our judgements are of the form
      $\sigma, \Gamma \vdash \varphi$, where $\sigma \in \Sigma$
      and $\Gamma \subseteq \LVar \times \Sigma$.

      The tableau rules are the following:
      \begin{center}
        \begin{tabular}{c c}
          $\inference{\sigma, \Gamma \vdash \varphi_1 \lor \varphi_2}
          {\sigma, \Gamma \vdash \varphi_1}[($\lor_1$)]$
          & $\inference{\sigma, \Gamma \vdash \varphi_1 \lor \varphi_2}
          {\sigma, \Gamma \vdash \varphi_2}[($\lor_2$)]$  \\
          & \\
          $\inference{\sigma, \Gamma \vdash \varphi_1 \land \varphi_2}
          {\sigma, \Gamma \vdash \varphi_1 \qquad \sigma, \Gamma
          \vdash \varphi_2}[($\land$)]$
          & $\inference{\sigma, \Gamma \vdash x}
          {\sigma, \Gamma \cup \set{(x, \sigma)} \vdash
          \varphi_x}[($\X$)]$  \\
          &  \\
          $\inference{\sigma, \Gamma \vdash \mu x. \varphi_x}
          {\sigma, \Gamma \vdash x}[($\mu. x$)]$
          & $\inference{\sigma, \Gamma \vdash \nu x. \varphi_x}
          {\sigma, \Gamma \vdash x}[($\nu. x$)]$ \\
          & \\
          $\inference{\sigma, \Gamma  \vdash \Box \varphi}
          {\sigma_1, \Gamma \vdash \varphi \qquad \dots \qquad
          \sigma_n, \Gamma \vdash \varphi}[($\Box$)]$
          & $\sigma \trel \sigma_i, \ i = 1, \dots, n$ \\
        \end{tabular}
      \end{center}

      A leaf is \emph{successful} if it is of the form
      $\sigma, \Gamma
      \vdash x$ with
      $(x, \sigma) \in \Gamma$ and $x$ associated with a
      greatest fixpoint (namely bound in a subformula $\nu x. \varphi_x$), or
      $\sigma, \Gamma \vdash p$ with $\sigma \vdash p$.
      A leaf is a \emph{failure} if it is of the form
      $\sigma, \Gamma \vdash x$ with
      $(x, \sigma) \in \Gamma$ and $x$ such that is associated
      to a least fixpoint (namely bound in a subformula
      $\mu x. \varphi_x$), or
      $\sigma, \Gamma \centernot \vdash p$ with
      $\sigma \centernot \vdash p$.
      A tableau is successful if all its leaves are successful.

      Let us give a name to the programs corresponding to
      subformulae $\mu x. \varphi_x$ or $\nu x. \varphi_x$ of
      $\varphi$.  Recall that
      \begin{align*}
        \PtoRn{\mu x.\varphi_x} & = \push; \reset;\underset{=
        \r_x}{\underbrace{\upmu \X.\Bigl ((\nloops;
        \add; \PtoRn{\varphi_x}) \oplus \loops  \Bigr )}}; \pop \\
        \PtoRn{\nu x.\varphi_x} & = \push; \reset;  \underset{=
        \r_x}{\underbrace{\upmu \X. \Bigl ( \nloops; \add;
        \PtoRn{\varphi_x} \Bigr ) }}; \pop
      \end{align*}
      As indicated above we name $\r_x$ the underbraced programs
      and call $\X$ the $\MOKA$ variable corresponding to the
      logical variable $x$. Then we define an environment $\eta$
      for $\MOKA$ programs by letting
      $\eta(\X) = \sem[]{\r_{x}}$. Note that the semantics of
      $\r_{x}$ is independent from the environment since
      $\r_{x}$ is closed.

      \medskip

      Now we prove that for all subformulae $\psi$ of $\varphi$,
      for all stacks $\frame{\sigma, \Delta}\concat S$, if $y=\bvar{\psi}$ is the
      variable bound in the enclosing fixpoint and $\Gamma \subseteq
      \left (\Var \setminus \set{y} \right) \times \Sigma$, then we have
      \begin{equation*}
        \sem{\PtoRn{\psi}} \set{\frame{\sigma, \Delta}\concat S} =
        \begin{cases}
          \set{\frame{\sigma, \Delta}\concat S} &
          \mbox{if there is no successful tableau }
          \sigma, \Gamma \cup \left ( \set{y} \times \Delta \right
          )\vdash \psi \\
          \varnothing & \mbox{otherwise}
        \end{cases}
      \end{equation*}

      Now the proof goes by induction on
      $\frame{\sigma, \Delta}\concat S$ and $\psi$ ordered
      by
      \begin{equation*} \frame{\sigma, \Delta}\concat S, \psi
        < \frame{\sigma', \Delta'}\concat S', \psi'
      \end{equation*}
      if one of the following holds
        
      \begin{itemize}
        \item $\psi = \varphi_z$ and $\psi' = \eta z.
          \varphi_z$ for
          some $\eta \in \set{\mu, \nu}$ and $z$ variable;
        \item The previous condition is not satisfied, $\bvar{\psi}=\bvar{\psi'}$ and either
          \begin{itemize}
            \item $\Delta' \subset \Delta$
            \item $\Delta=\Delta'$ and $\psi$ is a subformula of
              $\psi'$;
          \end{itemize}
      \end{itemize}

      \subparagraph*{($\psi = p$)}
      We have
      \begin{equation*}
        \sem{\PtoRn{p}} \set{  \frame{\sigma, \Delta}\concat S} =
        \set{  \frame{\sigma, \Delta}\concat S \mid \sigma
        \centernot\vdash p}
      \end{equation*}
      which is the empty set if and only if $\sigma \vdash p$, i.e.,
      $\sigma,  \Gamma \cup  \left ( \set{y} \times \Delta \right )
      \vdash p$ alone is a
      successful tableau.

      \subparagraph*{($\psi = \neg p$)} Analogous to the case $\psi = p$.

      \subparagraph*{($\psi = \psi_1 \lor \psi_2$)}
      \begin{align*}
        \sem{\PtoRn{\psi_1 \lor \psi_2}} \set{  \frame{\sigma,
        \Delta}\concat S} & = \sem{\PtoRn{\psi_1}};
        \sem{\PtoRn{\psi_2}}
        \set{  \frame{\sigma, \Delta}\concat S} \\
        & = \sem{\PtoRn{\psi_2}} \set{  \frame{\sigma,
        \Delta}\concat S} \\
        & = \set{  \frame{\sigma, \Delta}\concat S}
      \end{align*}
      if and only if, by inductive hypothesis, since both $\psi_1,
      \psi_2$ are subformulae of $\psi_1 \lor \psi_2$,
      there is no
      successful tableau neither for $\sigma, \Gamma \cup \left (\set{y}
      \times \Delta \right ) \vdash \psi_1$
      nor for $\sigma, \Gamma \cup  \left ( \set{y} \times \Delta
      \right ) \vdash
      \psi_2$, since a tableau for
      $\sigma, \Gamma \cup  \left ( \set{y} \times \Delta \right )
      \vdash \psi_1 \lor
      \psi_2$ is of the shape
      \begin{equation*}
        \inference{\sigma, \Gamma \cup  \left ( \set{y} \times \Delta
          \right ) \vdash \psi_1 \lor
        \psi_2}{\inference{\sigma,  \Gamma \cup  \left ( \set{y}
            \times \Delta \right ) \vdash
        \psi_1}{T}} \qquad
        \qquad  \inference{\sigma, \Gamma \cup  \left ( \set{y}
          \times \Delta \right )
          \vdash \psi_1 \lor
        \psi_2}{\inference{\sigma, \Gamma \cup  \left ( \set{y}
            \times \Delta \right )
        \vdash \psi_2}{T'}}
      \end{equation*}
      with $T, T'$ successful tableaux.

      \subparagraph*{($\psi = \psi_1 \land \psi_2$)}
      \begin{align*}
        \sem{\PtoRn{\psi_1 \land \psi_2}} \set{  \frame{\sigma,
        \Delta}\concat S}  & = \sem{\PtoRn{\psi_1} \oplus
        \PtoRn{\psi_2}}
        \set{  \frame{\sigma, \Delta}\concat S} \\
        & = \sem{\PtoRn{\psi_1}}\set{  \frame{\sigma,
        \Delta}\concat S}
        \cup \sem{\PtoRn{\psi_2}} \set{  \frame{\sigma,
        \Delta}\concat S} \\
        & = \set{  \frame{\sigma, \Delta}\concat S}
      \end{align*}
      if and only if, by inductive hypothesis, since both $\psi_1,
      \psi_2$ are subformulae of $\psi_1 \land
      \psi_2$, either there
      is no successful tableau for $\sigma, \Gamma \cup \left (\set{y} \times
      \Delta \right ) \vdash \psi_1$ or
      for $\sigma, \Gamma \cup  \left ( \set{y} \times \Delta \right
      ) \vdash \psi_2$,
      since a tableau for
      $\sigma, \Gamma \cup  \left ( \set{y} \times \Delta \right )
      \vdash \psi_1 \land
      \psi_2$ is of the shape
      \begin{equation*}
        \inference{\sigma, \Gamma \cup  \left ( \set{y} \times \Delta
          \right ) \vdash \psi_1 \land
        \psi_2}{\inference{\sigma, \Gamma \cup  \left ( \set{y}
            \times \Delta \right )
          \vdash \psi_1}{T} \qquad
          \inference{\sigma, \Gamma \cup  \left ( \set{y} \times \Delta
        \right ) \vdash \psi_2}{T'}}
      \end{equation*}
      with $T, T'$ successful tableaux.
      \subparagraph*{($\psi = \Box \psi_1$)}
      \begin{align*}
        \sem{\PtoRn{\Box{\psi_1}}}\set{  \frame{\sigma, \Delta}\concat
        S} & = \sem{\push{}; \nexte; \PtoRn{\psi_1}; \pop} \set{
        \frame{\sigma, \Delta}\concat S} \\
        & = \sem{\pop}\Bigl ( \sem{\PtoRn{\psi_1}}\set{
            \frame{\sigma', \Delta}\concat\, \frame{\sigma,
            \Delta}\concat S
        \mid \sigma \trel \sigma'}\Bigr ) \\
        & \overset{\clubsuit}{=} \sem{\pop}\Bigl ( \left \{
            \frame{\sigma', \Delta} \concat\, \frame{\sigma, \Delta} \concat S
            \mid \sigma \trel \sigma', \right . \\ 
        & \left . \hspace{50pt} \mbox{no successful tableau
            for } \sigma', \Gamma \cup  ( \{y\} \times \Delta) 
            \vdash \psi_1 \right \} \Bigr ) \\
        & =
        \begin{cases}
          \set{\frame{\sigma, \Delta}\concat S} & \mbox{if no successful tableau for } 
          \sigma, \Gamma \cup  \left (
          \set{y} \times \Delta \right )
          \vdash \Box\psi_1 \\
          \varnothing & \mbox{otherwise}
        \end{cases}
      \end{align*}
      where at $(\clubsuit)$ we applied the inductive hypothesis
      since $\psi_1$ is a subformula of $\Box \psi_1$.
      \subparagraph*{($\psi = y$ variable associated with $\mu  y. \varphi_y$)}
      We have
      \begin{equation*}
        \sem{\PtoRn{y}} = \sem{\Y} = \eta(\Y) = \sem{\r_y}
      \end{equation*}
      where
      \begin{equation*}
        \r_{y} = \upmu \Y. \Bigl ((\nloops; \add;
        \PtoRn{\varphi_y}) \oplus \loops \Bigr )
      \end{equation*}
      We have two possible cases:
      \begin{itemize}
        \item if $\sigma \in \Delta$ then
          \begin{equation*}
            \sem{\Y} \set{\frame{\sigma, \Delta}\concat S} =
            \sem{\r_y} \set{\frame{\sigma, \Delta}\concat S} =
            \set{\frame{\sigma, \Delta}\concat S}
          \end{equation*}
          and indeed $\sigma, \Gamma \cup  \left ( \set{y} \times
          \Delta \right ) \vdash
          y$ alone is a failed tableau
          since $y$ corresponds to a minimal fixpoint formula $\mu
          y. \varphi_y$.
        \item if $\sigma \notin \Delta$ then
          \begin{align*}
            \sem{\Y} \set{\frame{\sigma, \Delta}\concat S}
            & = \sem{\r_y} \set{\frame{\sigma, \Delta}\concat S} \\
            & = \sem{(\nloops; \add;
            \PtoRn{\varphi_y}) \oplus \loops}\set{\frame{\sigma, \Delta}\concat S}\\
            & = \sem{\PtoRn{\varphi_y}}\set{\frame{\sigma,
            \Delta\cup \set{\sigma}}\concat S} \\
            & \overset{\clubsuit}{=} \set{\frame{\sigma,
            \Delta\cup \set{\sigma}}\concat S}
          \end{align*}
      \end{itemize}
      if and only if there is no successful tableau $\sigma, \Gamma
      \cup \left (\set{y} \times (\Delta
      \cup \set{\sigma}) \right ) \vdash \varphi_y$, if and only if
      there is no successful tableau $\sigma, \Gamma \cup \left ( \set{y}
      \times \Delta \right ) \vdash y$,
      since the only applicable rule is
      \begin{equation*}
        \inference{\sigma, \Gamma \cup  \left ( \set{y} \times \Delta \right )
        \vdash y}{\sigma, \Gamma \cup \left ( \set{y} \times (\Delta \cup
        \set{\sigma}) \right ) \vdash \varphi_y}
      \end{equation*}
      At $(\clubsuit)$ we can apply the inductive hypothesis since
      $\Delta \subseteq \Delta \cup \set{\sigma}$.
      \subparagraph*{($\psi = \mu z. \varphi_z$)}
      \begin{align*}
        \sem{\PtoRn{\mu z. \varphi_z}}\set{\frame{\sigma,
        \Delta}\concat S} & = \sem{\push{}; \upmu \Z. \Bigl ( \loops
          \oplus (\nloops; \add; \PtoRn{\varphi_z}) \Bigr ); \pop
        }\set{\underset{= S'}{\underbrace{\frame{\sigma,
        \Delta}\concat S}}} \\
        & = \sem{\pop}\Bigl [ \sem{ \upmu \Z. \Bigl ( \loops \oplus
          (\nloops; \add; \PtoRn{\varphi_z}) \Bigr )}
        \set{\frame{\sigma, \varnothing}\concat S'} \Bigr ] \\
        & \overset{\diamond}{=} \sem{\pop} \Bigl ( \sem{\loops
            \oplus (\nloops; \add;
          \PtoRn{\varphi_z})}\set{\frame{\sigma, \varnothing}\concat
        S'} \Bigr ) \\
        & = \sem{\pop}\Bigl (
          \underset{(\heartsuit)}{\underbrace{\sem{
        \PtoRn{\varphi_z}}\set{\frame{\sigma, \set{\sigma}}\concat S'}}} \Bigr )
      \end{align*}
      where at $(\diamond)$ we use the fact that $\sem{\upmu \Z. \r} =
      \llbracket \r \rrbracket_{\eta[\Z \mapsto \sem{\upmu \Z. \r}]}$.
      Now the first clause of
      the order holds, since we have $\psi = \varphi_z$ and
      $\psi' = \mu z. \varphi_z$, hence by inductive hypothesis we
      deduce that $(\heartsuit)$ is
      \begin{itemize}
        \item $\set{\frame{\sigma, \set{\sigma}}\concat S'}$ if and
          only if there is no successful tableau $\sigma, \Gamma \cup
          \left ( \set{y} \times \Delta \right ) \cup
          \set{(z,\sigma)} \vdash \varphi_z$,
          if and only if there is no
          successful tableau $\sigma, \Gamma \cup \left
          (\set{y}\times\Delta \right )
          \vdash \mu z.
          \varphi_z$ since the only applicable rule is
          \begin{equation*}
            \inference{\sigma, \Gamma \cup  \left ( \set{y} \times
              \Delta \right ) \vdash \mu z.
            \varphi_z}{\inference{\sigma, \Gamma \cup  \left (
                \set{y} \times \Delta \right ) \vdash
              z}{\sigma, \Gamma \cup  \left ( \set{y} \times \Delta
                \right ) \cup
            \set{(z,\sigma)} \vdash \varphi_z}}
          \end{equation*}
        \item $\varnothing$ if and only if there is a successful
          tableau for $\sigma, \Gamma \cup  \left ( \set{y} \times
          \Delta \right ) \cup
          \set{(z, \sigma)} \vdash \varphi_z$,
          hence for $\sigma, \Gamma \cup  \left ( \set{y} \times
          \Delta \right ) \vdash \mu
          z. \varphi_z$.
      \end{itemize}
      \subparagraph*{($\psi = y$ variable associated with $\nu y.
      \varphi_y$)}   We have
      \begin{equation*}
        \sem{\PtoRn{y}} = \sem{\Y} = \eta(\Y) = \sem{\r_y}
      \end{equation*}
      where
      \begin{equation*}
        \r_{y} = \upmu \Y. \Bigl (\nloops; \add; \PtoRn{\varphi_y} \Bigr )
      \end{equation*}
      We have two possible cases:
      \begin{itemize}
        \item if $\sigma \in \Delta$ then
          \begin{equation*}
            \sem{\Y} \set{\frame{\sigma, \Delta}\concat S} =
            \sem{\r_y} \set{\frame{\sigma, \Delta}\concat S} = \varnothing
          \end{equation*}
          and indeed $\sigma, \Gamma \cup  \left ( \set{y} \times
          \Delta \right ) \vdash y$
          alone is a successful
          tableau since $y$ corresponds to a maximal fixpoint formula
          $\nu y. \varphi_y$.
        \item if $\sigma \notin \Delta$ then
          \begin{align*}
            \sem{\Y} \set{\frame{\sigma, \Delta}\concat S} & =
            \sem{\r_y} \set{\frame{\sigma, \Delta}\concat S} \\
            & = \sem{\nloops; \add;
            \PtoRn{\varphi_y}}\set{\frame{\sigma, \Delta}\concat S}\\
            & = \sem{\PtoRn{\varphi_y}}\set{\frame{\sigma,
            \Delta\cup \set{\sigma}}\concat S} \\
            & \overset{\clubsuit}{=} \set{\frame{\sigma,
            \Delta\cup \set{\sigma}}\concat S}
          \end{align*}
      \end{itemize}
      if and only if there is no successful tableau $\sigma, \Gamma
      \cup  \left ( \set{y} \times \Delta \right ) \cup \set{(y, \sigma)} \vdash
      \varphi_y$, if and only if there
      is no successful tableau $\sigma, \Gamma \cup  \left ( \set{y}
      \times \Delta \right )
      \vdash y$, since the
      only applicable rule is
      \begin{equation*}
        \inference{\sigma, \Gamma \cup  \left ( \set{y} \times \Delta
          \right ) \vdash
        y}{\sigma, \Gamma \cup  \left ( \set{y} \times \Delta \right ) \cup
        \set{(y, \sigma)} \vdash \varphi_y}
      \end{equation*}
      At $(\clubsuit)$ we can apply the inductive hypothesis since
      $\Delta \subseteq \Delta \cup \set{\sigma}$.
      \subparagraph*{($\psi = \nu z. \varphi_z$)}
      \begin{align*}
        \sem{\PtoRn{\nu z. \varphi_z}}\set{\frame{\sigma,
        \Delta}\concat S} & = \sem{\push{}; \upmu \Z. \Bigl ( \nloops; \add;
        \PtoRn{\varphi_z} \Bigr ); \pop }
        \set{\frame{\sigma, \Delta}\concat S} \\
        & = \sem{\pop}\Bigl ( \sem{ \upmu \Z. (\nloops; \add;
          \PtoRn{\varphi_z})} \set{\frame{\sigma, \varnothing}\concat\,
        \frame{\sigma, \Delta}\concat S} \Bigr ) \\
        & \overset{\diamond}{=} \sem{\pop} \Bigl ( \sem{\nloops;
          \add; \PtoRn{\varphi_z}}\set{\frame{\sigma, \varnothing}\concat\,
        \frame{\sigma, \Delta}\concat S} \Bigr ) \\
        & = \sem{\pop}\Bigl ( \underset{(\heartsuit)}{\underbrace{\sem{
              \PtoR{\varphi_z}}\set{\frame{\sigma,
        \set{\sigma}}\concat\, \frame{\sigma, \Delta}\concat S}}} \Bigr )
      \end{align*}
      where at $(\diamond)$ we use the fact that $\sem{\upmu \Z. \r} =
      \llbracket \r \rrbracket_{\eta[\Z \mapsto \sem{\upmu \Z. \r}]}$.
      Now the first clause of
      the order holds, since we have $\psi = \varphi_z$ and
      $\psi' = \mu z. \varphi_z$, by inductive hypothesis we deduce
      that $(\heartsuit)$ is
      \begin{itemize}
        \item $\set{\frame{\sigma, \set{\sigma}}\concat S'}$ if and only
          if there is no successful tableau $\sigma, \Gamma
      \cup  \left ( \set{y} \times \Delta \right ) \cup \set{(z, \sigma)}
          \vdash \varphi_z$, if and only if there is no successful
          tableau $\sigma, \Gamma \cup  \left ( \set{y} \times \Delta
          \right ) \vdash \nu z.
          \varphi_z$ since
          the only applicable rule is
          \begin{equation*}
            \inference{\sigma, \Gamma \cup  \left ( \set{y} \times
              \Delta \right ) \vdash \nu z.
            \varphi_z}{\inference{\sigma, \Gamma \cup  \left (
                \set{y} \times \Delta \right ) \vdash
              z}{\sigma, \Gamma \cup  \left ( \set{y} \times \Delta
                \right ) \cup
            \set{(z, \sigma)} \vdash \varphi_z}}
          \end{equation*}
        \item $\varnothing$ if and only if there is a successful
          tableau for $\sigma, \Gamma \cup  \left ( \set{y} \times
          \Delta \right ) \cup
          \set{(z, \sigma)} \vdash \varphi_z$,
          hence for $\sigma, \Gamma \cup  \left ( \set{y} \times
          \Delta \right ) \vdash \nu
          z. \varphi_z$.
      \end{itemize}
      \hfill
    \end{proof}

    In order to conclude, we need to show that the result above still
    works for the original encoding $\PtoR{\varphi}$. This can be
    obtained as an easy corollary.

    \begin{corollary}
      \label{co:programs-mu}
      For all frames
      $\frame{\sigma, \Delta}\concat S \in
      \str{\operatorname{F}_{\Sigma}}$, for any $\mu_\Box$-calculus
      formula $\varphi$, environment $\eta$ and valuation
      $\mathcal{V}$ it holds
      $\sem{\PtoR{\varphi}} \set{\frame{\sigma, \Delta}\concat S} =
      \sem{\PtoRn{\varphi}} \set{\frame{\sigma, \Delta}\concat S}$ and
      thus
      \begin{equation*}
        \sem{\PtoR{\varphi}} \set{\frame{\sigma, \Delta}\concat S} =
        \set{\frame{\sigma, \Delta}\concat S \mid \sigma \notin
          \sem[]{\varphi}}
      \end{equation*}
    \end{corollary}

    \begin{proof}
      The non-trivial cases are those concerning fixpoints formulas.
      \begin{itemize}
      \item For least fixpoints $\mu x.\, \varphi_x$, recall that
        \begin{quote}
          $\PtoRn{\mu x.\varphi_x}$ $\triangleq$ $\push; \reset; \mu
          \X. \Bigl ( (\nloops;
          \add; \PtoRn{\varphi_x}) \oplus \loops \Bigr ); \pop$ \\
          $\PtoR{\mu x.\varphi_x}$ $\triangleq$ $\push; \reset; \mu
          \X. \Bigl ( (\add; \PtoR{\varphi_x}) \oplus \loops \Bigr ); \pop$ \\
        \end{quote}

        Thus
        $\sem{\PtoRn{\mu x.\varphi_x}} = \bsem{\pop} \circ
        \operatorname{lfp}(F_n) \circ \bsem{\reset} \circ
        \bsem{\push}$ and
        $\sem{\PtoR{\mu x.\varphi_x}} = \bsem{\pop} \circ
        \operatorname{lfp}(F) \circ \bsem{\reset} \circ \bsem{\push}$
        where $F_n, F : C^C \to C^C$, with
        $C = \pow{\str{\operatorname{F}_{\Sigma}}}$, are the operators
        below:
        \[
          F_n(f) = \sem[{\eta[\X \mapsto f]}]{(\nloops; \add; \PtoR{\varphi_x}) \oplus \loops}\\
          F(f) = \sem[{\eta[\X \mapsto f]}]{(\add; \PtoR{\varphi_x}) \oplus \loops}
        \]
        We conclude simply by observing that, in our setting we can
        restrict to functions $f \in C^C$ which are reductive, i.e.,
        such that $f(X) \subseteq X$  
        for $X \in C$, and on such
        functions $F_n(f) =F(f)$.
        In fact, we have that
        \begin{itemize}
        \item if $\sigma \not\in \Delta$
          then  $F_n(f)(\frame{\sigma, \Delta} \concat S) = \sem[{\eta[\X \mapsto f]}]{\PtoR{\varphi_x}}(\frame{\sigma, \Delta\cup\set{\sigma}} \concat S) = F(f)(\frame{\sigma, \Delta} \concat S)$
    \item if $\sigma \in \Delta$
      then
      \begin{align*}
        F_n(f)(\frame{\sigma, \Delta} \concat S)
        & = \set{\frame{\sigma, \Delta} \concat S}\\
        & \overset{\clubsuit}{=}  \set{\frame{\sigma, \Delta} \concat S} \cup \sem[{\eta[\X \mapsto f]}]{\PtoR{\varphi_x}}(\frame{\sigma, \Delta} \concat S) \\
        & \overset{\heartsuit}{=} \set{\frame{\sigma, \Delta} \concat S} \cup \sem[{\eta[\X \mapsto f]}]{\PtoR{\varphi_x}}(\frame{\sigma, \Delta\cup\set{\sigma}} \concat S) \\
        & = F(f)(\frame{\sigma, \Delta} \concat S)
      \end{align*}
      where the second passage $(\clubsuit)$ is motivated by the fact that 
      $\sem[{\eta[\X \mapsto f]}]{\PtoR{\varphi}}(\frame{\sigma,
        \Delta} \concat S) \subseteq \set{\frame{\sigma, \Delta}
        \concat S}$ and the third $(\heartsuit)$ by the hypothesis $\sigma\in\Delta$.
    \end{itemize}
    
  \item For greatest fixpoints $\nu x.\, \varphi_x$, instead, we
    provide a direct proof of the fact that
    $\PtoR{\nu x. \varphi} \triangleq \upmu \X.  \PtoR{\varphi_x}$ is
    correctly computing the counterexamples, i.e., we show directly
    that
    \begin{center}
      $\sem{\PtoR{\nu x. \varphi}} \set{\frame{\sigma, \Delta}\concat S} =
      \set{\frame{\sigma, \Delta}\concat S \mid \sigma \notin
        \sem[]{\varphi}}$
    \end{center}      
    
    To this aim consider the approximants for
    $\nu x.\, \varphi$, i.e., $\nu^{0} x. \varphi = \true$ and
    $\nu^{n+1} x. \varphi = \varphi [(\nu^{n} x. \varphi)/x]$.
    Similarly, consider the approximants for the encoding
    $\upmu \X. \PtoR{\varphi}$, i.e., $\upmu^0 \X. \PtoR{\varphi} = \zero$
    and
    $\upmu^{n+1} \X. \PtoR{\varphi} = \PtoR{\varphi} [(\upmu^{n}
    \X. \PtoR{\varphi})/\X]$.

    Then one can show that for all $n$ it holds that
    $\PtoR{\nu^{n} x. \varphi} = \upmu^{n} \X. \PtoR{\varphi}$ (this
    immediately follows from the fact that
    $\PtoR{\psi_1[\psi_2/x]} = \PtoR{\psi_1}[\PtoR{\psi_2}/\X]$) and thus
    we can conclude. In fact
    \begin{align*}
      \sem{\PtoR{\nu x. \varphi}} \set{\frame{\sigma, \Delta}\concat S}
      & = \sem{\upmu \X. \PtoR{\varphi}} \set{\frame{\sigma, \Delta}\concat S}\\
      & = \bigcup_{n \in \N} \sem{\upmu^n \X. \PtoR{\varphi}} \set{\frame{\sigma, \Delta}\concat S}\\
      & = \bigcup_{n \in \N} \sem{\PtoR{\nu^{n} x. \varphi}} \set{\frame{\sigma, \Delta}\concat S}\\
      & = \bigcup_{n \in \N} \set{\frame{\sigma, \Delta}\concat S \mid \sigma \not\models \nu^{n} x. \varphi }\\
      & = \set{\frame{\sigma, \Delta}\concat S \mid \sigma \not\models \bigwedge_{n \in \N} \nu^{n} x. \varphi}\\
      & = \set{\frame{\sigma, \Delta}\concat S \mid \sigma \not\models  \nu x. \varphi}
    \end{align*}
  \end{itemize}
\end{proof}
    
  \end{toappendix}

  \begin{example}[Control flow graphs and $\mu_\Box$-calculus encoding]
    \label{ex:mu}
    Consider again Examples~\ref{ex:cfg} and~\ref{ex:lts}. We now
    check the property \emph{``if after three
      steps we end up in program point $3$ then we will loop forever,
    every $4$ steps, on this program point 3.''} It is known that regular
    properties of this kind cannot be expressed in $\ACTL$
    (see, e.g.,~\cite{PW83}). This property can be expressed in the
    $\mu_\Box$-calculus as
    $\psi \triangleq \Box^3 (n=3\ \to \nu x.\, (n=3\ \land\ \Box^4 x))$, 
    where
    $\Box^k$ is a shortand for $\Box \ldots \Box$ repeated $k$ times.
    Letting $\r^k$ as a shorthand for
    a $k$ times composition
    $\r; \ldots; \r$ of a $\MOKA$ program $\r$,
    the $\MOKA$ program encoding $\psi$ is
    $\PtoR{\psi}= (\push;\nexte)^3;(\mathsf{n = 3} ;
      \PtoR{\nu x.\, (n=3\
    \land\ \Box^4 x)});\pop^3$, with
    $\PtoR{\nu x.\, (n=3\ \land\ \Box^4 x)} = \upmu \X.\, ([\mathsf{n \neq 3?} \oplus ((\push;\nexte)^4; \X;
    \pop^4)])$.
    \qed
  \end{example}

It is well known that all ACTL formulae can be expressed as
$\mu_\Box$-calculus formulae.
For example,
$\AX \varphi = \Box \varphi$,
$\AF \varphi = \mu x.\, (\varphi \lor \Box x)$, and
$\AG \varphi = \nu x.\, (\varphi \land \Box x)$.
  Hence one could  obtain programs generating
  counterexamples for $\ACTL$ by encoding $\ACTL$
  in the
  $\mu_{\Box}$-calculus and then generating the corresponding program.
  However, this in general produces programs which are
  (unnecessarily) more complex than those
  produced for ACTL formulae. For instance, the program for $\AF
  \varphi$ obtained through the $\mu_\Box$-calculus encoding above would be
  $\push; \reset; \upmu \X. \Bigl (\loops \oplus (\add;
  \PtoR{\varphi}; \push; \nexte; \X; \pop)   \Bigr ); \pop$.

  \begin{toappendix}
    \begin{remark}[ACTL counterexamples programs as
      $\mu$-calculus formulae]
      For completeness, we explicitly write the programs $\PtoR{\varphi}$
      that correspond to (the modal fragment of) ACTL formulae
      $\varphi$ in the $\mu$-calculus encoding of ACTL formulae. 
      Let us note that, as expected, $\AX \varphi$ encoding as $\Box \varphi$ 
      is exactly the encoding of $\AX \varphi$.

      \begin{align*}
        \PtoR{\AF \varphi}
        & = \push; \reset; \upmu \X. \Bigl ( (\add; \PtoR{\varphi}; \push; \nexte; \X; \pop) \oplus  \loops \Bigr ); \pop\\
        \PtoR{\AG \varphi}
        & = \upmu \X. \Bigl (\PtoR{\varphi} \oplus (\push;\nexte; \X; \pop) \Bigr ) \\
        \PtoR{\varphi_1 \AU \varphi_2}
        & = \push; \reset;\upmu \X. \Bigl [ 
          \add; \Bigl ( \PtoR{\varphi_2}; (\PtoR{\varphi_1} \oplus
          \push; \nexte; \X; \pop) \Bigr ) \oplus \loops \Bigr ]; \pop
      \end{align*}
    \end{remark}
  \end{toappendix}

  \section{%
  Abstract Interpretation of $\MOKA$}
  \label{sec:stackabs}

  We show how to lift an abstraction over the states
  of a transition system to an abstraction over the domain of
  stacks. This is achieved stepwise by first considering
  an abstraction
  applied to each single stack, and then by merging classes of
  stacks through a
  suitable equivalence. The resulting stack abstraction will induce
  the abstract interpretation of $\MOKA$ programs.

  \subparagraph{Lifting the abstraction.}
  Let $(\Sigma, I, \trel)$  be a fixed transition system and let
  $\langle \alpha,\gamma \rangle: \pow{\Sigma} \rightleftarrows A$
  be an abstraction of state properties.
  We consider the lattice
  $\operatorname{F}_{A} \triangleq A \times A$,
  with componentwise order, whose elements $\frame{\sigma^{\sharp}, \delta^{\sharp}}$ are called \emph{abstract
  frames}, and, in turn,
  $\str{\operatorname{F}_{A}}$ whose elements $\frame{\sigma^{\sharp}, \delta^{\sharp}}\concat S^{\sharp}$ are called
  \emph{abstract stacks}.
  As in the concrete case, we abbreviate $\frame{\sigma^{\sharp}, \bot_A}$ as $\sigma^{\sharp}$.

  The abstraction map $\alpha$ on state properties can be extended to
  a frame abstraction, that by abusing the notation, we still
  denote by $\alpha$,
  and is defined by
  $\alpha(\frame{\sigma, \Delta}) \triangleq
  \frame{\alpha(\set{\sigma}), \alpha(\Delta)}$. In turn, the
  abstraction is inductively defined
  on stacks as
  $\alpha(\frame{\sigma, \Delta}\concat S^{n}) \triangleq
  \alpha(\frame{\sigma, \Delta}) \concat
  \alpha(S^{n})$.
  
  Now,  a set of
  stacks is abstracted to a set of abstract stacks by first applying $\alpha$
  pointwise using the
  image adjunction  (see Example~\ref{ex:image})
  and then joining classes of abstract stacks
  in a way which is parameterised by a suitable equivalence.
  Given an equivalence  $\mathbin{\sim} \subseteq L \times L$,
  we let $\class{x}{\sim}$
  denote the equivalence class of $x\in L$ w.r.t.~$\sim$. Given $x \in L$ we
  let  
$\cone{x} \triangleq \set{y \in L \mid y \leq x}$.

  \begin{definition}[Equivalence adjunction]
    \label{de:compatible-equivalence}
    Given a complete lattice
    $L$, a \emph{compatible equivalence} is an equivalence
    $\sim\subseteq L \times L$ such that for all $x \in
    L$, it holds that
    $\class{x}{\sim}$ is closed by joins of non-empty subsets. Let
    $\poweq{L}{\ssim} \triangleq \set{X \in \pow{L} \mid \forall x \in
      X.\, \class{x}{\sim} \cap X = \set{x}}$, ordered as follows: $X
    \leq_\ssim Y$ if for all $x \in X$ there is $y \in Y$ such that $x
    \sim y$ and $x \leq y$.

    \noindent
    Then, the pair $\langle \alpha_\ssim, \gamma_\ssim \rangle : \pow{L}
    \rightleftarrows
    \poweq{L}{\ssim}$
    defined, for $X \in \pow{L}$, $Y \in \poweq{L}{\ssim}$, by
    \begin{center}
      $\alpha_\ssim(X) \triangleq \set{ \bigvee (X \cap \class{x}{\sim})
      \mid x \in X}$
      \qquad
      $\gamma_\ssim(Y) \triangleq \bigcup \set{ \class{y}{\sim} \cap \cone{y}
      \mid y \in Y}$
    \end{center}
    is a Galois connection, called \emph{equivalence adjunction}
    (see Lemma~\ref{le:equiv}).
  \end{definition}
  In words, $\poweq{L}{\ssim}\subseteq \pow{L}$
  consists of the subsets of $L$ containing at most one
  representative for each
  equivalence class, and $\alpha_\ssim(X)$ abstracts $X$ in
  the subset (in $\poweq{L}{\ssim}$) consisting of the least upper bounds of
  $\sim$-equivalent elements in $X$.

  \begin{toappendix}
    Given a complete lattice $L$ and a compatible equivalence $\sim$ on
    $L$, note that given $W \subseteq \poweq{L}{\ssim}$ we have
    $\bigvee W = \alpha_\ssim(\bigcup W)$.  With the additional
    observation that $\alpha_\ssim$ is idempotent, this immediately shows
    that $\alpha_\ssim$ is additive and thus part of a Galois
    connection. A detailed proof is below.

    \begin{lemma}
      \label{le:equiv}
      Let $L$ be a lattice and $\sim$ be a compatible equivalence on
      $L$. Then the equivalence adjunction $\langle \alpha_\ssim,
      \gamma_\ssim \rangle : \pow{L} \rightleftarrows \poweq{L}{\ssim}$ is a Galois
      connection.
    \end{lemma}

    \begin{proof}
      We start observing that $\leq_\ssim$ is a partial order on
      $\poweq{L}{\ssim}$. In fact, let $X, Y, Z \in \poweq{L}{\ssim}$. Then
      \begin{itemize}
        \item
          $X \leq_\ssim X$ since $\sim$ is reflexive.
        \item Assume $X \leq_\ssim Y$ and $Y \leq_\ssim X$. We show
          that $X \subseteq Y$. In fact, for all $x \in X$, since
          $X \leq_\ssim Y$, there exists $y \in Y$ such that
          $x \sim y$ and $x \leq y$. In turn, by $Y \leq_\ssim X$
          there is $x' \in X$ such that $x' \sim y$ and $y \leq
          x'$. Therefore $x \sim x'$ from which we obtain $x=x'$ and
          thus $x=y$ by antisymmetry of $\leq$. Dually, one can show
          $Y \subseteq X$ and thus $X=Y$, as desired.

        \item Assume $X \leq_\ssim Y \leq_\ssim Z$. Then for all
          $x \in X$, since $X \leq_\ssim Y$, there exists $y \in Y$
          such that $x \sim y$ and $x \leq y$. In turn, by
          $Y \leq_\ssim Z$ this implies that there is $z \in Z$ such
          that $y \sim z$ and $y \leq z$. Hence $x \sim z$ and
          $x \leq z$. Thus $X \leq_\ssim Z$.
      \end{itemize}

      We now observe that $\poweq{L}{\ssim}$ is a complete lattice,
      by showing that, for a collection of elements
      $X_i \in \poweq{L}{\ssim}$ for $i \in I$, it holds
      \begin{center}
        $\bigvee_{i \in I} X_i = \alpha_\ssim(\bigcup_{i\in I} X_i)$
      \end{center}
      In fact, let $i \in I$ and $x \in X_i$. Then, if we define
      $y = \bigvee (\class{x}{\sim} \cap \bigcup_i X_i)$ we have
      that $x \sim y$ and
      $y \in \alpha_\ssim(\bigcup_{i\in I} X_i)$. Therefore
      $X_i \leq_\ssim \alpha_\ssim(\bigcup_{i\in I} X_i)$, i.e.,
      $\alpha_\ssim(\bigcup_{i\in I} X_i)$ is an upper bound.

      Let us show that it is the least upper bound. Given any other
      upper bound $Z$ for the $X_i$'s, take
      $y \in \alpha_\ssim(\bigcup_{i\in I} X_i)$. By definition of
      $\alpha_\ssim$,
      $y = \bigvee \class{x}{\sim} \cap \bigcup_{i\in I, X_i \cap
      \class{x}{\sim} \neq \varnothing} X_i$ for some $x \in
      L$. Now, for all $i\in I$ such that
      $X_i \cap \class{x}{\sim} \neq \varnothing$ let
      $X_i \cap \class{x}{\sim} = \{x_i\}$. Since $X_i \leq_\ssim Z$
      there is $z_i \in Z$ such that $z_i \sim x_i$ and
      $x_i \leq z_i$. Now all the $x_i$'s are equivalent and thus
      all the $z_i$'s are so. This implies that they all coincide
      with some $z \in Z$. Now
      $y = \bigvee \{ x_i \mid X \cap \class{x}{\sim} \cap X_i =
      \{x_i\}\} \leq z$ and $x \sim x_i \sim z$. Thus
      $\alpha_\ssim(\bigcup_{i\in I} X_i) \leq_\ssim Z$.

      From the above characterisation of the least upper bound in
      $\poweq{L}{\ssim}$, it follows immediately that $\alpha_\ssim$
      is additive. Then an easy calculation proves that
      $\gamma_\ssim$, as given in the statement, is its adjoint.
    \end{proof}

    \begin{definition}[Pointwise lifting]
      \label{de:pointwise-lifting}
      Given
      $\langle \alpha, \gamma \rangle: \pow{\Sigma} \rightleftarrows
      A$ we call the corresponding image adjunction
      $\langle \dirIm{\alpha}, \invIm{\alpha} \rangle:
      \pow{\str{\operatorname{F}_{\Sigma}}} \to
      \pow{\str{\operatorname{F}_{A}}}$ the \emph{pointwise lifting} of
      $\langle \alpha, \gamma \rangle$.
    \end{definition}

    \begin{note}
      \label{no:BCAs}
      Observe that the fact that
      the lifting is a Galois connection does
      not depend on the fact that
      $\langle \alpha, \gamma \rangle$ is a Galois
      connection. However this is useful to have BCAs of relevant
      operators, like $\loops$ and $\add$, that are directly computable in the
      abstract domain
      (see Theorem~\ref{th:BCAs}).
    \end{note}
  \end{toappendix}

  We can now define the stack abstraction parameterised
  by an equivalence on the lattice of abstract
  frames.
  Given an equivalence $\mathbin{\sim}\subseteq X\times X$ on any
  set $X$, $\sim_+$ denotes
  the equivalence on $\str{X}$ induced by $\sim$ as follows: for all
  $x_1 \cdots x_n,\: y_1\cdots y_m \in \str{A}$,
  $x_1 \ldots x_n \sim_+ y_1\ldots y_m \in \str{X}$ if $n=m$ and
  $x_i \sim y_i$ for all $i \in \set{1,\ldots, n}$.

  \begin{definition}[Stack abstraction]
    \label{de:equivalence-abstraction}
    Let $\sim$ be a compatible equivalence on the abstract domain $A$
    and let us extend it to the lattice of abstract frames
    $\operatorname{F}_{A}$ by
    $\frame{\sigma_1^{\sharp},\delta_1^{\sharp}} \sim
    \frame{\sigma_2^{\sharp},\delta_2^{\sharp}}$ if
    $\sigma_1^{\sharp} \sim \sigma_2^{\sharp}$.  The
    \emph{$\sim$-stack abstraction} is the domain $\As \triangleq \poweq{\str{\operatorname{F}_{A}}}{\ssim}$ along with the Galois connection
    $\langle \alpha^s_{\ssim}, \gamma^s_{\ssim} \rangle :
    \pow{\str{\operatorname{F}_{\Sigma}}} \rightleftarrows
    \As$ defined by
    $\alpha^s_\ssim \triangleq \alpha_{\ssimplus} \circ \dirIm{\alpha}$
    and
    $\gamma^s_\ssim \triangleq \invIm{\alpha} \circ \gamma_{\ssimplus}$.
  \end{definition}

  In words, given a set $X$ of stacks, its abstraction
  $\alpha^s_\ssim(X)$ first applies pointwise the underlying $\alpha$
  to each stack in $X$, and then joins equivalent abstract stacks.  As
  corner cases we can have the identity relation $\mathbin{\sim} =
  \mathrm{id}$, which joins
  abstract frames with identical first component, and
  $\mathbin{\sim} =A \times A$,
  the trivial relation, which
  joins all abstract stacks into one.

  \subparagraph{An abstract interpreter for counterexamples.}
  Given an abstraction for state properties
  $\langle \alpha,\gamma \rangle: \pow{\Sigma} \rightleftarrows A$ and
  a compatible equivalence on $A$, following the standard approach in
  abstract interpretation~\cite{cousot21}, we consider an abstract
  semantics
  $\sem{\r}^\sharp: \As \to
  \As$, defined inductively as explained
  in~\S~\ref{sec:mocha} and 
  using as abstract semantics for basic
  expressions $\e$ their BCAs on $\As$, 
  denoted 
  by $\bsem{\e}^{\As}$.
  We write $\sem[]{\mathsf{r}}^{\sharp}$ instead of
  $\sem{\mathsf{r}}^{\sharp}$ when the abstract environment is inessential.
  Notably, the BCAs of basic expressions 
  independent from 
  the underlying system can be effectively defined, 
  and some of them result to be complete.
  \begin{toappendix}
    We next prove technical lemma which shows that if a function $f$ over 
    concrete stacks only acts on the top frame, then also its BCA does.
    \begin{lemma}
      \label{lem:top-frame}
      Let $f: \pow{\operatorname{F}_{\Sigma}^+} \to \pow{\operatorname{F}_{\Sigma}^+}$ and assume that there is a function
      $g : \operatorname{F}_{\Sigma} \to \operatorname{F}_{\Sigma}$ such that $f$ is the additive extension of
      \begin{equation*}
        f(\set{\frame{\sigma, \Delta} \concat S}) = \set{\frame{\sigma', \Delta'}\concat S \mid \frame{\sigma', \Delta'} \in g(\frame{\sigma, \Delta})}
      \end{equation*}
      Then $f^{A} : \pow{\operatorname{F}_{A}^+} \to \pow{\operatorname{F}_{A}^+}$ is the additive extension of
      \begin{equation*}
        f^{A}(\set{\frame{\sigma^{\sharp}, \delta^{\sharp}} \concat S^{\sharp}}) = \set{\frame{\sigma^{\flat}, \Delta^{\flat}}\concat S^{\sharp} \mid 
        \frame{\sigma^{\flat}, \Delta^{\flat}} \in g^{A}(\frame{\sigma^{\sharp},\delta^{\sharp}})}
    \end{equation*}
    and $f^{\As} : \poweq{\operatorname{F}_{A}^+}{\ssim} \to \poweq{\operatorname{F}_{A}^+}{\ssim}$ is the additive extension of
      \begin{equation*}
        f^{\As}(\set{\frame{\sigma^{\sharp}, \delta^{\sharp}} \concat S^{\sharp}}) = \set{\frame{\sigma^{\flat}, \Delta^{\flat}}\concat S^{\sharp} \mid 
        \frame{\sigma^{\flat}, \Delta^{\flat}} \in g^{\As}(\frame{\sigma^{\sharp},\delta^{\sharp}})}
    \end{equation*}
      \end{lemma}
      \begin{proof}
        \begin{align*}
          f^{A}\left ( \set{\frame{\sigma^{\sharp}, \delta^{\sharp}} \concat S^{\sharp}} \right ) 
          & = \dirIm{\alpha} f \invIm{\alpha}\left ( \set{\frame{\sigma^{\sharp}, \delta^{\sharp}} \concat S^{\sharp}} \right ) \\
          & = \dirIm{\alpha} f \left ( \set{\frame{\sigma, \Delta} \concat S 
          \mid \alpha(\frame{\sigma, \Delta}) = \frame{\sigma^{\sharp}, \delta^{\sharp}} \land \alpha(S) = S^{\sharp}} \right ) \\
          & = \dirIm{\alpha} \left ( \left \{  \frame{\sigma', \Delta'} \concat S
          \mid \frame{\sigma', \Delta'} \in g(\frame{\sigma, \Delta}) \right . \right . \\ 
          & \qquad \left . \left . \land \ \alpha(\frame{\sigma, \Delta}) = \frame{\sigma^{\sharp}, \delta^{\sharp}} \land \alpha(S) = S^{\sharp} \right \} \right ) \\
          & = \set{\frame{\sigma'^{\sharp}, \delta'^{\sharp}} \concat S^{\sharp} \mid \frame{\sigma'^{\sharp}, \delta'^{\sharp}} 
          \in g^A(\frame{\sigma^{\sharp}, \delta^{\sharp}})}
        \end{align*}
        \begin{align*}
          f^{\As}\left ( \set{\frame{\sigma^{\sharp}, \delta^{\sharp}} \concat S^{\sharp}} \right )  
          & = \alpha_{\ssimplus} f^{A} \gamma_{\ssimplus}\left ( \set{\frame{\sigma^{\sharp}, \delta^{\sharp}} \concat S^{\sharp}} \right ) \\
          & = \alpha_{\ssimplus} f^{A} \left ( \left \{ \frame{\sigma^{\flat}, \delta^{\flat}}\concat S^{\flat} 
          \mid \frame{\sigma^{\flat}, \delta^{\flat}}\concat S^{\flat} \sim \frame{\sigma^{\sharp}, \delta^{\sharp}}\concat S^{\sharp} \right . \right . \\
          & \qquad \left .  \left . \land \ \frame{\sigma^{\flat}, \delta^{\flat}}\concat S^{\flat} 
          \leq \frame{\sigma^{\sharp}, \delta^{\sharp}}\concat S^{\sharp}\right \} \right ) \\
          & = \alpha_{\ssimplus} \left ( \left \{ \frame{\sigma'^{\flat}, \delta'^{\flat}} \concat S^{\flat}
          \mid \frame{\sigma'^{\flat}, \delta'^{\flat}} \in g^A(\frame{\sigma^{\flat}, \delta^{\flat}}) \right . \right . \\
          & \qquad \left . \left . \land \  \frame{\sigma^{\flat}, \delta^{\flat}}\concat S^{\flat} \sim \frame{\sigma^{\sharp}, \delta^{\sharp}}\concat S^{\sharp} 
          \land \frame{\sigma^{\flat}, \delta^{\flat}}\concat S^{\flat} \leq \frame{\sigma^{\sharp}, \delta^{\sharp}}\concat S^{\sharp}
          \right \} \right ) \\
          & = \set{\frame{\sigma'^{\sharp}, \delta'^{\sharp}} \concat S^{\sharp} \mid \frame{\sigma'^{\sharp}, \delta'^{\sharp}} 
          \in g^{\As}(\frame{\sigma^{\sharp}, \delta^{\sharp}})}
        \end{align*}
      \end{proof}

    \begin{theorem}
      \label{th:BCAs}
      Let
      $\frame{\sigma^{\sharp}, \delta^{\sharp}}\concat S^{\sharp}\in
      \str{\operatorname{F}_{A}}$ such that
      $\sigma^{\sharp} = \alpha(\set{\sigma})$ for some
      $\sigma \in \Sigma$.  Then the BCAs of the basic expressions
      $\loops, \add, \reset, \push, \pop$ in
      $\pow{{\str{\operatorname{F}_{A}}}}$ for the pointwise lifting
      (Definition~\ref{de:pointwise-lifting}) are as follows, where we write
      $\bsem{\cdot}^A$ as a shortcut for
      $\bsem{\cdot}^{\pow{{\str{\operatorname{F}_{A}}}}}$:
      \begin{center}
        \begin{tabular}{r@{\hspace{2pt}}c@{\hspace{2pt}}l}
          $\bsem{\add}^A \set{\frame{\sigma^{\sharp},
          \delta^{\sharp}} \concat S^{\sharp}}$ & $=$ &
          $\set{\frame{\sigma^{\sharp}, \delta^{\sharp} \vee
          \sigma^{\sharp}} \concat S^{\sharp}}$ \\
          $\bsem{\reset}^{A} \set{\frame{\sigma^{\sharp},
          \delta^{\sharp}} \concat S^{\sharp}}$ & $=$ &
          $\set{\frame{\sigma^{\sharp}, \bot} \concat S^{\sharp}}$ \\
          $\bsem{\loops}^{A} \set{\frame{\sigma^{\sharp},
          \delta^{\sharp}} \concat S^{\sharp}}$ & $=$ &
          $\set{\frame{\sigma^{\sharp}, \delta^{\sharp}} \concat
          S^{\sharp} \mid \sigma^{\sharp} \leq \delta^{\sharp}}$ \\
          $\bsem{\push}^{A} \set{\frame{\sigma^{\sharp},
          \delta^{\sharp}} \concat S^{\sharp}}$ & $=$ &
          $\set{\frame{\sigma^{\sharp}, \delta^{\sharp}} \concat
          \frame{\sigma^{\sharp}, \delta^{\sharp}} \concat S^{\sharp}}$ \\
          $\bsem{\pop}^{A}\set{\frame{\sigma^{\sharp},
          \delta^{\sharp}} \concat S^{\sharp}}$ & $=$ & $\set{S^{\sharp}}$ \\
        \end{tabular}
      \end{center}
      Moreover, the three operations $\reset, \push, \pop$ are
      globally complete.
    \end{theorem}

    \begin{proof}
      \textbf{(BCA computation)} By exploiting Lemma~\ref{lem:top-frame}, we can compute the BCAs of $\add, \reset, \loops$
      considering only the top frame.
      \begin{align*}
        \bsem{\add}^{A} \set{\frame{\sigma^{\sharp}, \delta^{\sharp}}}
        & =
        \dirIm{\alpha}(\bsem{\add}(\invIm{\alpha}(\set{\frame{\sigma^{\sharp},
        \delta^{\sharp}}})))\\
        & = \dirIm{\alpha}(\bsem{\add}(\set{\frame{\sigma,
              \Delta} \mid \alpha(\set{\sigma}) = \sigma^{\sharp}
        \ \land \ \alpha(\Delta) = \delta^{\sharp}}))\\
        & = \dirIm{\alpha}( \set{ \frame{\sigma, \set{\sigma}
            \cup \Delta}  \mid \alpha(\set{\sigma}) =
        \sigma^{\sharp} \ \land \ \alpha(\Delta) = \delta^{\sharp}})\\
        & =  \set{ \frame{\alpha(\set{\sigma}), \alpha(\set{
          \sigma } \cup \Delta)} \mid \alpha(\set{\sigma}) =
        \sigma^{\sharp} \ \land \ \alpha(\Delta) = \delta^{\sharp}}\\
        & \overset{\clubsuit}{=}  \set{
          \frame{\alpha(\set{\sigma}), \alpha(\set{
          \sigma } ) \vee \alpha(\Delta)}  \mid \alpha(\set{\sigma}) =
        \sigma^{\sharp} \ \land \ \alpha(\Delta) = \delta^{\sharp}}\\
        & =  \set{ \frame{\sigma^{\sharp}, \delta^{\sharp} \vee
        \sigma^{\sharp}}  } %
      \end{align*}
      where the ($\clubsuit$) passage uses additivity of $\alpha$.

      \begin{align*}
      \bsem{\reset}^{A} \set{\frame{\sigma^{\sharp}, \delta^{\sharp}}}
        & =
        \dirIm{\alpha}(\bsem{\reset}(\invIm{\alpha}(\set{\frame{\sigma^{\sharp},
        \delta^{\sharp}}}))) \\
        & = \dirIm{\alpha}(\bsem{\reset}(\set{\frame{\sigma,
              \Delta}
              \mid \alpha(\set{\sigma}) = \sigma^{\sharp}
        \land \alpha(\Delta) = \delta^{\sharp}})) \\
        & = \dirIm{\alpha}(\set{\frame{\sigma,
            \varnothing} \mid \alpha(\set{\sigma}) =
        \sigma^{\sharp}}) \\
        & = \set{\frame{\sigma^{\sharp}, \bot}} %
      \end{align*}
      \begin{align*}
        \bsem{\loops}^{A}  (\set{\frame{\sigma^{\sharp},
        \delta^{\sharp}}})
        & =
        \dirIm{\alpha}(\bsem{\loops}(\invIm{\alpha}(\set{\frame{\sigma^{\sharp},
        \delta^{\sharp}}})))\\
        & = \dirIm{\alpha}( \bsem{\loops}(\set{\frame{\sigma,
              \Delta} \mid \alpha(\set{\sigma}) =
        \sigma^{\sharp}\ \land\ \alpha(\Delta) = \delta^{\sharp}}))\\
        & = \dirIm{\alpha}( \set{\frame{\sigma, \Delta} \mid
            \alpha(\set{\sigma}) =
        \sigma^{\sharp}\ \land\ \alpha(\Delta) = \delta^{\sharp} })\\
        & \overset{\heartsuit}{=} \set{ \frame{\sigma^{\sharp},
          \delta^{\sharp}} \mid \sigma^{\sharp}
          \leq \delta^{\sharp}} %
      \end{align*}
      where the ($\heartsuit$) passage uses the fact that there if there is
      $\sigma$ such that $\sigma^{\sharp} =
      \sigma_\Sigma(\set{\sigma})$, then by the
      properties of the adjunction
      $\set{ \sigma} \leq \gamma(\delta^{\sharp})$ iff
      $\sigma^{\sharp} = \alpha(\set{\sigma}) \leq \delta^{\sharp}$.
      \begin{align*}
        & \bsem{\push}^{A} \set{\frame{\sigma^{\sharp}, \delta^{\sharp}}
        \concat S^{\sharp}} %
        = \dirIm{\alpha}(\bsem{\push}(\invIm{\alpha}(\set{\frame{\sigma^{\sharp},
        \delta^{\sharp}}\concat S^{\sharp}}))) \\
        & = \dirIm{\alpha}(\bsem{\push}(\set{\frame{\sigma,
              \Delta}\concat S \mid
              \alpha(\set{\sigma}) = \sigma^{\sharp}
        \land \alpha(\Delta) = \delta^{\sharp} 
        \ \land \ S \in \invIm{\alpha}(S^{\sharp})})) \\
        & = \dirIm{\alpha}(\set{\frame{\sigma,
            \Delta}\concat\, \frame{\sigma,
            \Delta}\concat S \mid
            \alpha(\set{\sigma}) = \sigma^{\sharp}
        \land \alpha(\Delta) = \delta^{\sharp} 
        \ \land \ S \in \invIm{\alpha}(S^{\sharp})}) \\
        & = \set{\frame{\sigma^{\sharp}, \delta^{\sharp}}\concat\,
        \frame{\sigma^{\sharp}, \delta^{\sharp}}\concat S^{\sharp}}\\
        & \\ 
        & \bsem{\pop}^{A}  \set{\frame{\sigma^{\sharp}, \delta^{\sharp}}
        \concat S^{\sharp}} %
        = \dirIm{\alpha}(\bsem{\pop}(\invIm{\alpha}(\set{\frame{\sigma^{\sharp},
        \delta^{\sharp}}\concat S^{\sharp}}))) \\
        & = \dirIm{\alpha}(\bsem{\pop}(\set{\frame{\sigma,
              \Delta}\concat S \mid
              \alpha(\set{\sigma}) = \sigma^{\sharp}
        \land \alpha(\Delta) = \delta^{\sharp} 
        \ \land \ S \in \invIm{\alpha}(S^{\sharp})})) \\
        & = \dirIm{\alpha}(\invIm{\alpha}(S^{\sharp})) %
        =  \set{S^{\sharp}}%
      \end{align*}
      \textbf{(Global Completeness)} Let $\cel = \set{\frame{\sigma_i,
        \Delta_i}\concat S_i \mid  i
      \in I }$.
      \begin{align*}
        (\dirIm{\alpha} \circ \bsem{\reset})(\cel) & =
        \dirIm{\alpha} \left (
          \set{\frame{\sigma_i, \varnothing}
        \concat S_i \mid \ i \in I} \right ) \\
        & = \set{\frame{\alpha(\set{\sigma_i}), \bot} \concat
        \dirIm{\alpha}(S_i) \mid i \in I}\\
        (\bsem{\reset}^{A}\circ \dirIm{\alpha})(\cel) & =
        \bsem{\reset}^{A} \left (
          \set{\frame{\alpha(\set{\sigma_i}),
            \alpha(\Delta_i)}\concat \dirIm{\alpha}(S_i) \mid
        i \in I} \right ) \\
        & = \set{\frame{\alpha(\set{\sigma_i}), \bot} \concat
        \dirIm{\alpha}(S_i) \mid i \in I}
      \end{align*}
      \begin{align*}
        (\dirIm{\alpha} \circ \bsem{\push})(\cel) & =
        \dirIm{\alpha} \left (
          \set{\frame{\sigma_i, \Delta_i}
            \concat\, \frame{\sigma_i, \Delta_i}
        \concat S_i \mid \ i \in I} \right ) \\
        & =
        \set{\frame{\alpha(\set{\sigma_i}),\alpha(\Delta_i)}
          \concat\, \frame{\alpha(\set{\sigma_i}),\alpha(\Delta_i)}
          \concat
        \dirIm{\alpha}(S_i) \mid i \in I}\\
        (\bsem{\push}^{A}\circ \dirIm{\alpha})(\cel) & =
        \dirIm{\alpha} \left (
          \set{\frame{\sigma_i, \Delta_i}
            \concat\, \frame{\sigma_i, \Delta_i}
        \concat S_i \mid \ i \in I} \right ) \\
        & =
        \set{\frame{\alpha(\set{\sigma_i}),\alpha(\Delta_i)}
          \concat\, \frame{\alpha(\set{\sigma_i}),\alpha(\Delta_i)}
          \concat
        \dirIm{\alpha}(S_i) \mid i \in I}
      \end{align*}
      \begin{align*}
        (\dirIm{\alpha} \circ \bsem{\pop})(\cel) & =
        \dirIm{\alpha} \left (
        \set{S_i \mid \ i \in I} \right ) \\
        & =
        \set{\dirIm{\alpha}(S_i) \mid i \in I}\\
        (\bsem{\pop}^{A}\circ \dirIm{\alpha})(\cel) & =
        \bsem{\pop}^{A} \left (
          \set{\frame{\alpha(\set{\sigma_i}),
            \alpha(\Delta_i)}\concat \dirIm{\alpha}(S_i) \mid
        i \in I} \right ) \\
        & =
        \set{\dirIm{\alpha}(S_i) \mid i \in I}%
      \end{align*}
    \end{proof}

    When we move to coarser abstractions based on a compatible
    equivalence over the domain $A$, we can still retain a similar result.

    \begin{definition}[strong equivalence]
      \label{def:strong-eq}
      Let
      $\langle \alpha,\gamma \rangle: \pow{\Sigma} \rightleftarrows A$
      be an abstraction of state properties and let $\sim$ be a
      compatible equivalence on the abstract domain $A$. Say that a
      set of states $X \subseteq \Sigma$ is $\sim$-uniform if for all
      $\sigma, \sigma' \in X$
      $\alpha(\{\sigma\}) \sim \alpha(\{\sigma'\})$. We say that $\sim$
      is \emph{strong} when for all $X \subseteq \Sigma$ $\sim$-uniform it
      holds that $\gamma(\alpha(X))$ is $\sim$-uniform.
    \end{definition}
    \begin{remark}[equivalence properties]
      \label{rem:eq}
      Let $\sigma$ be a strong equivalence and let
      $\sigma_1^\sharp, \sigma_2^\sharp$ such that
      $\sigma_1^\sharp \leq \sigma_2^\sharp$. If
      $\sigma_2^\sharp = \alpha(X)$ for some $\sim$-uniform set
      $X \subseteq \Sigma$ then
      $\sigma_1^\sharp = \alpha(\gamma(\sigma_1^\sharp))$ with
      $\gamma(\sigma_1^\sharp) \subseteq \gamma(\sigma_2^\sharp) =
      \gamma(\alpha(X))$. Hence also $\sigma_1^\sharp = \alpha(Y)$ for
      $Y = \gamma(\sigma_2^\sharp)$ a $\sim$-uniform set and thus
      $\sigma_1^\sharp \sim \sigma_2^{\sharp}$.
    \end{remark}

    Below,
    given a compatible equivalence $\sim$, we denote by
    $\sem[]{\cdot}^\sim$ the BCA of $\sem[]{\cdot}^{\sharp}$ with
    respect to the frame abstraction:

    \begin{theorem}
      \label{th:BCAs-equiv}
      Let
      $\frame{\sigma^{\sharp}, \delta^{\sharp}}\concat S^{\sharp}\in
      \str{\operatorname{F}_{A}}$ and consider an abstraction as in
      Definition~\ref{de:equivalence-abstraction}.
      Then the BCAs of the basic expressions $\loops, \add, \reset,$
      in the domain $\As = \pow{\str{\operatorname{F}_A}}_\ssim$, denoted
      $\bsem{\cdot}^{\As}$, computed with respect to the
      pointwise lifting (Definition~\ref{de:pointwise-lifting}) are the same
      as Theorem~\ref{th:BCAs}, with the exception of operator $\loops$ for which we have
      \begin{center}
        \begin{tabular}{r@{\hspace{2pt}}c@{\hspace{2pt}}l}
          $\bsem{\loops}^{\As} \set{\frame{\sigma^{\sharp},
          \delta^{\sharp}} \concat S^{\sharp}}$ & $\leq_\ssim$ &
          $\set{\frame{\sigma^{\sharp},
            \delta^{\sharp}} \concat S^{\sharp} \mid \sigma^{\sharp}
          \wedge \delta^{\sharp} \neq \bot}$ \\
        \end{tabular}
      \end{center}
      If $\sim$ is strong then  equality holds also for $\loops$, i.e.,
           \begin{center}
        \begin{tabular}{r@{\hspace{2pt}}c@{\hspace{2pt}}l}
          $\bsem{\loops}^{\As} \set{\frame{\sigma^{\sharp},
          \delta^{\sharp}} \concat S^{\sharp}}$ & $=$ &
          $\set{\frame{\sigma^{\sharp} \wedge \delta^{\sharp},
            \delta^{\sharp}} \concat S^{\sharp} \mid \sigma^{\sharp}
          \wedge \delta^{\sharp} \neq \bot}$ \\
        \end{tabular}
      \end{center}
    \end{theorem}

    \begin{proof}
      \textbf{(BCA computation)} By exploiting Lemma~\ref{lem:top-frame}, we can compute the BCAs of $\add, \reset, \loops$
      considering only the top frame.
      \begin{align*}
        \bsem{\add}^{\As}(\set{\frame{\sigma^{\sharp}, \delta^{\sharp}}})
        &= \alpha_{\ssimplus}(\bsem{\add}^{A}(\gamma_{\ssimplus}(\set{\frame{\sigma^{\sharp},
        \delta^{\sharp}}})))\\
        & \hspace{-5ex} = \alpha_{\ssimplus}(\bsem{\add}^{A}(
            \class{\frame{\sigma^{\sharp}, \delta^{\sharp}}}{\sim}
        \cap \cone{\frame{\sigma^{\sharp}, \delta^{\sharp}}}))\\
        & \hspace{-5ex}=   \alpha_{\ssimplus} \left (\bsem{\add}^{A}\left(
            \set{\frame{\sigma^{\flat}, \delta^{\flat}} \mid
              \frame{\sigma^{\flat}, \delta^{\flat}} \sim
              \frame{\sigma^{\sharp}, \delta^{\sharp}} \land
              \frame{\sigma^{\flat}, \delta^{\flat}} \leq
        \frame{\sigma^{\sharp}, \delta^{\sharp}}}\right)\right )\\
        & \hspace{-5ex}=   \alpha_{\ssimplus} \left (
          \set{\frame{\sigma^{\flat}, \delta^{\flat} \vee
            \sigma^{\flat}} \mid
            \frame{\sigma^{\flat}, \delta^{\flat}} \sim
            \frame{\sigma^{\sharp}, \delta^{\sharp}} \land
            \frame{\sigma^{\flat}, \delta^{\flat}} \leq
        \frame{\sigma^{\sharp}, \delta^{\sharp}}} \right )\\
        &  \hspace{-5ex} \overset{\clubsuit}{=}  \set{\frame{\sigma^{\sharp},
        \delta^{\sharp} \vee \sigma^{\sharp}}}
      \end{align*}
      \begin{align*}     
        \hspace{-2ex} \bsem{\reset}^{\As} \set{\frame{\sigma^{\sharp},
        \delta^{\sharp}}} %
        & = \alpha_{\ssimplus}(\bsem{\reset}^{A}(\gamma_{\ssimplus}(\set{\frame{\sigma^{\sharp},
        \delta^{\sharp}}})))\\
        & = \alpha_{\ssimplus} \left (\bsem{\reset}^{A} \left (
            \class{\frame{\sigma^{\sharp}, \delta^{\sharp}}}{\sim}
        \cap \cone{\frame{\sigma^{\sharp}, \delta^{\sharp}}}\right) \right)\\
        & = \alpha_{\ssimplus}\left (\bsem{\reset}^{A} \left (
            \set{\frame{\sigma^{\flat}, \delta^{\flat}} \mid
              \frame{\sigma^{\flat}, \delta^{\flat}}\sim
              \frame{\sigma^{\sharp}, \delta^{\sharp}} \land
              \frame{\sigma^{\flat}, \delta^{\flat}}\leq
        \frame{\sigma^{\sharp}, \delta^{\sharp}}}\right ) \right )\\
        & = \alpha_{\ssimplus} \left (
          \set{\frame{\sigma^{\flat}, \bot} \mid
            \frame{\sigma^{\flat}, \delta^{\flat}} \sim
            \frame{\sigma^{\sharp}, \delta^{\sharp}} \land
            \frame{\sigma^{\flat}, \delta^{\flat}} \leq
        \frame{\sigma^{\sharp}, \delta^{\sharp}}}\right )\\
        & \overset{\clubsuit}{=} \set{\frame{\sigma^{\sharp}, \bot}}
      \end{align*}
      
      Where  at 
      ($\clubsuit$) we use the fact that frame equivalence only depends on its first component.
      
      For $\loops$ we have
      \begin{align*} 
        \hspace{-2ex} \bsem{\loops}^{\As} \set{\frame{\sigma^{\sharp},
        \delta^{\sharp}}}  %
        & = \alpha_{\ssimplus}(\bsem{\loops}^{A}(\gamma_{\ssimplus}(\set{\frame{\sigma^{\sharp},
        \delta^{\sharp}}})))\\
        & = \alpha_{\ssimplus}(\bsem{\loops}^{A}(
            \class{\frame{\sigma^{\sharp}, \delta^{\sharp}}}{\sim}
        \cap \cone{\frame{\sigma^{\sharp}, \delta^{\sharp}}}))\\
        & = \alpha_{\ssimplus}\left (\bsem{\loops}^{A}\left(
            \set{\frame{\sigma^{\flat}, \delta^{\flat}} \mid
              \frame{\sigma^{\flat}, \delta^{\flat}} \sim
              \frame{\sigma^{\sharp}, \delta^{\sharp}} \land
              \frame{\sigma^{\flat}, \delta^{\flat}} \leq
        \frame{\sigma^{\sharp}, \delta^{\sharp}}}\right ) \right)\\
        & = \alpha_{\ssimplus}\left (
          \set{\frame{\sigma^{\flat}, \delta^{\flat}} \mid
            \frame{\sigma^{\flat}, \delta^{\flat}} \sim
            \frame{\sigma^{\sharp}, \delta^{\sharp}} \land
            \frame{\sigma^{\flat}, \delta^{\flat}} \leq
            \frame{\sigma^{\sharp}, \delta^{\sharp}}\land \sigma^{\flat}
        \leq \delta^{\flat}}\right)\\
        & = \alpha_{\ssimplus}\left(
          \set{\frame{\sigma^{\flat}, \delta^{\flat}} \mid
            \frame{\sigma^{\flat}, \delta^{\flat}} \sim
            \frame{\sigma^{\sharp}, \delta^{\sharp}} \land
            \sigma^{\flat} \leq \sigma^{\sharp} \land
            \delta^{\flat} \leq  \delta^{\sharp}
        \land \sigma^{\flat} \leq \delta^{\flat}}\right )\\
        & = \set{\frame{\bigvee \set{\sigma^{\flat} \mid \sigma^{\flat} \sim
          \sigma^{\sharp} \land
          \sigma^{\flat} \leq \sigma^{\sharp}          
          \land \sigma^{\flat} \leq \delta^{\sharp}}, \delta^{\sharp}} }\\
        & = \set{\frame{\bigvee \set{\sigma^{\flat} \mid \sigma^{\flat} \sim
          \sigma^{\sharp} \land
          \sigma^{\flat} \leq \sigma^{\sharp} \wedge \delta^{\sharp}    
           }, \delta^{\sharp}}}\\
        & \leq_\ssim
        \set{\frame{\sigma^{\sharp}, \delta^{\sharp}}
          \mid \sigma^{\sharp}
        \wedge \delta^{\sharp} \neq \bot}
      \end{align*}

      When the equivalence is strong, in the last passage we obtain
      $\bsem{\loops}^{\As} \set{\frame{\sigma^{\sharp},
          \delta^{\sharp}}} = \set{\frame{\sigma^{\sharp}\wedge
          \delta^{\sharp}, \delta^{\sharp}} \mid
        \sigma^{\sharp} \wedge \delta^{\sharp} \neq \bot}$. In fact,
      as observed in Remark~\ref{rem:eq}, the assumption of dealing with a strong equivalence, 
      guarantees that for all $\sigma^{\flat}$ such that
      $\sigma^{\flat} \leq \sigma^{\sharp}$ it holds
      $\sigma^\flat \sim \sigma^{\sharp}$.
      \begin{align*}
       & \bsem{\push}^{\As}(\set{\frame{\sigma^{\sharp}, \delta^{\sharp}}\concat S^{\sharp}})
        =\\
        & =
          \alpha_{\ssimplus}(\bsem{\push}^{A}(\gamma_{\ssimplus}(\set{\frame{\sigma^{\sharp},
          \delta^{\sharp}}\concat S^{\sharp}})))\\
          & = \alpha_{\ssimplus}(\bsem{\push}^{A}(
              \class{\frame{\sigma^{\sharp}, \delta^{\sharp}}\concat S^{\sharp}}{\sim}
          \cap \cone{\frame{\sigma^{\sharp}, \delta^{\sharp}}\concat S^{\sharp}}))\\
          & = \alpha_{\ssimplus}\left (\bsem{\push}^{A} \left (
              \set{\frame{\sigma^{\flat}, \delta^{\flat}}\concat S^{\flat} \mid
                \frame{\sigma^{\flat}, \delta^{\flat}}\concat S^{\flat} \sim
                \frame{\sigma^{\sharp}, \delta^{\sharp}}\concat S^{\sharp} \land
                \frame{\sigma^{\flat}, \delta^{\flat}}\concat S^{\flat} \leq
          \frame{\sigma^{\sharp}, \delta^{\sharp}}\concat S^{\sharp}}\right ) \right)\\
          & = \alpha_{\ssimplus} \left (\set{\frame{\sigma^{\flat}, \delta^{\flat}}\concat
          \frame{\sigma^{\flat}, \delta^{\flat}}\concat S^{\flat} \mid
              \frame{\sigma^{\flat}, \delta^{\flat}}\concat S^{\flat} \sim
              \frame{\sigma^{\sharp}, \delta^{\sharp}}\concat S^{\sharp} \land
              \frame{\sigma^{\flat}, \delta^{\flat}}\concat S^{\flat} \leq
        \frame{\sigma^{\sharp}, \delta^{\sharp}}\concat S^{\sharp}}\right)\\
        & = \set{\frame{\sigma^{\sharp}, \delta^{\sharp}}\concat\frame{\sigma^{\sharp}, \delta^{\sharp}}\concat S^{\sharp}} \\ 
      & \\
        & \bsem{\pop}^{\As}(\set{\frame{\sigma^{\sharp}, \delta^{\sharp}}\concat S^{\sharp}})
        =\\
        & =
          \alpha_{\ssimplus}(\bsem{\pop}^{A}(\gamma_{\ssimplus}(\set{\frame{\sigma^{\sharp},
          \delta^{\sharp}}\concat S^{\sharp}})))\\
          & = \alpha_{\ssimplus}(\bsem{\pop}^{A}(
              \class{\frame{\sigma^{\sharp}, \delta^{\sharp}}\concat S^{\sharp}}{\sim}
          \cap \cone{\frame{\sigma^{\sharp}, \delta^{\sharp}}\concat S^{\sharp}}))\\
          & = \alpha_{\ssimplus}\left (\bsem{\pop}^{A}\left(
              \set{\frame{\sigma^{\flat}, \delta^{\flat}} \concat S^{\flat} \mid
                \frame{\sigma^{\flat}, \delta^{\flat}}\concat S^{\flat} \sim
                \frame{\sigma^{\sharp}, \delta^{\sharp}} \concat S^{\sharp} \land
                \frame{\sigma^{\flat}, \delta^{\flat}}\concat S^{\flat} \leq
          \frame{\sigma^{\sharp}, \delta^{\sharp}} \concat S^{\sharp}}\right) \right)\\
          & = \alpha_{\ssimplus} \set{S^{\flat} \mid 
          \frame{\sigma^{\flat}, \delta^{\flat}}\concat S^{\flat} \sim
          \frame{\sigma^{\sharp}, \delta^{\sharp}} \concat S^{\sharp} \land
          \frame{\sigma^{\flat}, \delta^{\flat}}\concat S^{\flat} \leq
    \frame{\sigma^{\sharp}, \delta^{\sharp}} \concat S^{\sharp}}\\
        & = \set{S^{\sharp}}
      \end{align*}
      \textbf{(Global Completeness)} Let $\ela = \set{\frame{\sigma^{\sharp}_i,
          \delta^{\sharp}_i}\concat S^{\sharp}_i \mid i \in I } \in \As = \pow{\str{\operatorname{F}_A}}$.

      \subparagraph*{Operator $\reset$}
        \begin{align*}
          (\alpha_{\ssimplus} \circ \bsem{\reset}^A)(\ela) & =
          \alpha_{\ssimplus} \left (
            \set{\frame{\sigma^{\sharp}_i, \bot}
          \concat S^{\sharp}_i \mid \ i \in I} \right ) \\
          & = \set{\bigvee \class{\frame{\sigma^{\sharp}_j, \bot} \concat S^{\sharp}_j}{\sim} 
          \cap  \set{\frame{\sigma^{\sharp}_i, \bot}
          \concat S^{\sharp}_i \mid \ i \in I} \mid j \in I} \\
          & = \set{\frame{\overline{\sigma^{\sharp}_j}, \bot} \concat \overline{S^{\sharp}_j} \mid j \in I}
        \end{align*}
        where we named $\frame{\overline{\sigma^{\sharp}_j}, \bot} \concat \overline{S^{\sharp}_j} = 
        \bigvee \class{\frame{\sigma^{\sharp}_j, \bot} \concat S^{\sharp}_j}{\sim} 
        \cap  \set{\frame{\sigma^{\sharp}_i, \bot}
          \concat S^{\sharp}_i \mid \ i \in I}$.
        \begin{align*}
          (\bsem{\reset}^{\As}\circ \alpha_{\ssimplus})(\ela) & =
          \bsem{\reset}^{\As} \left (
            \set{\bigvee \class{\frame{\sigma^{\sharp}_j, \delta^{\sharp}_j} \concat S^{\sharp}_j}{\sim} 
            \cap  \set{\frame{\sigma^{\sharp}_i, \bot}
          \concat S^{\sharp}_i \mid \ i \in I} \mid j \in I} \right ) \\
          & = \bsem{\reset}^{\As}\set{\frame{\widehat{\sigma^{\sharp}_j}, \widehat{\delta^{\sharp}_j}} 
          \concat \widehat{S^{\sharp}_j} \mid j \in I} \\
          & = \set{\frame{\widehat{\sigma^{\sharp}_j}, \bot} 
          \concat \widehat{S^{\sharp}_j} \mid j \in I} 
        \end{align*}
where $\frame{\widehat{\sigma^{\sharp}_j}, \widehat{\delta^{\sharp}_j}} 
        \concat \widehat{S^{\sharp}_j} =  \bigvee \class{\frame{\sigma^{\sharp}_j, \delta^{\sharp}_j} \concat S^{\sharp}_j}{\sim} 
        \cap  \set{\frame{\sigma^{\sharp}_i, \delta_i^{\sharp}}
      \concat S^{\sharp}_i \mid \ i \in I}$ and actually for all $j \in I$ we have 
      $\frame{\overline{\sigma^{\sharp}_j}, \bot} \concat \overline{S^{\sharp}_j} = \frame{\widehat{\sigma^{\sharp}_j}, \bot} 
          \concat \widehat{S^{\sharp}_j}$ by the strongness of the equivalence and the fact that 
        the order $\leq$ is pointwise on frame components and on stacks (sequences). 

        \subparagraph*{Operator $\push$}
        \begin{align*}
          (\alpha_{\ssimplus} \circ \bsem{\push}^A)(\ela) & =
          \alpha_{\ssimplus} \left (
            \set{\frame{\sigma^{\sharp}_i, \delta^{\sharp}_i} \concat \frame{\sigma^{\sharp}_i, \delta^{\sharp}_i}
          \concat S^{\sharp}_i \mid \ i \in I} \right ) \\
          & = \set{\bigvee \class{\frame{\sigma^{\sharp}_j, \delta^{\sharp}_j} \concat \frame{\sigma^{\sharp}_j, \delta^{\sharp}_j} \concat S^{\sharp}_j}{\sim} 
          \cap  \set{\frame{\sigma^{\sharp}_i, \delta^{\sharp}_i} \concat \frame{\sigma^{\sharp}_i, \delta^{\sharp}_i}
          \concat S^{\sharp}_i \mid \ i \in I} \mid j \in I} \\
          & = \set{\frame{\overline{\sigma^{\sharp}_j}, \overline{\delta^{\sharp}_j}} \concat \frame{\overline{\sigma^{\sharp}_j}, \overline{\delta^{\sharp}_j}} \concat \overline{S^{\sharp}_j} \mid j \in I}
        \end{align*}
        where $\frame{\overline{\sigma^{\sharp}_j}, \overline{\delta^{\sharp}_j}} \concat \frame{\overline{\sigma^{\sharp}_j}, 
        \overline{\delta^{\sharp}_j}} \concat \overline{S^{\sharp}_j} = 
        \bigvee \class{\frame{\sigma^{\sharp}_j, \delta^{\sharp}_j} \concat \frame{\sigma^{\sharp}_j, \delta^{\sharp}_j} \concat S^{\sharp}_j}{\sim} 
        \cap  \set{\frame{\sigma^{\sharp}_i, \delta^{\sharp}_i}
        \concat \frame{\sigma^{\sharp}_i, \delta^{\sharp}_i}
        \concat S^{\sharp}_i \mid \ i \in I}$.
      
        \begin{align*}
          (\bsem{\push}^{\As}\circ \alpha_{\ssimplus})(\ela) & =
          \bsem{\push}^{\As} \left (
            \set{\bigvee \class{\frame{\sigma^{\sharp}_j, \delta^{\sharp}_j} \concat S^{\sharp}_j}{\sim} 
            \cap  \set{\frame{\sigma^{\sharp}_i, \bot}
          \concat S^{\sharp}_i \mid \ i \in I} \mid j \in I} \right ) \\
          & = \bsem{\push}^{\As}\set{\frame{\widehat{\sigma^{\sharp}_j}, \widehat{\delta^{\sharp}_j}} 
          \concat \widehat{S^{\sharp}_j} \mid j \in I} \\
          & = \set{\frame{\widehat{\sigma^{\sharp}_j}, \widehat{\delta^{\sharp}_j}} 
          \concat \frame{\widehat{\sigma^{\sharp}_j}, \widehat{\delta^{\sharp}_j}} 
          \concat \widehat{S^{\sharp}_j} \mid j \in I} \\
        \end{align*}
        where $\frame{\widehat{\sigma^{\sharp}_j}, \widehat{\delta^{\sharp}_j}} 
        \concat \widehat{S^{\sharp}_j} =  \bigvee \class{\frame{\sigma^{\sharp}_j, \delta^{\sharp}_j} \concat S^{\sharp}_j}{\sim} 
        \cap  \set{\frame{\sigma^{\sharp}_i, \delta_i^{\sharp}}
      \concat S^{\sharp}_i \mid \ i \in I}$ and actually for all $j \in I$ we have 
      $\frame{\overline{\sigma^{\sharp}_j}, \overline{\delta^{\sharp}_j}} \concat \frame{\overline{\sigma^{\sharp}_j}, 
      \overline{\delta^{\sharp}_j}} \concat \overline{S^{\sharp}_j} = \frame{\widehat{\sigma^{\sharp}_j}, \widehat{\delta^{\sharp}_j}} 
      \concat \frame{\widehat{\sigma^{\sharp}_j}, \widehat{\delta^{\sharp}_j}} 
      \concat \widehat{S^{\sharp}_j}$ since 
        the order $\leq$ is pointwise on frame components and on stacks (sequences).

        \subparagraph*{Operator $\pop$}

        \begin{align*}
          (\alpha_{\ssimplus} \circ \bsem{\pop}^A)(\ela) & =
          \alpha_{\ssimplus} \left (
            \set{S^{\sharp}_i \mid \ i \in I} \right ) \\
          & = \set{\bigvee \class{S^{\sharp}_j}{\sim} 
          \cap  \set{S^{\sharp}_i \mid \ i \in I} \mid j \in I} \\
          & = \set{\overline{S^{\sharp}_j} \mid j \in I}
        \end{align*}
        where $\overline{S^{\sharp}_j} = 
        \bigvee \class{S^{\sharp}_j}{\sim} 
        \cap  \set{S^{\sharp}_i \mid \ i \in I}$.

          \begin{align*}
          (\bsem{\pop}^{\As}\circ \alpha_{\ssimplus})(\ela) & =
          \bsem{\pop}^{\As} \left (
            \set{\bigvee \class{\frame{\sigma^{\sharp}_j, \delta^{\sharp}_j} \concat S^{\sharp}_j}{\sim} 
            \cap  \set{\frame{\sigma^{\sharp}_i, \bot}
          \concat S^{\sharp}_i \mid \ i \in I} \mid j \in I} \right ) \\
          & = \bsem{\pop}^{\As}\set{\frame{\widehat{\sigma^{\sharp}_j}, \widehat{\delta^{\sharp}_j}} 
          \concat \widehat{S^{\sharp}_j} \mid j \in I} \\
          & = \set{\widehat{S^{\sharp}_j} \mid j \in I}
        \end{align*}
        where $\frame{\widehat{\sigma^{\sharp}_j}, \widehat{\delta^{\sharp}_j}} 
        \concat \widehat{S^{\sharp}_j} =  \bigvee \class{\frame{\sigma^{\sharp}_j, \delta^{\sharp}_j} \concat S^{\sharp}_j}{\sim} 
        \cap  \set{\frame{\sigma^{\sharp}_i, \delta_i^{\sharp}}
      \concat S^{\sharp}_i \mid \ i \in I}$ and actually for all $j \in I$ we have $\overline{S^{\sharp}_j} = \widehat{S^{\sharp}_j}$ 
      since the order $\leq$ is pointwise on frame components and on stacks (sequences), and whenever it happens that 
      $\frame{\sigma^{\sharp}_1, \delta^{\sharp}_1}\concat S^{\sharp}_1 \not\sim \frame{\sigma^{\sharp}_2, \delta^{\sharp}_2}\concat S^{\sharp}_2$
      but $S^{\sharp}_1 \sim S^{\sharp}_2$, then the $\bsem{\pop}^{\As}$ also takes their least upper bound, 
      yielding the same least upper bound we find if we remove the top element first. 
    \end{proof}

  \end{toappendix}

  \begin{theoremrep}[Basic abstract operations]
    \label{th:gen-BCAs}
    Let $\sim$ be a compatible equivalence on the abstract domain
    $A$. The BCAs of the basic expressions
    $\add$, $\reset$, $\push$, $\pop$ for the
    $\sim$-stack abstraction
    are as follows: for all
    $\frame{\sigma^{\sharp}, \delta^{\sharp}}\concat S^{\sharp}\in
    \str{\operatorname{F}_{A}}$
    \begin{center}
      \begin{tabular}{r@{\hspace{2pt}}c@{\hspace{2pt}}l
        r@{\hspace{2pt}}c@{\hspace{2pt}}l}
        $\bsem{\add}^{\As} \set{\frame{\sigma^{\sharp},
        \delta^{\sharp}} \concat S^{\sharp}}$ & $=$ &
        $\set{\frame{\sigma^{\sharp}, \delta^{\sharp} \vee
        \sigma^{\sharp}} \concat S^{\sharp}}$ &
        $\bsem{\reset}^{\As} \set{\frame{\sigma^{\sharp},
        \delta^{\sharp}} \concat S^{\sharp}}$ & $=$ &
        $\set{\frame{\sigma^{\sharp}, \bot} \concat S^{\sharp}}$ \\
        $\bsem{\push}^{\As} \set{\frame{\sigma^{\sharp},
        \delta^{\sharp}} \concat S^{\sharp}}$ & $=$ &
        $\set{\frame{\sigma^{\sharp}, \delta^{\sharp}} \concat\,
        \frame{\sigma^{\sharp}, \delta^{\sharp}} \concat S^{\sharp}}$ &
        $\bsem{\pop}^{\As} \set{\frame{\sigma^{\sharp},
        \delta^{\sharp}} \concat S^{\sharp}}$ & $=$ & $\set{S^{\sharp}}$
      \end{tabular}
    \end{center}
    and the operations $\reset$, $\push$, $\pop$ are
    globally complete. Moreover
    \begin{center}
      \begin{tabular}{r@{\hspace{2pt}}c@{\hspace{2pt}}l}
        $\bsem{\p}^{\As} \set{\frame{\sigma^{\sharp},
        \delta^{\sharp}} \concat S^{\sharp}}$ & $\leq_{\ssim}$ & 
        $\set{\frame{\sigma^{\sharp} \wedge \p^{\As},
        \delta^{\sharp}} \concat S^{\sharp}}$ \\ 
        $\bsem{\np}^{\As} \set{\frame{\sigma^{\sharp},
        \delta^{\sharp}} \concat S^{\sharp}}$ & $\leq_{\ssim}$ & 
        $\set{\frame{\sigma^{\sharp} \wedge \np^{\As},
        \delta^{\sharp}} \concat S^{\sharp}}$ \\ 
        $\bsem{\loops}^{\As} \set{\frame{\sigma^{\sharp},
        \delta^{\sharp}} \concat S^{\sharp}}$ & $\leq_{\ssim}$ & 
         $\set{\frame{\sigma^{\sharp},
      \delta^{\sharp}} \concat S^{\sharp} \mid \sigma^{\sharp}
    \wedge \delta^{\sharp} \neq \bot}$
      \end{tabular}
    \end{center}
    where $\p^{\As} = \alpha(\sem[]{\p}\Sigma)$ (resp. $\np^{\As} = \alpha(\sem[]{\np}\Sigma)$ ) is the abstraction of 
    the set of concrete states satisfying $p$ (resp $\neg p$).  %
  \end{theoremrep}

  For $\loops$  one can get
  $\bsem{\loops}^{\As} \set{\frame{\sigma^{\sharp}, \delta^{\sharp}}
    \concat S^{\sharp}} = \set{\frame{\sigma^{\sharp} \wedge
      \delta^{\sharp}, \delta^{\sharp}} \concat S^{\sharp} \mid
    \sigma^{\sharp} \wedge \delta^{\sharp} \neq \bot}$ under additional conditions on $\sim$ (see
  Definition~\ref{def:strong-eq} in the appendix).    

  \begin{proof}
    Corollary of Theorems~\ref{th:BCAs} and \ref{th:BCAs-equiv}. The
    property for $\p^{\As}$ and $\np^{\As}$ is a direct consequence of
    Remark~\ref{re:bca}.
  \end{proof}

\begin{toappendix}
  The next lemma shows that, besides providing a result which depends only on the top frame, programs associated with formulae operate as filters also at the abstract level.
  
    \begin{proposition}[Top-frame]
      \label{prop:top-frame}
      Let $\sim$ be a compatible equivalence on the abstract domain
  $A$. For all
  $\frame{\sigma^{\sharp}, \delta^{\sharp}}\concat S^{\sharp}\in
  \str{\operatorname{F}_{A}}$ and $\varphi$ formula, it holds 
  \begin{equation}\label{eq:strong-top-frame}
  \sem{\PtoR{\varphi}}^{\As}\set{\frame{\sigma^{\sharp}, \delta^{\sharp}}\concat S^{\sharp}} 
  = \begin{cases} 
    \varnothing \\
    \set{\frame{\sigma_1^{\sharp},\delta^{\sharp}}\concat S^{\sharp}} & \sigma_1^{\sharp} \leq \sigma^{\sharp}, \ \sigma_1^{\sharp} \sim \sigma^{\sharp}
  \end{cases}
  \end{equation}
\end{proposition}
\begin{proof}
We prove the statement 
by induction on the structure of $\varphi$. 
\paragraph*{ACTL.} Let us first note that the environment is inessential for ACTL formulae, so we omit it.
\subparagraph*{($\varphi = p$)}  By exploiting Lemma~\ref{lem:top-frame}, we can compute counterexamples to $p$ by looking at the top frame. 
\begin{align*}
  & \hspace{-3ex} \sem[]{\PtoR{p}}^{\As}\set{\frame{\sigma^{\sharp}, \delta^{\sharp}}} %
  = \sem[]{\np}\set{\frame{\sigma^{\sharp}, \delta^{\sharp}}} \\ 
  &  \hspace{-3ex} = \alpha_{\ssimplus}\dirIm{\alpha}\sem[]{\np}\invIm{\alpha}\gamma_{\ssimplus}\set{\frame{\sigma^{\sharp}, \delta^{\sharp}}}  \\ 
  &  \hspace{-3ex} = \alpha_{\ssimplus}\dirIm{\alpha}\sem[]{\np}\invIm{\alpha}\left \{ \frame{\sigma^{\flat}, \delta^{\flat}} 
  \mid \frame{\sigma^{\flat}, \delta^{\flat}}\sim \frame{\sigma^{\sharp}, \delta^{\sharp}} \right . \\ & \qquad \left . 
  \land \ \frame{\sigma^{\flat}, \delta^{\flat}}\ \leq \frame{\sigma^{\sharp}, \delta^{\sharp}}\right \}  \\
  &  \hspace{-3ex}= \alpha_{\ssimplus}\dirIm{\alpha}\sem[]{\np} \left \{ \frame{\sigma, \Delta} 
  \mid \alpha(\frame{\sigma, \Delta}) = \frame{\sigma^{\flat}, \delta^{\flat}} \
  \land  \ \frame{\sigma^{\flat}, \delta^{\flat}} \sim \frame{\sigma^{\sharp}, \delta^{\sharp}}
  \land \frame{\sigma^{\flat}, \delta^{\flat}} \leq \frame{\sigma^{\sharp}, \delta^{\sharp}} \right \}  \\
  &  \hspace{-3ex} = \alpha_{\ssimplus}\dirIm{\alpha} \left \{ \frame{\sigma, \Delta} 
  \mid \alpha(\frame{\sigma, \Delta}) = \frame{\sigma^{\flat}, \delta^{\flat}}, \ \sigma \models \neg p 
  \land \ \frame{\sigma^{\flat}, \delta^{\flat}} \sim \frame{\sigma^{\sharp}, \delta^{\sharp}}
  \land \frame{\sigma^{\flat}, \delta^{\flat}} \leq \frame{\sigma^{\sharp}, \delta^{\sharp}} \right \}  \\
  &  \hspace{-3ex}= \alpha_{\ssimplus} \left \{ \frame{\sigma^{\flat}, \delta^{\flat}} \mid \exists \sigma \in \invIm{\alpha}(\set{\sigma^{\flat}}). \ \sigma \models \neg p 
  \land \ \frame{\sigma^{\flat}, \delta^{\flat}} \sim \frame{\sigma^{\sharp}, \delta^{\sharp}}
  \land \frame{\sigma^{\flat}, \delta^{\flat}} \leq \frame{\sigma^{\sharp}, \delta^{\sharp}}\right \} \\
  &  \hspace{-3ex}= \set{\frame{\sigma_1^{\sharp}, \delta^{\sharp}} \mid \exists \sigma \in \invIm{\alpha}(\sigma_1^{\sharp}). \sigma \models \neg p}
\end{align*}
where $\sigma_1^{\sharp}$ is defined as
\begin{equation*}
  \sigma_1^{\sharp} = \bigvee_{\sigma^{\flat} \leq_{\ssim} \sigma^{\sharp} } \set{\sigma^{\flat} 
  \mid \exists \sigma \in \invIm{\alpha}(\set{\sigma^{\flat}}). \sigma \models \neg p} 
\end{equation*} 
hence $\sigma_1^{\sharp} \sim \sigma^{\sharp}$ by transitivity and $\sigma_1^{\sharp}\leq \sigma^{\sharp}$. 
The case $\neg p$ is analogous.
\subparagraph*{($\varphi = \varphi_1 \land \varphi_2$)} 
\begin{align*}
  & \sem[]{\PtoR{\varphi_1 \land \varphi_2}}^{\As}\set{\frame{\sigma^{\sharp}, \delta^{\sharp}}\concat S^{\sharp}} %
  = 
  \sem[]{\PtoR{\varphi_1} \oplus \PtoR{\varphi_2}}^{\As}\set{\frame{\sigma^{\sharp}, \delta^{\sharp}}\concat S^{\sharp}} \\
  & = \sem[]{\PtoR{\varphi_1}}^{\As}\set{\frame{\sigma^{\sharp}, \delta^{\sharp}}\concat S^{\sharp}} \vee
  \sem[]{\PtoR{\varphi_2}}^{\As}\set{\frame{\sigma^{\sharp}, \delta^{\sharp}}\concat S^{\sharp}} \\
  & \overset{\heartsuit}{=} \begin{cases} 
    \varnothing \\
    \set{\frame{\sigma_1^{\sharp},\delta^{\sharp}}\concat S^{\sharp}} & \sigma_1^{\sharp} \leq \sigma^{\sharp}, \ \sigma_1^{\sharp} \sim \sigma^{\sharp}
  \end{cases} \vee \ \begin{cases} 
    \varnothing \\
    \set{\frame{\sigma_2^{\sharp},\delta^{\sharp}}\concat S^{\sharp}} & \sigma_2^{\sharp} \leq \sigma^{\sharp}, \ \sigma_2^{\sharp} \sim \sigma^{\sharp}
  \end{cases}  \\
  & = \begin{cases}
    \varnothing \\ 
    \set{\frame{\sigma_1^{\sharp},\delta^{\sharp}}\concat S^{\sharp}} & \sigma_1^{\sharp} \leq \sigma^{\sharp}, \ \sigma_1^{\sharp} \sim \sigma^{\sharp} \\ 
    \set{\frame{\sigma_2^{\sharp},\delta^{\sharp}}\concat S^{\sharp}} & \sigma_2^{\sharp} \leq \sigma^{\sharp}, \ \sigma_2^{\sharp} \sim \sigma^{\sharp} \\
    \set{\frame{\sigma_3^{\sharp},\delta^{\sharp}}\concat S^{\sharp}} & \sigma_3^{\sharp} = \sigma_1^{\sharp} \vee \sigma_2^{\sharp} \leq \sigma^{\sharp}, 
    \ \sigma_3^{\sharp} \sim \sigma_1^{\sharp} \sim \sigma_2^{\sharp} \sim \sigma^{\sharp}
  \end{cases}
\end{align*}
where at $(\heartsuit)$ we used the inductive hypothesis for $\varphi_1$ and $\varphi_2$.
\subparagraph*{($\varphi = \varphi_1 \lor \varphi_2$)} 
\begin{align*}
  & \sem[]{\PtoR{\varphi_1 \lor \varphi_2}}^{\As}\set{\frame{\sigma^{\sharp}, \delta^{\sharp}}\concat S^{\sharp}} %
  =
  \sem[]{\PtoR{\varphi_1} ; \PtoR{\varphi_2}}^{\As}\set{\frame{\sigma^{\sharp}, \delta^{\sharp}}\concat S^{\sharp}} \\
  & = \sem[]{\PtoR{\varphi_2}}^{\As}\left (\sem[]{\PtoR{\varphi_1}}^{\As}\set{\frame{\sigma^{\sharp}, \delta^{\sharp}}\concat S^{\sharp}} \right) \\
  & \overset{\heartsuit}{=} \sem[]{\PtoR{\varphi_2}}^{\As} \begin{cases} 
    \varnothing \\
    \set{\frame{\sigma_1^{\sharp},\delta^{\sharp}}\concat S^{\sharp}} & \sigma_1^{\sharp} \leq \sigma^{\sharp}, \ \sigma_1^{\sharp} \sim \sigma^{\sharp}
  \end{cases} \\
  & \overset{\clubsuit}{=} \begin{cases} 
    \varnothing \\
    \set{\frame{\sigma_2^{\sharp},\delta^{\sharp}}\concat S^{\sharp}} & \sigma_2^{\sharp} \leq \sigma_1^{\sharp} \leq \sigma^{\sharp}, 
    \ \sigma_2^{\sharp} \sim \sigma_1^{\sharp} \sim \sigma^{\sharp}
  \end{cases} \\
\end{align*}
where at $(\heartsuit)$ we used the inductive hypothesis for $\varphi_1$ and at $(\clubsuit)$ for $\varphi_2$.
\subparagraph*{($\varphi = \AX \varphi_1$)} 
\begin{align*}
  & \sem[]{\PtoR{\AX \varphi_1}}^{\As}\set{\frame{\sigma^{\sharp}, \delta^{\sharp}}\concat S^{\sharp}} 
  = \sem[]{\push; \nexte; \PtoR{\varphi_1}; \pop}^{\As}\set{\frame{\sigma^{\sharp}, \delta^{\sharp}}\concat S^{\sharp}} \\
  & = \sem[]{\pop}^{\As}\left (\sem[]{\PtoR{\varphi_1}}^{\As}\left (\sem[]{\nexte}^{\As}
  \left (\set{\frame{\sigma^{\sharp}, \delta^{\sharp}}\concat \frame{\sigma^{\sharp}, \delta^{\sharp}}\concat S^{\sharp}} \right ) \right ) \right ) \\
  & \overset{\clubsuit}{=} \sem[]{\pop}^{\As}\left (\sem[]{\PtoR{\varphi_1}}^{\As}
  \left (\set{\frame{\sigma^{\flat}, \delta^{\sharp}}\concat \frame{\sigma^{\sharp}, \delta^{\sharp}}\concat S^{\sharp} 
  \mid \frame{\sigma^{\flat}, \delta^{\sharp}} \in \sem[]{\nexte}^{\As}\set{\frame{\sigma^{\sharp}, \delta^{\sharp}}}} \right ) \right ) \\
  & \overset{\heartsuit}{=} \sem[]{\pop}^{\As}
  \begin{cases} 
    \varnothing \\
    T \subseteq \set{\frame{\sigma_1^{\sharp},\delta^{\sharp}}\concat \frame{\sigma^{\sharp}, \delta^{\sharp}}\concat S^{\sharp}} 
    & \sigma_1^{\sharp} \leq \sigma^{\sharp}, \ \sigma_1^{\sharp} \sim \sigma^{\sharp}
  \end{cases} \\
  & = \begin{cases} 
    \varnothing \\
    \set{\frame{\sigma^{\sharp},\delta^{\sharp}}\concat S^{\sharp}}
  \end{cases} \\
\end{align*}
as desired, since $\sigma^{\sharp} \sim \sigma^{\sharp}$ and $\sigma^{\sharp} \leq \sigma^{\sharp}$. 
Note that at $(\heartsuit)$ we used the inductive hypothesis and at $(\clubsuit)$ we use Lemma~\ref{lem:top-frame}, 
noticing additionally that the operator $\nexte$ 
only changes the current state and leaves the trace $\delta^{\sharp}$ unchanged.
\subparagraph*{($\varphi = \AF \varphi_1$)} 
\begin{align*}
  \sem[]{\PtoR{\AF \varphi_1}}^{\As}\set{\frame{\sigma^{\sharp}, \delta^{\sharp}}\concat S^{\sharp}}
  & = \sem[]{\PtoR{\varphi_1}; \r}^{\As}
  \set{\frame{\sigma^{\sharp}, \delta^{\sharp}}\concat S^{\sharp}} \\
  & \overset{\heartsuit}{=} \sem[]{\r}^{\As}\begin{cases} 
    \varnothing \\
    \set{\frame{\sigma_1^{\sharp},\delta^{\sharp}} \concat S^{\sharp}} 
    & \sigma_1^{\sharp} \leq \sigma^{\sharp}, \ \sigma_1^{\sharp} \sim \sigma^{\sharp} \\ 
  \end{cases} 
\end{align*}
where$\r = \push; \reset; (\add; \nexte; \PtoR{\varphi_1})^*; \loops;  \pop$ and at $(\heartsuit)$ we use the inductive hypothesis. 
Now if the result is $\varnothing$, the computation will yield $\varnothing$. Hence we proceed as in the second case: 
\begin{align*}
  & \sem[]{\push; \reset; (\add; \nexte; \PtoR{\varphi_1})^*; \loops;  \pop}^{\As}\set{\frame{\sigma_1^{\sharp},\delta^{\sharp}} \concat S^{\sharp}} \\ 
  & \quad = \sem[]{\pop}^{\As}\left ( \sem[]{\loops}^{\As}\left ( \sem[]{(\add; \nexte; \PtoR{\varphi})^*}^{\As}\left ( 
      \set{\frame{\sigma_1^{\sharp},\bot} \concat \frame{\sigma_1^{\sharp},\delta^{\sharp}} \concat S^{\sharp}} \right ) \right ) \right ) %
\end{align*}
We obtain a subset $T$ of $\set{\frame{\sigma^{\sharp}_{k_i}, \delta^{\sharp}_{k_i}} \concat \frame{\sigma_1^{\sharp}, \delta^{\sharp}} \concat S^{\sharp} \mid 
\frame{\sigma^{\sharp}_{k_i}, \delta^{\sharp}_{k_i}} \in \sem[]{(\add; \nexte; \PtoR{\varphi})^k}^{\As}(\frame{\sigma_1^{\sharp}, \bot})}$, 
exploiting the definition of $\sem[]{\add}^{\As}$, the lemma ensuring the property of the $\nexte$ operator and the inductive hypothesis on $\PtoR{\varphi}$. 
Then, we filter $T$ with $\sem[]{\loops}^{\As}$, which is the BCA of a filter, hence has the top-frame dependence property, 
and finally we apply $\sem[]{\pop}^{\As}$, which let us conclude as in the $\AX$ case.
The cases $\AG \varphi_1$ and $\varphi_1 \AU \varphi_2$ are analogous. 
\paragraph*{$\mu$-calculus}
We first prove that the following weaker property holds for all closed formulae, assuming that 
the environment $\eta$ associates to each variable $\X$ 
a function that also satisfies the following property:
\begin{equation}\label{eq:weak}
\sem{\PtoR{\varphi}}^{\As}\set{\frame{\sigma^{\sharp}, \delta^{\sharp}}\concat S^{\sharp}} 
  = \begin{cases} 
    \varnothing \\
    \set{\frame{\sigma_1^{\sharp},\delta_1^{\sharp}}\concat S^{\sharp}} &
    \sigma_1^{\sharp} \leq \sigma^{\sharp}, \ \sigma_1^{\sharp} \sim \sigma^{\sharp}, \ \delta_1^{\sharp} \geq \delta^{\sharp}
  \end{cases}
\end{equation}
  \begin{itemize}
    \item The cases $\varphi = p$ and $\varphi = \neg p$ hold by the same argument as in the ACTL case, with $\delta_1^{\sharp} = \delta^{\sharp}$.
    \item The case $\varphi = \varphi_1 \land \varphi_2$ holds by the same argument as in the ACTL case, with $\delta_1^{\sharp}$ 
    as the least upper bound of the traces resulting from the two branches.
    \item The case $\varphi = \varphi_1 \lor \varphi_2$ holds by the same argument as in the ACTL case.
    \item The case $\varphi = x$ holds by the hypothesis on $\eta(\X)$. 
    \item The case $\varphi = \Box \varphi_1$ is analogous to the case $\AX \varphi_1$ in ACTL.
  \end{itemize}
\subparagraph*{($\varphi = \mu x . \varphi_x$)}
\begin{align*}
  \hspace{-10pt}\sem{\PtoR{\mu x . \varphi_x}}^{\As}\set{\frame{\sigma^{\sharp}, \delta^{\sharp}}\concat S^{\sharp}} & = 
  \sem{\push; \reset; \upmu \X . ( \loops \oplus \add; \PtoR{\varphi_x}); \pop}^{\As}\set{\frame{\sigma^{\sharp}, \delta^{\sharp}}\concat S^{\sharp}} \\
  & = \sem{\pop}^{\As}\left ( \sem{\upmu \X . ( \loops \oplus \add; \PtoR{\varphi_x})}^{\As}
  \set{\frame{\sigma^{\sharp}, \bot}\concat \frame{\sigma^{\sharp}, \delta^{\sharp}}\concat S^{\sharp}}  \right )  \\
\end{align*}
Let $F(f) = \sem[{\eta[\X \mapsto f]}]{\loops \oplus (\add; \PtoR{\varphi_1})}$, then if $f$ has property~\eqref{eq:weak} also $F(f)$ has it, 
since $\loops$ (BCA of a filter), $\add$ (computation) and $\PtoR{\varphi_1}$ (inductive hypothesis) have it, and the operators $\oplus$ and $;$ preserve it. 
Then the least fixpoint $\upmu \X. \r_x$ has the property~\eqref{eq:weak}.  The case $\varphi = \nu x . \varphi_x$ is straightforward  by inductive hypothesis.

By relying on the property~\eqref{eq:weak}, the fact that for any formula $\varphi$ in $\mu_{\Box}$-calculus 
also the strongest property~(\ref{eq:strong-top-frame}) holds, follows by a further  routine induction.
\end{proof}
\end{toappendix}

  The soundness-by-design of this abstract semantics entails a
  sound program verification.

  \begin{proposition}[Program verification for satisfaction]
    \label{pr:prg-ver-sat}
    Given a transition system
  $T = (\Sigma, \operatorname{I}, \mathbf{P}, \trel, \vdash)$,
    for all ACTL or $\mu_{\Box}$-calculus formulae 
    $\varphi$, abstract domain $A$, and a compatible equivalence $\sim$ on A, 
    if
    $\sem{\PtoR{\varphi}}^{\sharp}(\alpha^s_\ssim(I)) = \bot$ then, for all
    $\sigma \in I$, we have $\sigma \models \varphi$.
  \end{proposition}

  As in the concrete case, the abstract semantics
  of programs encoding formulae is a lower closure.  This,
  together with the observation that the semantics only depends on the
  top frame of a stack (see Proposition~\ref{prop:top-frame} in the
  appendix) is relevant for the concrete implementation.

  \begin{remark}
  \label{re:bca}
  \rm
  We recalled in~\S~\ref{sec:s4} that for any Boolean test $b$ and abstract domain $A$, 
  the identity function %
  $\lambda a\in A.\, a$ is %
  a correct approximation of the filtering semantics $\sem[]{b}$ of $b$. 
  This function can be enhanced to $\lambda a\in A.\, 
  \alpha(\sem[]{b}\Sigma) \wedge_A a$, which is also a correct approximation of $\sem[]{b}$~\cite[Section~4.1]{mine17}. 
  If $A$ is a partitioning abstraction 
  then one can show that 
  $\lambda a\in A.\, \alpha(\sem[]{b}\Sigma) \wedge_A a$ 
  is indeed the \emph{best} correct approximation of $\sem[]{b}$.
  However, this property does not hold in general, even assuming that $A$ is a disjunctive abstraction, i.e.,
  the additivity of $\gamma$ (see Example~\ref{ex:filter} in the appendix).
  \qed
  \end{remark}

  \begin{toappendix}
  \begin{example}[on the BCA for Boolean tests]
    \label{ex:filter}
    We provide an example of a disjuctive abstraction where the
    $\lambda a\in A.\, \alpha(\sem[]{b}\Sigma) \wedge_A a$ is bot the
    BCA of $\sem[]{b}$.
 
  This can be shown by considering the following 
  simple sign abstraction for an integer variable $x$: $\mathrm{Sign}_{\alpha,\gamma}\triangleq \{x\in \mathbb{Z},\, x\in \mathbb{Z}_{\leq 0},\, x\in \mathbb{Z}_{\geq 0},\, x\in \mathbb{Z}_{= 0},\, x \in \varnothing\}$, which, being closed under unions, is disjunctive, whereas $\mathrm{Sign}$ is not partitioning (e.g., $\mathrm{Sign}$ does not include the complement of $\mathbb{Z}_{\leq 0}$).  
Then, by considering the Boolean test $x< 0 $, we have that 
$\alpha (\sem[]{x< 0 }\mathbb{Z}) \wedge_{\mathrm{Sign}} \mathbb{Z}_{\geq 0} = \mathbb{Z}_{\leq 0} \wedge_{\mathrm{Sign}} \mathbb{Z}_{\geq 0}=
\mathbb{Z}_{= 0}$, whereas the best approximation is
$\alpha(\sem[]{x< 0  }\mathbb{Z}_{\geq 0}) = \alpha(\varnothing) = \varnothing$. 
  \qed
\end{example}
\end{toappendix}
  \begin{example}[Control flow graphs and abstract frames]\label{ex:cfg-abstraction}
    In our running example 
    , abstract frames are of the shape
    $A \times A = \mathbb{P}^N \times \mathbb{P}^N$. 
    The abstraction of singletons is $\alpha(\set{(n, \rho)}) = (n \mapsto a)$ 
    with $a(n') = \bot$ for all $n' \neq n$.
    With the aim of joining past states when they correspond to the
    same program point, we consider the equivalence
    $\sim \subseteq \mathbb{P}^N \times \mathbb{P}^N$ on abstract frames
    defined by $\sigma^{\sharp}_1 \sim \sigma^{\sharp}_2$ when
    $\supp{\sigma^{\sharp}_1}= \supp{\sigma^{\sharp}_2}$.
    The property
    $\psi = \Box^3 (n=3\ \to \nu x.\, (n=3\ \land\ \Box^4 x))$ from
    Example~\ref{ex:mu} holds in the system, and we can prove it with
    this abstraction.  Let
    $ \sigma^{\sharp} = (s \mapsto \top) = \alpha(\set{(s, xyzw)})$ we
    can compute $\sem{\PtoR{\psi}}^{\sharp}(\sigma^{\sharp}) = \bot$,
    implying that the formula holds from any initial states
    (convergence is after a single full iteration). 
    Instead, the abstract computation for $\varphi = \AG (n = e \to z = 0)$ from
    Example~\ref{ex:actl} yields a false positive. \qed
  \end{example}

  \section{Locally Complete Analyses}\label{sec:lcl}

  If the abstract interpretation of a $\MOKA$ program returns an
  alarm, i.e.,
  $\sem[]{\PtoR{\psi}}^{\sharp}\alpha^s_{\ssim}(I) \neq \bot$, then any initial state
  in the concretisation of
  $\sem[]{\PtoR{\psi}}^{\sharp}\alpha^s_{\ssim}(I)$ is a candidate
  counterexample to the validity of $\psi$.
  However,
  due to over-approximation, we cannot distinguish spurious
  counterexamples from true ones.
  Here, we discuss how to combine under- and
  over-approximation for the analysis of $\MOKA$ programs
  $\PtoR{\psi}$ %
  to overcome this problem. %
  In particular we leverage Local Completeness Logic
  (LCL)~\cite{DBLP:journals/jacm/BruniGGR23,BruniGGR21} possibly paired with
  Abstract Interpretation Repair (AIR)
  strategies~\cite{DBLP:conf/pldi/BruniGGR22}, to improve the
  analysis precision. %

  \subparagraph{LCL.}%
  Most abstract domains are not globally complete for program
  analysis, so that the corresponding analyses may well yield false
  alarms. Accordingly,~\cite{DBLP:journals/jacm/BruniGGR23,BruniGGR21}
  studies how completeness can be locally weakened to an analysis of
  interest:
  given  a function $f: C \to C$,
  an abstract domain $A_{\alpha,\gamma} \in \operatorname{Abs}(C)$ is
   \emph{locally complete} 
  \emph{on a value}
  $c \in C$, denoted by $\mathbb{C}_c^A(f)$, when $\alpha\circ f(c) =
  \alpha \circ f \circ \gamma\circ \alpha(c)$ (hence global completeness amounts to $\mathbb{C}_c^A(f)$ for all $c\in C$).
  Intuitively, the absence of false alarms in an abstract computation
  comes as a
  consequence of the local completeness of the abstract transfer
  functions on the traversed concrete states. 
  
  Moreover, local completeness
  is a convex property, and this allows to check
  local completeness on suitable under-approximations $u\leq c$ such
  that $\alpha(u)=\alpha(c)$, as $\mathbb{C}_u^A(f)$ implies
  $\mathbb{C}_c^A(f)$.
  The work~\cite{DBLP:journals/jacm/BruniGGR23,BruniGGR21} %
  implements this idea
  through LCL, an under-approximating
  program logic (in the style
  of incorrectness logic~\cite{DBLP:journals/pacmpl/OHearn20}), %
  parameterised by the abstract domain $A$ which provides %
  an over-approximation. 
  The LCL proof rules for $\kat$
  programs are recalled in
  Table~\ref{table:LCL}.
  A provable LCL triple $\vdash_A [P] \ \r \ [Q]$ %
  ensures that each state satisfying $Q$ is reachable from some state satisfying $P$,
  and also
  guarantees that $Q$ and $\sem[]{\r}P$ 
  have the same abstraction in $A$.
  This means that the LCL program logic
  is sound w.r.t.\ the following notion of validity.

  \begin{definition}[Valid LCL triples]\ Let $P, Q \in C$ and $\r$ be a
    $\kat$ program. A
    program triple $\vdash_A [P]\ \r \ [Q]$  is valid if
    $\mathbb{C}^A_P(\r)\ \wedge\ Q \leq \sem[]{\r}P\ \wedge\ \alpha(\sem[]{\r}P)=\alpha(Q)$.
  \end{definition}
  
  The condition $\mathbb{C}^A_P(\r)$ states that $A$ is locally complete for $\sem[]{\r}$ on $P$, while $Q \leq \sem[]{\r}P$ that $Q$ under-approximates $\sem[]{\r}P$ and  $\alpha(\sem[]{\r}P)=\alpha(Q)$ that $\gamma(\alpha(Q))$ over-approximates $\sem[]{\r}P$.
  \begin{table}[t]
    \centering
    \resizebox{\textwidth}{!}{
      \begin{tabular}{c c}

        $\inference {\mathbb{C}^A_P(\e)}{\vdash_A [P] \ \e
        \ [\sem[]{\e}P]}[(transfer)]$ &
        $\inference {
        P' \leq P \leq \gamma(\alpha(P')) &
                                                      Q \leq Q' \leq \gamma(\alpha(Q))\\
        \vdash_A [P'] \ \r \ [Q']}{\vdash_A [P] \ \r
        \ [Q]}[(relax)]$ \\[15pt]
        $\inference {\vdash_A [P] \ \mathsf{r_1} \ [W] &
        \vdash_A [W] \ \mathsf{r_2} \ [Q]}{\vdash_A [P]
        \ \mathsf{r_1; r_2} \ [Q]}[(seq)]$  &
        $\inference {\vdash_A [P] \ \mathsf{r_1} \ [Q_1] &
        \vdash_A [P] \ \mathsf{r_2} \ [Q_2]}{\vdash_A [P]
        \ \mathsf{r_1\oplus r_2} \ [Q_1 \lor Q_2]}[(join)]$ \\[15pt]
        $\inference {\vdash_A [P] \ \r \ [W] & \vdash_A [P \lor
        W] \ \mathsf{r^{*}}\ [Q]}{\vdash_A [P] \ \mathsf{r^{*}}
        \ [Q]}[(rec)]$ &
        $\inference {\vdash_A [P] \ \r \ [Q] & Q \leq
        \gamma(\alpha(P))}{\vdash_A [P] \ \mathsf{r^{*}} \ [P \lor Q]}[(iterate)]$\\[15pt]
      \end{tabular}
      }
      \caption{Rules of the Local Completeness Logic~\cite{DBLP:journals/jacm/BruniGGR23}.}
      \label{table:LCL}
    \end{table}

  \subparagraph{LCL for $\MOKA$.}
  The LCL proof 
  system has been defined for $\kat$ programs only, hence it is enough for $\MOKA$
  programs induced by ACTL formulae.
  For
  analysing general $\kaf$ and thus $\MOKA$ programs, we
  extend it to $\upmu$LCL including the
  inference rules in Table~\ref{table:muLCL},
  where $\subst{\r}{\X}{\s}$ denotes the
  capture-avoiding substitution of the free occurrences of $\X$ for $\s$ in $\r$
  and the term $\upmu^{n} \X. \r$ represents the $n$-th
  fixpoint
  approximant with the expected semantics 
    ($\sem{\upmu^{0} \X. \r} \triangleq \lambda x.\, \bot$
   and
  $\sem{\upmu^{n+1} \X. \r} \triangleq \sem[{\eta[\X \mapsto \sem{\upmu^{n}
      \X. \r}]}]{\r}$).
  The syntax $\upmu^{n} \X. \r$
  is used only in $\upmu$LCL
  derivations. 
  Intuitively, since $\upmu^0 \X. \r$ behaves as $\zero$, the unique
  valid judgement is $\vdash_A [P] \ \upmu^0 \X. \r \ [\bot]$.
  The rule $(\upmu^{+})$ serves just to unfold $\upmu^{n+1} \X
  . \r$ as many times as needed. The key rule is (fix) which can be
  used whenever the $n$-th
  approximant provides enough information: in case the premise
  holds, by local completeness we will have
  $\sem[]{\upmu^n \X. \r}^{\sharp} \alpha(P) = \alpha (Q)$.
  Since it is always the case that
  $\sem[]{\upmu^n \X. \r}^{\sharp} \alpha(P) \leq_A \sem[]{\upmu
  \X. \r}^{\sharp} \alpha(P)$, if the side condition
  $\sem[]{\upmu \X. \r}^{\sharp} \alpha(P) \leq_A \alpha(Q)$ holds,
  then local completeness for the fixpoint term $\upmu \X . \r$ can be inferred.
  It turns out that $\upmu$LCL is a  sound proof system
  (see Lemma~\ref{le:soundness} in the appendix).
  \begin{table}
  \[
    \resizebox{\textwidth}{!}{
      \inference{}{\vdash_A [P] \ \upmu^0 \X. \r \ [\bot]}[($\upmu^0$)]
      \quad
      \inference{\vdash_A [P] \ \subst{\r}{\X}{\upmu^n \X .\r}
      \ [Q]}{\vdash_A [P] \  \upmu^{n+1} \X . \r \ [Q]}[($\upmu^{+}$)]
      \quad
      \inference{
        \sem[]{\upmu \X. \r}^{\sharp} \alpha(P) \leq_A \alpha(Q)\\
      \vdash_A [P] \  \upmu^{n} \X . \r \ [Q]}
      {\vdash_A [P] \  \upmu \X . \r \ [Q] }[(fix)
      ]
    }
  \]
  \caption{Novel rules for $\upmu$LCL.}
  \label{table:muLCL}
\end{table}

  \begin{toappendix}
    \begin{lemma}[Soundness]
      \label{le:soundness}
      The LCL rules for the fixpoint are sound.
    \end{lemma}

    \begin{proof}
      We equivalently show that given an abstraction
      $A \in \operatorname{Abs}(C)$, $\r \in \kaf$,
      $P, Q \in C$ and $\vdash_A [P] \ \r \ [Q]$ then
      \begin{enumerate}
        \item $Q \leq \sem[]{\r}P$
        \item $\sem[]{\r}^{\sharp} \alpha (P) = \alpha(Q)$
      \end{enumerate}

      For rule ($\upmu^0$) condition (1) is
      $\bot \leq \sem[]{\upmu^0 \X. \r}P$ and thus it trivially
      holds. Similarly, (2)
      $\sem[]{\upmu^0 \X. \r}^{\sharp} \alpha (P) = \sem[]{\zero}^{\sharp} \alpha (P) = \bot_A = \alpha(\bot)$, just by using the
      definitions.

      \bigskip

      Concerning rule ($\upmu^{+}$), conditions (1) and (2) trivially
      follow by the observation that the commands $\upmu^{n+1} \X. \r$ and
      $r[\upmu^{n} \X. \r/\X]$ in the premise and conclusion,
      respectively, have the same semantics.

      \bigskip
      Let us consider the rule (fix). Let us assume that the premise
      $\vdash_A [P] \  \upmu^{n} \X . \r \ [Q]$ holds, i.e.,
      \begin{center}
        $Q \leq \sem[]{\upmu^n \X. \r} P$ ($\dag$) \qquad $\sem[]{\upmu^n \X.
        \r}^{\sharp} \alpha (P) = \alpha(Q)$ ($\ddag$)
      \end{center}

      Then condition (1) holds since
      \begin{align*}
        Q & \leq \sem[]{\upmu^n \X. \r} P & \mbox{[by ($\dag$)]}\\
        & \leq \sem[]{\upmu \X. \r} P   & \mbox{[since $\sem[]{\upmu \X. \r } =
        \bigvee_n \sem[]{\upmu^n \X. \r }$]}.
      \end{align*}
      Concerning (2)
      \begin{align*}
        \alpha(Q)
        & \geq_A \sem[]{\upmu \X. \r}^{\sharp} \alpha (P)
        & \mbox{[side condition of the rule]}\\
        & \geq_A \sem[]{\upmu^n \X. \r}^{\sharp} \alpha (P) &
        \mbox{[definition of approximant]}\\
        & = \alpha(Q) & \mbox{[by ($\ddag$)]}%
      \end{align*}
    \end{proof}
  \end{toappendix}

  \begin{remark}\rm
    When a least fixpoint appears
    in $\PtoR{\varphi}$, we can exploit the following 
    heuristic for choosing the
    value of $n$: if the abstract domain satisfies the Ascending
    Chain Condition (ACC) then, for every $P \in C$
    there is $n_P\in \nat$ such that
    $\sem[]{\upmu^{n_P} \X. \r}^{\sharp} \alpha(P) = \sem[]{\upmu
    \X. \r}^{\sharp} \alpha(P)$. Then, by taking any $n\geq n_P$ such
    that $\vdash_A [P] \  \upmu^{n} \X . \r \ [Q]$ is provable,
    the  condition $\sem[]{\upmu \X. \r}^{\sharp} \alpha(P) \leq_A
    \alpha(Q)$ is readily 
    satisfied, so that $\vdash_A [P] \  \upmu \X . \r
    \ [Q]$ is valid. More precisely, the 
    rule below is sound:
    \begin{equation*}
      \inference{\vdash_A [P] \  \upmu^{n} \X . \r \ [Q]
        \qquad
        \sem[]{\upmu^{n} \X. \r}^{\sharp}
        \alpha(P) = \sem[]{\upmu \X. \r}^{\sharp}
        \alpha(P)
      }{\vdash_A [P]
      \  \upmu \X . \r \ [Q]}[(afix)]
      \tag*{\qed}
    \end{equation*}
  \end{remark}

  The following result relating $\upmu$LCL proofs with
  validity of formulae follows as an easy consequence
  of~\cite[Corollary~5.6]{DBLP:journals/jacm/BruniGGR23} and
  Theorem~\ref{th:counterexamples-general}.

  \begin{corollary}[Precision]
   Let $T = (\Sigma, \operatorname{I}, \mathbf{P}, \trel, \vdash)$ be a transition system. For all ACTL or $\mu_{\Box}$-calculus
   formulae 
    $\varphi$, abstract domain $A$ and compatible equivalence $\sim$ on A, 
    if the triple
    $\vdash_{\As}~[I] \ \PtoR{\varphi} \ [Q]$ is derivable, then $Q
    \subseteq I$ and
    \begin{itemize}
      \item $Q = \varnothing$ if and only if for all $\sigma \in I$ we have
        $\sigma \models \varphi$;
      \item if $Q \neq \varnothing$ then for all $\sigma \in Q$ we have $\sigma
        \centernot \models \varphi$.
    \end{itemize}
  \end{corollary}
  \begin{example}[Control flow graphs and LCL derivations]
    As pointed out in
    Example~\ref{ex:cfg-abstraction} the computation of the program $\mathsf{m}_{\varphi}$,
    which encodes in $\MOKA$ the property $\varphi = \AG (n = e \to z = 0)$ (see Example~\ref{ex:actl}), 
    yields a false alarm when we use the abstract domain
    $A$ depicted in Figure~\ref{fig:predicate}.
    Consider as set of initial states $I_0 =
    \set{\frame{(s, 0100), \varnothing}, \frame{(s, 0011),
    \varnothing}}$,
    whose abstraction covers the set of possible
    initial states
    $\set{\frame{(s, xyzw), \varnothing} \mid x, y, z, w \in \ev}$,
    because
    $\alpha_{\ssim}^s(I_0) = (s \mapsto \top)$. 
    Since the property holds, trying to derive a triple $\vdash_{\As} [I_0] \ \mathsf{m}_{\varphi} \ [\varnothing]$ 
    leads to 
    the failure of some
    local completeness assumption, which can drive the refinement of $A$.
    The imprecision is found in the fourth iteration of the $\nexte$ operator, 
    in which the set $\set{(3, 1100),(3, 0111)}$ is abstracted to $(3, xyzw)$, 
    losing the relation between $p$ and $q$ 
    being either both valid or both invalid. Several domain repairs are possible. 
    Following~\cite{DBLP:conf/pldi/BruniGGR22} we can %
    repair the domain by adding the abstract point $q\to p$ (a more abstract repair than adding the element $q\leftrightarrow p$, as proposed 
    in~\cite{DBLP:conf/tacas/BallPR01}). 
    Here we can also exploit a different route, by refining the equivalence so
    that $\sigma^{\sharp}_1 \sim \sigma^{\sharp}_2$
    when $\supp{\sigma^{\sharp}_1} = \supp{\sigma^{\sharp}_2} \neq 3$. This intuitively 
    corresponds to abstract separately the states at program point $n=3$.
    In both cases in the new domain $A'$ we can show $\sem[]{\PtoR{\varphi}}^{\sharp} \alpha'^{s}_{\ssim}(I_0) = \bot$ 
    (see Example~\ref{ex:cfg-appendix} in the appendix for the details). 
    \qed
  \end{example}

\begin{toappendix}
  \begin{example}[Control flow graphs - Abstract frames, LCL
    derivation and repair]\label{ex:cfg-appendix}
    We exemplify the use of LCL for showing that
    the property $\varphi = \AG (n = e \to z = 0)$ from Example~\ref{ex:actl}, 
    encoded in $\MOKA$ by the program 
    $\mathsf{m}_{\varphi} = \push; \nexte^*; \mathsf{n = e ?}; \mathsf{z \neq 0 ?}; \pop$ 
    holds for our running
    example in Figure~\ref{fig:example}. We use the abstraction
    $A_{\alpha,\gamma}$ depicted in Figure~\ref{fig:predicate}. 
    The set of initial states is $I_0 =
    \set{\frame{(s, 0100), \varnothing}, \frame{(s, 0011),
    \varnothing}}$ that satisfy the properties $x = 0$, $y \neq
    z = w \leq 1$. Note that the abstraction of
    $I_0$ covers the set of possible
    initial states
    $\set{\frame{(s, xyzw), \varnothing} \mid x, y, z, w \in \ev}$,
    since
    $\dirIm{\alpha}(I_0) = (s \mapsto \top)$.
    
    Let us consider the triple
      $\vdash_{\As}~[I_0] \ \mathsf{m}_{\varphi}
      \ [\varnothing]$, which reads \begin{equation*}
        \vdash_{\As}~ [I_0] \ \push; \nexte^*; \mathsf{n = e ?}; \mathsf{z \neq 0 ?}; \pop \ [\varnothing]
      \end{equation*}
    As pointed out in Example~\ref{ex:cfg-abstraction} the computation of the program $\mathsf{m}_{\varphi}$ 
    yields a false alarm when using the abstract domain $A$. 
    Hence this triple cannot be inferred in LCL with this 
    abstract domain $\As$ due to the failure of some local completeness 
    assumptions, and, on the other hand, this lack of local completeness
    drives the refinement of $\As$.  

    The root of the
    derivation tree is the judgement above.
    Since, by Theorem~\ref{th:gen-BCAs}, $\push$, and $\pop$ are
  complete we
  can apply the
  rule $(\operatorname{seq})$ twice to infer the
  triple below, where, to improve readability, we abuse the notation by
  noticing that stacks are inessential and that computation paths
  are never exploited (because the commands $\add$ and $\loops$
  do not occur in the program), to write $\sigma$ in place of
  $\frame{\sigma, \Delta}\concat S$, e.g., we write $I_0 =
  \set{(s, 0100), (s, 0011)}$:
    \begin{equation*}
      \vdash_{\As}~[I_0] \ \nexte^*;
      \mathsf{n =  e ?}; \mathsf{z \neq 0?} \ [\varnothing]
    \end{equation*}
    To ease the notation, any variable $x,
    y, z, w$ free is inteded to be in $\ev$. 
    Moreover, let $T_i = \bigcup_{j \leq i} I_j$.

    The derivation begins as
    
    \begin{equation*}
      \resizebox{\textwidth}{!}{
        \inference{
          \inference{
            \inference{\mathbb{C}^{\As}_{I_0}(\nexte)}{\vdash_{\As}~
              [I_0]\ \nexte
            \ [I_1]}[{\normalfont (transfer)}]
            &
            \inference{\heartsuit}{\vdash_{\As}~[T_1]
              \ \nexte^* \ [R]}[{\normalfont (rec)}]
          }{\vdash_{\As}~
          [I_0]\ \nexte^*\ [R]}[{\normalfont (rec)}]
          &
          \inference{\spadesuit}{\vdash_{\As}~ [R] \ \mathsf{n =
          e?}; \mathsf{z \neq 0 ?} \ [\varnothing]}[{\normalfont (seq)}]
        }{\vdash_{\As} [I_0]\ \nexte^*; \mathsf{n =  e ?};
        \mathsf{\neq 0?} \ [\varnothing]}[{\normalfont (seq)}]
    }
    \end{equation*}
    where $R$ needs to be found to infer the triples $\vdash_{\As} [I_0] \ \nexte^*\ [R]$ 
    and $\vdash_{\As} [R] \ \mathsf{n =  e ?}; \mathsf{z \neq 0?} \ [\varnothing]$,
    $I_1 = \sem[]{\nexte}(I_0) = \set{(1, 1100), (1, 0001)}$.
    The local completeness hypothesis
    at $I_0$ on the domain $\As$, that we denote $\mathbb{C}^{\As}_{I_0}$, is
    satisfied since $\alpha^s_\ssim \circ \sem[]{\nexte} \circ \gamma^s_\ssim \circ \alpha^s_\ssim (\set{(s, 0100),
    (s, 0011)}) =
    \alpha^s_\ssim \circ \sem[]{\nexte} (\set{(s, 0100), (s, 0011)})$, because
    \begin{align*}
      \alpha^s_\ssim \circ \sem[]{\nexte} \circ \gamma^s_\ssim \circ \alpha^s_\ssim(\set{(s, 0100), (s, 0011)}) & =
      \alpha^s_\ssim \circ \sem[]{\nexte} (\set{(s, xyzw)}) \\
      & =\alpha^s_\ssim (\set{(1, yy0w)}) \\
      & = (1 \mapsto p \land q)
      & \\
      \alpha^s_\ssim \circ \sem[]{\nexte} (\set{(s, 0100), (s, 0011)}) & = \alpha^s_\ssim (\set{(1,
      1100), (1, 0001)}) \\ 
      & = (1 \mapsto p \land q)
    \end{align*}
    The branch which is most interesting
    to explore is now $(\heartsuit)$, in which we apply again the $(\operatorname{rec})$
    rule:
    \begin{equation*}
      \inference{
        \inference
        {\mathbb{C}^{\As}_{T_1}(\nexte)}
        {\vdash_{\As} [T_1] \ \nexte \ [I_2]}
        [{\normalfont(transfer)}]
        &
        \inference
        {\diamondsuit}
        {\vdash_{\As} [T_2] \ \nexte^* \ [R]}
        [{\normalfont (rec)}]
      }
      {\vdash_{\As} [T_1] \ \nexte^* \ [R]}[{\normalfont (rec)}]
    \end{equation*}
    where $I_2 = \sem[]{\nexte}{T_1} = I_1 \cup \set{(2, 0001), (3,
    1100)}$. Also the completeness hypothesis $\mathbb{C}^{\As}_{T_1}(\nexte)$
    is satisfied:
    \begin{align*}
      \alpha^s_\ssim  \circ \sem[]{\nexte} \circ \gamma^s_\ssim \circ \alpha^s_\ssim& (\set{(s, 0100), (s, 0011),
          (1,
      1100), (1, 0001)}) \\
      & = \alpha^s_\ssim \circ \sem[]{\nexte}(\set{(s, xyzw),
      (1, yy0w)}) \\
      & = \alpha^s_\ssim (\set{(1, yy0w), (2, y y 0 1), (3, yy00)}) \\
      & = (1 \mapsto p \land q, 2 \mapsto p \land q, 3 \mapsto p \land q) \\
      & \\
      \alpha^s_\ssim  \circ \sem[]{\nexte} & (\set{(s, 0100), (s, 0011), (1,
      1100), (1, 0001)})    \\
      & = \alpha^s_\ssim (\set{(1, 1100), (1, 0001), (2, 0001), (3, 1100)}) \\
      & = (1 \mapsto p \land q, 2 \mapsto p \land q, 3 \mapsto p \land q) %
    \end{align*}
    We proceed in the exploration of the $(\diamondsuit)$ branch:
    \begin{equation*}
      \inference{
        \inference
        {\mathbb{C}^{\As}_{T_2}(\nexte)}
        {\vdash_{\As} [T_2] \ \nexte \ [I_3]}
        [{\normalfont(transfer)}]
        &
        \inference
        {\clubsuit}
        {\vdash_{\As} [T_3] \ \nexte^* [R]}[{\normalfont(iterate)}]
      }{
        \vdash_{\As} [T_2] \ \nexte^* \ [R]
      }[{\normalfont (rec)}]
    \end{equation*}
    where $I_3 = \sem[]{\nexte}T_2 = I_2 \cup \set{(3, 0111), (e, 1100)}$. 
    Also the completeness hypothesis 
    $\mathbb{C}^{\As}_{T_2}(\nexte)$ holds, because
    \begin{align*}
      \alpha^s_\ssim \circ \sem[]{\nexte} \circ \gamma^s_\ssim \circ \alpha^s_\ssim & (\set{(s,
          0100), (s, 0011), (1, 1100), (1,0001), (2, 000 1), (3,
      1100)}) = \\
      & =  \alpha^s_\ssim \circ \sem[]{\nexte}(\set{(s, xyzw), (1, yy0w), (2, yy0w),
      (3, yy0w)}) \\
      & = \alpha^s_\ssim (\set{(1, yy0w), (2, yy0w), (3, yy0w), (3, y [y\! +\!1] 1 w),
      (e, yy0w)}) \\
      & = \set{(1 \mapsto p \land q, 2 \mapsto p \land q, 
      3 \mapsto \top, e \mapsto p \land q)}  \\
      \alpha^s_\ssim \circ \sem[]{\nexte} & (\set{(s,
          0100), (s, 0011), (1, 1100), (1,0001), (2, 000 1), (3,
      1100)})  = \\
      & = \alpha^s_\ssim (\set{(1, 1100), (1, 0001), (2, 0001), (3, 1100), (3,
      0111), (e, 1100)}) \\
      & = \set{(1 \mapsto p \land q, 2 \mapsto p \land q, 
      3 \mapsto \top, e \mapsto p \land q)}  \\
    \end{align*}
    Now the source of incompleteness is found in exploring the
    $(\clubsuit)$ branch, since the proof obligation for
    $\mathbb{C}^{\As}_{T_3}(\nexte)$ fails:
    {%
      \begin{equation*}
        \inference{
          \inference{\mathbb{C}^{\As}_{T_3}(\nexte)}
          {\vdash_{\As}
            [T_3] \ \nexte \ [I_4]
            }
            [{\normalfont (transfer)}] 
            &
            I_4 \leq A(T_3)
            }
            {\vdash_{\As} [T_3] \ \nexte^*
        \ [R]}
        [{\normalfont(iterate)}]
    \end{equation*}}
    where $I_4 = \sem[]{\nexte}T_3 = I_3 \cup \set{(s, 0111), (e, 1100)}$, we set $R \triangleq T_4$ and the
    condition $I_4 \leq A(T_3)$ guarantees that, thanks to the
    abstraction, we can stop the
    unfolding of the fixpoint, while the concrete computation would
    instead require $4k$ (where $k$ is the modulo in $\ev$) steps to reach a fixpoint.
    We have that 
    \begin{align*}
      \hspace{-10ex}\alpha^s_\ssim \circ  & \sem[]{\nexte}  \circ  \gamma^s_\ssim \circ \alpha^s_\ssim  (\left \{(s,0100), (s, 0011), (1,
      1100), (1, 0001), \right . \\ 
      & \qquad \qquad \qquad \left . 
      (2, 0001), (3, 1100), (3, 0111), (e, 1100)\right \}) = \\
      & = \alpha^s_\ssim \circ \sem[]{\nexte} (\set{(s, xyzw), (1, yy0w),
          (2, yy0w), (3,
      xyzw), (e, yy0w)}) \\
      & = \alpha^s_\ssim (\set{(1, yy0w), (2, yy0w), (3, yy0w), (3, y [y\! +\!1] 1 w),
      (e, yy0w), (s, x^{\neq y} z w), (e, yyzw)}) \\
      & = (s \mapsto \overline{q}, 1 \mapsto p \land q, 2 \mapsto p \land q, 3 \mapsto \top,
      e \mapsto q) \\
      & \\
      \hspace{-10ex} \alpha^s_\ssim \circ & \sem[]{\nexte}  (\left \{(s,0100), (s, 0011), (1,
      1100), (1, 0001),
      \right . \\ 
      & \qquad \quad \left .  
      (2, 0001), (3, 1100), (3, 0111), (e, 1100)\right \}) = \\
      & = \alpha^s_\ssim (\set{(1, 1100), (1, 0001), (2, 0001), (3, 1100), (e,
      1100), (3, 0 1 1 1), (s, 01 1 1)}) \\
      & = \set{(s, x^{\neq y} y z^{\neq 0} w), (1, yy0w), (2, yy0w), (3,
      xyzw), (e, yy0w)} \\ 
      & = (s \mapsto \overline{q} \land \overline{p}, 
      1 \mapsto p \land q, 2 \mapsto p \land q, 3 \mapsto \top,
      e \mapsto p \land q) %
    \end{align*}
    The source of imprecision is due to the fact that the abstraction $(3, xyzw)$ 
  is too coarse for the set $\set{(3, 1100),(3, 0111)}$, because the relation between $p$ and $q$ 
  being either both valid or both invalid is lost. Several domain repairs are possible. 
  In~\cite{DBLP:conf/pldi/BruniGGR22} it is shown that adding the abstract element $q\to p$ to the domain 
  $A$ offers a more abstract repair than adding the element $q\leftrightarrow p$, as originally proposed 
  in~\cite{DBLP:conf/tacas/BallPR01}. 
  In fact, letting $A_1=A\cup\set{q\leftrightarrow p}$, 
  the abstract interpreter computes $\sem[]{\PtoR{\varphi}}^{\sharp}{\alpha_1}^s_\ssim(I_0) 
  = \sem[]{\PtoR{\varphi}}^{\sharp}{(s \mapsto \top)} = \bot$, implying
  that the formula holds from any initial states.
 Here we can also exploit a different route, by refining the equivalence so
  that $\sigma^{\sharp}_1 \sim \sigma^{\sharp}_2$
  when $\supp{\sigma^{\sharp}_1} = \supp{\sigma^{\sharp}_2} \neq 3$.
  Intuitively, this corresponds to abstract separately the states at program point $n=3$.
  If we denote by $\approx$ be refined equivalence, 
  the proof obligation
  $\mathbb{C}^{\Aa}_{T_3}(\nexte)$ is now satisfied:
    \begin{align*}
      \hspace{-15ex} & \alpha^s_{\aapprox} \circ  \sem[]{\nexte}  \circ \gamma^s_{\aapprox} \circ \alpha^s_{\aapprox} (
        \left \{(s, 0100), (s, 0011), (1,
      1100), (1, 0001), \right . \\ & \left . 
      \qquad \qquad \qquad \qquad \qquad 
      (2, 0001), (3, 1100), (3, 0111), (e, 1100)\right \}) = \\
      & = \alpha^s_{\aapprox} \circ \sem[]{\nexte} (\set{(s, xyzw), (1, yy0w),
      (2, yy0w), (3,yy0w), (3, x^{\neq y} y z^{\neq 0} w), (e, yy0w)}) \\
      & = \alpha^s_{\aapprox} (\set{(1, yy0w), (2, yy0w), (3, yy0w), (3, y [y\! +\!1] 1 w),
      (e, yy0w), (s, x^{\neq y} z^{\neq 0} w), (e, yy0w)}) \\
      & = (s \mapsto
      \overline{p} \land \overline{q}, 1\mapsto p \land q, 2
      \mapsto p \land q,
      3 \mapsto p \land q, 3 \mapsto \overline{p} \land \overline{q},
    e \mapsto p \land q) \\
      & \\
      \hspace{-15ex} & \alpha^s_{\aapprox} \circ  \sem[]{\nexte}  (\left \{(s,0100), (s, 0011), (1,
      1100), (1, 0001), 
      \right . \\ & \left . 
      \qquad \qquad \qquad
      (2, 0001), (3, 1100), (3, 0111), (e, 1100)\right \}) = \\
      & = \alpha^s_{\aapprox} (\set{(1, 1100), (1, 0001), (2, 0001), (3, 1100), (e,
      1100), (3, 0 1 1 1), (s, 01 1 1)}) \\
     & = (s \mapsto
      \overline{p} \land \overline{q}, 1\mapsto p \land q, 2
      \mapsto p \land q,
      3 \mapsto p \land q, 3 \mapsto \overline{p} \land \overline{q},
    e \mapsto p \land q)
    \end{align*}
    and in fact one can compute $\sem[]{\PtoR{\varphi}}^{\sharp}\alpha^s_{\aapprox}(I_0) = \bot$.
    
    Then we can go back to the $(\spadesuit)$ branch to conclude:
    \begin{equation*}
      \resizebox{\textwidth}{!}{
      \inference{
        \inference
        {\mathbb{C}^{\Aa}_{R}(\mathsf{n =  e?})}
        {\vdash_{\Aa} [R] \ \mathsf{n =  e ?} \ [\set{(e, yy0w)}]}
        [{\tiny\normalfont(transfer)}]
        &
        \inference
        {\mathbb{C}^{\Aa}_{\set{(e, yy0w)}}(\mathsf{z \neq 0?})}
        {\vdash_{\Aa} [\set{(e, yy0w)}] \ \mathsf{z \neq 0?} \ [\varnothing]}
        [{\tiny\normalfont(transfer)}]
      }
      {
        \vdash_{\Aa} [R] \ \mathsf{n = e ?}; \mathsf{z \neq 0?} \ [\varnothing]
      }
      [{\tiny\normalfont (seq)}]
    }
    \end{equation*}
    where the proof obligations $\mathbb{C}^{\Aa}_{R}(\mathsf{n = e ?}),
    \ \mathbb{C}^{\Aa}_{\set{(e, yy0w)}}(\mathsf{z \neq 0?})$ are
    satisfied, since both properties $n = e$ and
    $z \neq 0$  are representable in the abstract domain.
    \qed
  \end{example}
\end{toappendix}
\begin{example}[Traffic light example, LCL derivation and repair,
  $\mu$-calculus version] \label{ex:sem-lcl} Consider the toy traffic
  light example from Figure~\ref{fig:semaphore-car} and the property
  $\psi = \AG (\mathsf{g} \to \AX \mathsf{d})$, briefly discussed in
  the introduction. We give some additional details in the light of
  the results presented. We first translate the formula in
  $\mu_{\Box}$-calculus as
  $\psi_{\mu} = \nu x . ((\mathsf{g} \rightarrow \Box \mathsf{d})
  \land \Box x)$ whose encoding is %
  $\PtoR{\psi_{\mu}} = \upmu \X . (\r_1 \oplus \push; \nexte; \X; \pop) \triangleq \upmu
  \X. \r_x$, where $\r_1\triangleq \mathsf{g?}; \push; \nexte;
  \mathsf{\neg d?}; \pop$. Then we try to derive the triple
  $\vdash_{\As} [\set{\mathsf{rs}}]\ \upmu \X. \ \r_x [\varnothing]$ in
  the abstract domain $A$ (see Figure~\ref{fig:semaphore-hasse})
  by applying the rule (afix) with $n = 2$ (the number of iterations
  needed for convergence of the abstract fixpoint in $A$). By definition, 
  $\upmu^0\X. \r_x = \zero$; then 
  the first approximant is $\upmu^1 \X. \r_x 
  = \r_1 \oplus \push; \nexte; \zero; \pop$,
  but we can omit the branch that contains $\zero$, so 
  $\upmu^1 \X. \r_x = \r_1$; and  
  the second approximant is 
  $\r_2\triangleq \upmu^2 \X. \r_x = \r_1 \oplus \push; \nexte; \r_1; \pop$.      
  The attempt to derive the triple $\vdash_{\As} [\set{\mathsf{rs}}]\ \r_2 \ [\varnothing]$ is sketched below, where the mid row gives an indication of the triples labelling the
  leaves of the derivation, reporting for each basic command involved, the
  pre-condition and post-condition of the triple, while the top row reports their abstractions.

\begin{center}
    \begin{tikzpicture}[xscale=0.7,yscale=0.7, every
      node/.style={font=\footnotesize}, every node/.style = {inner sep=0.8pt}]
      \node (vdash) {$\vdash_{\As}$};
      \node (rs) [base right = 2pt of vdash] {$[\mathsf{rs}]$};
      \node (mid) [right = 3pt of rs] {};
      \node (g1) [right = 3pt of mid] {$\mathsf{g?}$}; 
      \node (g1;) [base right = 0pt of g1] {$;$}; 
      \node (p1) [base right = 2pt of g1;] {$\push$}; 
      \node (p1;) [base right = 0pt of p1] {$;$}; 
      \node (n1) [base right = 2pt of p1;] {$\nexte$};
      \node (n1;) [base right = 0pt of n1] {$;$}; 
      \node (d1) [base right = 2pt of n1;] {$\neg \mathsf{d?}$};
      \node (d1;) [base right = 0pt of d1] {$;$}; 
      \node (p2) [base right = 2pt of d1;] {$\pop$};
      \node (p2;) [right = 0pt of p2] {};
      \node (init) [above = 0.25cm] at (mid.north) {\scriptsize $\left [ \arraycolsep=0.8pt \begin{array}{c}\mathsf{rs}\end{array} \right ]$};
      \node (ainit) [above = 0.1cm] at (init.north) {\small $\textcolor{red}{a}$};
      \node (v1)[above = 0.195cm] at (g1;.north) {$[\varnothing]$}; 
      \node (v2) [above = 0.195cm] at (p1;.north) {$[\varnothing]$};
      \node (v3)[above = 0.195cm] at (n1;.north) {$[\varnothing]$};  
      \node (v4) [above = 0.195cm] at (d1;.north) {$[\varnothing]$}; 
      \node (v5) [above = 0.29cm] at (p2;.north) {$[\varnothing]$}; 
      \node (bot1) [above = 0.14cm] at (v1.north) {\small $\bot$};
      \node (bot2) [above = 0.14cm] at (v2.north) {\small $\bot$};
      \node (bot3) [above = 0.14cm] at (v3.north) {\small $\bot$};
      \node (bot4) [above = 0.14cm] at (v4.north) {\small $\bot$};
      \node (bot5) [above = 0.14cm] at (v5.north) {\small $\bot$};
      \node (op)[base right = 5pt of p2] {$\oplus$};
      \node (op;) [right = 3pt of op] {};
      \node (p3) [right = 3pt of op;] {$\push$};
      \node (p3;) [base right = 0pt of p3] {$;$}; 
      \node (n2) [base right = 2pt of p3;] {$\nexte$};
      \node (n2;) [base right = 0pt of n2] {$;$}; 
      \node (g3) [base right = 7pt of n2;] {$\mathsf{g?}$};
      \node (g3;) [base right = 1pt of g3] {$;$}; 
      \node (p4) [base right = 2pt of g3;] {$\push$};
      \node (p4;) [base right = 0pt of p4] {$;$}; 
      \node (n3) [base right = 2pt of p4;] {$\nexte$};
      \node (n3;) [base right = 0pt of n3] {$;$}; 
      \node (d2) [base right = 2pt of n3;] {$\neg \mathsf{d?}$};
      \node (d2;) [base right = 0pt of d2] {$;$}; 
      \node (p5) [base right = 2pt of d2;] {$\pop$};
      \node (p5;) [base right = 0pt of p5] {$;$};
      \node (p6) [base right = 2pt of p5;] {$\pop$};
      \node (empty) [base right = 4pt of p6] {$[\varnothing]$};
      \node (init2) [above = 0.25cm] at (op;.north) {\scriptsize $\left [ \arraycolsep=0.8pt \begin{array}{c}\mathsf{rs}\end{array} \right ]$};
      \node (r1) [above = 0.18cm] at (p3;.north) {\scriptsize $\left [ \arraycolsep=0.8pt\begin{array}{c}\mathsf{rs}\\ \dots\end{array} \right ]$}; 
      \node (r2) [above = 0.18cm] at (n2;.north) {\scriptsize $\left [ \arraycolsep=1.2pt \begin{array}{c|c}\mathsf{rs} & \mathsf{gs} 
        \\ \dots & \dots \end{array} \right ]$}; 
      \node (r3) [above = 0.18cm] at (g3;.north) {\scriptsize $\left [ \arraycolsep=0.8pt\begin{array}{c}\mathsf{gs}\\ \dots \end{array} \right ]$}; 
      \node (r4) [above = 0.18cm] at (p4;.north) {\scriptsize $\left [ \arraycolsep=0.8pt\begin{array}{c}\mathsf{gs}\\ \dots \end{array} \right ]$}; 
      \node (e) [above = 0.20cm] at (n3;.north) {\textcolor{red}{\textbf{!}}}; 
      \node (a1) [above = 0.08cm] at (init2.north) {\small $\textcolor{red}{a}$};
      \node (a2) [above = 0.08cm] at (r1.north) {\small $\textcolor{red}{a}$};
      \node (a3)[above = 0.08cm] at (r2.north) {\small $a \lor c$};
      \node (a4) [above = 0.08cm] at (r3.north) {\small $\textcolor{blue}{c}$};
      \node (a4)[above = 0.08cm] at (r4.north) {\small $\textcolor{blue}{c}$};
    \end{tikzpicture}
    \end{center}

 \noindent
 While we can derive $\vdash_{\As} [\set{\mathsf{rs}}]\ \r_1 [\varnothing]$, the source of local incompleteness is found in the second branch, at the 
  \textcolor{red}{\textit{\textbf{!}}} position,
  where the proof obligation for
  $\mathbb{C}_{\set{\mathsf{gs}}}^{\As}(\nexte)$ fails, since
  $\alpha^s_{\ssim} \circ \sem[]{\nexte}(\set{\mathsf{gs}})=
  \alpha^s_{\ssim}(\set{\mathsf{gd, yd}}) = b \land c$, while
  $ \alpha^s_{\ssim} \circ \sem[]{\nexte} \circ \gamma^s_{\ssim} \circ \alpha^s_{\ssim}
  (\set{\mathsf{gs}}) = \alpha^s_{\ssim} \circ \sem[]{\nexte}
  \set{\mathsf{gs, gd, gs, ys}} = \alpha^s_{\ssim} (\set{\mathsf{gs, gd, gs, ys,
      rs}}) = a \lor c$.  We can thus repair the abstract domain $A$.
  Following the procedure  in~\cite{DBLP:conf/pldi/BruniGGR22}, we add the new abstract point $c_1$:
  \begin{center}
    $c_1 \triangleq \bigvee \set{T \subseteq \Sigma \mid T \subseteq
      \gamma \circ \alpha \set{\mathsf{gs}}, \ \sem[]{\nexte} T
      \subseteq \sem[]{\nexte}(\set{\mathsf{gs}})} = \set{\mathsf{gs,
        gd}}$
\end{center}
  
  The repaired domain is
  $A_1 \triangleq A \cup \set{c_1, b \land c_1}$, where
  $b \land c_1$ is also added because the domain must be closed under meets. Now, in $A_1$,  the proof obligation for
  $\mathbb{C}_{\set{\mathsf{gs}}}^{A_{1\ssim}^s}(\nexte)$ holds true:
  $\sem[]{\nexte}(\set{\mathsf{gs}}) = {\alpha_1}^s_{\ssim} \circ \sem[]{\nexte}
  \circ \gamma \circ {\alpha_1}^s_{\ssim} (\set{\mathsf{gs}}) = b \land c$. 
  In fact, in $A_1$ we just have
  $\sem[]{\PtoR{\psi}}^{\sharp}{\alpha_1}^s_{\ssim}(\set{\mathsf{rs}}) = \bot$ 
  (we omit the details for the first branch $\r_1$ that are as above): 

  \begin{center}
    \begin{tikzpicture}[xscale=0.7,yscale=0.7, every
      node/.style={font=\footnotesize}, every node/.style = {inner sep=0.8pt}]
      \node (vdash) {$\vdash_{A_{1\ssim}^s}$};
      \node (rs) [base right = 2pt of vdash] {$[\mathsf{rs}]$};
      \node (g1) [base right = 8pt of rs] {$\r_1$}; 
      \node (g1;) [base right = 2pt of g1] {}; 
      \node (init) [above = 0.25cm] at (mid.north) {\scriptsize $\left [ \arraycolsep=0.8pt \begin{array}{c}\mathsf{rs}\end{array} \right ]$};
      \node (ainit) [above = 0.1cm] at (init.north) {\small $\textcolor{red}{a}$};
      \node (v1)[above = 0.33cm] at (g1;.north) {$[\varnothing]$}; 
      \node (bot1) [above = 0.16cm] at (v1.north) {\small $\bot$};
      \node (op)[base right = 8pt of g1] {$\oplus$};
      \node (op;) [right = 4pt of op] {};
      \node (p3) [right = 4pt of op;] {$\push$};
      \node (p3;) [base right = 0pt of p3] {$;$}; 
      \node (n2) [base right = 2pt of p3;] {$\nexte$};
      \node (n2;) [base right = 0pt of n2] {$;$}; 
      \node (g3) [base right = 7pt of n2;] {$\mathsf{g?}$};
      \node (g3;) [base right = 1pt of g3] {$;$}; 
      \node (p4) [base right = 2pt of g3;] {$\push$};
      \node (p4;) [base right = 0pt of p4] {$;$}; 
      \node (n3) [base right = 2pt of p4;] {$\nexte$};
      \node (n3;) [base right = 0pt of n3] {$;$}; 
      \node (d2) [base right = 7pt of n3;] {$\neg \mathsf{d?}$};
      \node (d2;) [base right = 1pt of d2] {$;$}; 
      \node (p5) [base right = 2pt of d2;] {$\pop$};
      \node (p5;) [base right = 0pt of p5] {$;$};
      \node (p6) [base right = 2pt of p5;] {$\pop$};
      \node (p6;) [right = 0pt of p6] {};
      \node (empty) [base right = 3pt of p6] {$[\varnothing]$};
      \node (init2) [above = 0.24cm] at (op;.north) {\scriptsize $\left [ \arraycolsep=0.8pt \begin{array}{c}\mathsf{rs}\end{array} \right ]$};
      \node (r1) [above = 0.18cm] at (p3;.north) {\scriptsize $\left [ \arraycolsep=0.8pt\begin{array}{c}\mathsf{rs}\\ \dots\end{array} \right ]$}; 
      \node (r2) [above = 0.18cm] at (n2;.north) {\scriptsize $\left [ \arraycolsep=1.2pt \begin{array}{c|c}\mathsf{rs} & \mathsf{gs} 
        \\ \dots & \dots \end{array} \right ]$}; 
      \node (r3) [above = 0.18cm] at (g3;.north) {\scriptsize $\left [ \arraycolsep=0.8pt\begin{array}{c}\mathsf{gs}\\ \dots \end{array} \right ]$}; 
      \node (r4) [above = 0.18cm] at (p4;.north) {\scriptsize $\left [ \arraycolsep=0.8pt\begin{array}{c}\mathsf{gs}\\ \dots \end{array} \right ]$}; 
      \node (r5) [above = 0.18cm] at (n3;.north) {\scriptsize $\left [ \arraycolsep=1.2pt \begin{array}{c|c}\mathsf{gd} & \mathsf{yd} 
        \\ \dots & \dots \end{array} \right ]$}; 
        \node (v6)[above = 0.195cm] at (d2;.north) {$[\varnothing]$};  
        \node (v7) [above = 0.195cm] at (p5;.north) {$[\varnothing]$}; 
        \node (v8) [above = 0.29cm] at (p6;.north) {$[\varnothing]$}; 
      \node (a1) [above = 0.1cm] at (init2.north) {\small $\textcolor{red}{a}$};
      \node (a2) [above  = 0.1cm] at (r1.north) {\small $\textcolor{red}{a}$};
      \node (a3)[above = 0.1cm] at (r2.north) {\small $\a \lor c$};
      \node (a4) [above = 0.1cm] at (r3.north) {\small $c_1$};
      \node (a4)[above = 0.1cm] at (r4.north) {\small $c_1$};
      \node (a5) [above = 0.1cm] at (r5.north) {\small $b \land c$};
      \node (bot3) [above = 0.14cm] at (v6.north) {\small $\bot$};
      \node (bot4) [above = 0.14cm] at (v7.north) {\small $\bot$};
      \node (bot5) [above = 0.14cm] at (v8.north) {\small $\bot$};
    \end{tikzpicture} \qquad \qed
  \end{center}
\end{example}

  \section{Related Work}
  \label{sec:rw}
  Abstract interpretation and model checking have been applied to and
  combined with each other in several ways~\cite{DBLP:conf/tacas/ShohamG04,DBLP:reference/mc/JhalaPR18,DBLP:conf/cav/GuptaS05,DBLP:conf/cav/GrafS97,DBLP:journals/toplas/ClarkeGL94,DBLP:journals/tocl/BustanG03,DBLP:conf/tacas/BallPR01,DBLP:journals/logcom/RanzatoT07,DBLP:conf/sas/CousotGR07,DBLP:conf/cav/ClarkeGJLV00,DBLP:journals/jacm/ClarkeGJLV03,DBLP:conf/sas/SchmidtS98,DBLP:conf/sas/GiacobazziQ01,DBLP:journals/iandc/GiacobazziR06,DBLP:conf/esop/Ranzato01,DBLP:journals/toplas/DamsGG97,DBLP:journals/ase/CousotC99,DBLP:conf/sas/Masse02,DBLP:conf/sas/RanzatoT02,DBLP:conf/popl/CousotC00,DBLP:journals/iandc/GrumbergLLS07,DBLP:conf/concur/Baldan0P20,DBLP:journals/fmsd/LoiseauxGSBB95,schmidt,DBLP:conf/lics/LarsenT88,DBLP:conf/lpar/BandaG10}
  (see the survey~\cite{DBLP:reference/mc/DamsG18} for further references).
  Abstract model checking broadly refers to checking a (temporal)
  specification in
  an abstract rather than a concrete model. On the one hand, in state
  partition-based approaches, the abstract model is itself a
  transition system, possibly induced by a behavioural state
  equivalence such as (bi)simulation, so that the verification method
  for the abstract model remains unaltered, e.g.,~\cite{DBLP:conf/cav/ClarkeGJLV00,DBLP:journals/jacm/ClarkeGJLV03,DBLP:journals/toplas/ClarkeGL94,DBLP:journals/fmsd/LoiseauxGSBB95,DBLP:journals/tocl/BustanG03}.
  On the other hand, if the abstract model derives from
  any, possibly nonpartitioning, state abstraction then the verification
  on this abstract model
  relies on abstract interpretation, e.g.,~\cite{DBLP:conf/sas/GiacobazziQ01,DBLP:conf/lpar/BandaG10,DBLP:conf/sas/RanzatoT02,DBLP:journals/logcom/RanzatoT07,DBLP:conf/sas/CousotGR07,DBLP:conf/sas/Masse02}.
  Hybrid abstraction techniques, in between these two approaches,
  have also been studied, e.g., predicate abstraction~\cite{DBLP:reference/mc/JhalaPR18,DBLP:conf/cav/GrafS97}, and the
  mixed transition systems with
  several universal and/or existential abstract transition relations in~\cite{DBLP:journals/toplas/DamsGG97,BallKY05}.
  Moreover, several works investigated the role of
  complete/exact/strongly preserving abstractions  in model checking, e.g.,~\cite{DBLP:journals/fmsd/LoiseauxGSBB95,DBLP:conf/popl/CousotC00,DBLP:journals/iandc/GiacobazziR06,DBLP:conf/esop/Ranzato01},
  and how to refine abstract models, e.g.,~\cite{DBLP:conf/cav/GuptaS05,DBLP:journals/iandc/GrumbergLLS07,DBLP:conf/cav/ClarkeGJLV00,DBLP:journals/jacm/ClarkeGJLV03,DBLP:journals/ase/CousotC99,DBLP:conf/sas/CousotGR07,DBLP:conf/sas/GiacobazziQ01,DBLP:conf/sas/RanzatoT02,schmidt,DBLP:conf/tacas/ShohamG04}.
  Let us also mention that~\cite{DBLP:conf/concur/Baldan0P20} develops a
  theory of approximation for systems of fixpoint equations in the
  style of abstract interpretation, showing that up-to techniques can
  be interpreted as abstractions.

  Our refinement technique somehow resembles CounterExample-Guided
  Abstraction Refinement
  (CEGAR)~\cite{DBLP:journals/jacm/ClarkeGJLV03},
  a popular method for
  automatising the abstraction generation in model
  checking. CEGAR deals with state partition abstractions, thus
  merging sets of equivalent states: one starts from a
  rough abstraction which is iteratively refined on the basis of
  spurious counterexample traces arising due to
  over-approximation. The approach is sound for safety properties, i.e., no
  false positives can be found, and complete
  for a significant fragment
  of ACTL$^{*}$.
  It turns out that state partition abstractions are a specific
  instance in our approach. The
  attempt of proving the absence of counterexamples in LCL using an
  initial coarse abstraction will yield a computation by some means similar to
  CEGAR: failing LCL proof obligations lead to abstraction refinements
  which can be more general than partitions.

The general idea that model checking can be expressed as static analysis has been first investigated 
  in~\cite{NielsonN10}. 
  This work %
  shows that model checking of ACTL can be 
reduced to a logic-based static analysis formulated within the flow logic approach~\cite{NielsonN02},
which is then computed through a solver for the alternation-free least fixed point logic designed in~\cite{NielsonSN02}. 
This reduction technique has been later extended to the $\mu$-calculus in~\cite{ZhangNN12}.
One major goal of~\cite{NielsonN10} was to show the
close relationship and interplay between model checking and static analysis, 
coupled with early work in~\cite{Steffen91} %
and, later, in~\cite{Schmidt98,SchmidtS98,Steffen93},
proving that static program analysis can be reduced to model checking
of modal formulae.
Let us remark that the reduction of~\cite{NielsonN10,ZhangNN12} to a flow logic-based static analysis 
is given for concrete model checking only and 
does not encompass the chance of dealing with abstract model checking and related abstraction 
refinement techniques such as CEGAR, that can be instead achieved by-design in our approach.

\section{Conclusion and future work}
  \label{sec:conc}

  We have introduced a framework where model checking of temporal formulae
  in ACTL or in the universal fragment of the modal $\mu$-calculus can
  be reduced to program verification, paving the way to reuse 
  the full range of abstract interpretation techniques. 
  Formulae are mapped to programs
  of the $\MOKA$ language, that are then analysed through a
  sound-by-construction abstract interpretation. This exposes all the
  possible counterexamples, although false alarms can arise. We show
  how false alarms can be removed by inspecting the derivability of
  suitable judgements in LCL, a program logic exploiting under- and
  over-approximations leveraging the notion of locally
  complete abstract interpretation.
  We expect that our approach, relying on $\kaf$,
   naturally applies also to logics including operators based on
  regular expressions (see, e.g.,~\cite{Book:DynLog,BBL:RCTL}).  A first candidate is
  Propositional Dynamic Logic (PDL)~\cite{Book:DynLog}, a modal logic
  closely connected to $\kat$. Its distinctive feature is a modal
  operator $[\r] \varphi$ where $\r$ is a $\kat$ expression, which 
  is satisfied by a state $\sigma$ when all the computations of
  $\r$ starting from $\sigma$ end up in a state satisfying
  $\varphi$. For instance, the property expressed in the
  $\mu_\Box$-calculus in Example~\ref{ex:mu} can be equivalently
  written in PDL as $[\nexte]^3 (n=3 \to ([\nexte]^4)^* (n=3))$, where,
  since we work in an unlabelled setting, ``$\nexte$'' stands
  for the only action in the system. Indeed,  PDL smoothly fits in our
  setting (for space limitations, details are in Appendix~\ref{sec:pdl}).

  \subparagraph{Future Work.}  A number of interesting avenues of
  future research remain open. Our results are limited to universal
  fragments of temporal logics and finite state domains. A dual theory
  can be easily developed for existential fragments, focusing on the
  generation of witnesses rather than counterexamples. It would be
  interesting to combine the two
  approaches for dealing with universal and existential operators at
  the same time. Some ideas could come
  from~\cite{DBLP:journals/toplas/DamsGG97} that, for solving the
  problem, works with two different abstract transition relations.
  The restriction to finite state domains is due to the fact that
  our encoding of logical formulae into $\MOKA$ programs relies on the
  $\loops$ operator for detecting infinite traces which are
  identified with looping traces.
  Further work could overcome this restriction by considering an
  encoding that captures non-looping infinite traces in the concrete
  domain and by exploiting ACC domains for the abstraction.

  The use of LCL allows us to track the presence of false alarms back to
  the failure of local completeness proof obligations, which can be resolved by
  refining the abstract domain. This can be done at different levels: either
  refining the abstraction over states or refining the equivalence on
  abstract frames. A proper theory of refinements, possibly
  identifying optimal ways of patching the domain,
  is a matter of future investigations.

 We point out that providing a general bound for the complexity
    of the abstract model checking procedure is not straightforward,
    as it crucially depends on the choice of the abstract domain.
    Different domains may induce significantly different behaviors,
    especially for non-partitioning abstractions or when local
    completeness and domain refinement techniques are applied.  A
    precise complexity analysis tailored to specific domains is an
    interesting subject for future work.
  
  Concerning the automatisation, the abstract interpreter could be
  easily implementable leveraging the standard toolset of abstract
  interpretation. %
  The abstract interpreter should be instrumented to report, when the
  result is not $\bot$, an abstract counterexample trace %
  to check whether the counterexample is a false or true alarm.  Making
  refinements effective requires working in a class of domains where
  local refinements are representable, e.g., by predicate abstractions
  or state partition abstractions. Developing a theory of refinements within
  specific subclasses of domains is a direction of %
  future work.

  \begin{toappendix}
    \section{Propositional Dynamic Logic}
    \label{sec:pdl}
Propositional Dynamic Logic (PDL) is a modal logic that allows 
one to reason about propositional properties of programs~\cite{Book:DynLog}. 
As we did for the other logics, we restrict here to universal fragment, and we thus define 
three syntactic categories: the formulae $\varphi$ that we verify (universal formulae), 
the programs $\r$ that act on a transition system 
$(\Sigma, \operatorname{I}, \mathbf{P}, \trel, \vdash)$, 
and the formulae $\psi$ that can be nested inside programs (existential formulae). 
Thus,
given a set of atomic propositions $p\in \mathbf{P}$:
\begin{center}
    \begin{tabular}{r@{\hspace{2pt}}c@{\hspace{2pt}}l}
  $\varphi $ & $::=$ & $\true \mid  \false \mid p \mid \neg p\mid
  \varphi_1 \lor \varphi_2 \mid
  \varphi_1 \land \varphi_2 \mid
  [\r] \varphi_1$ \\
  $\r  $ & $::=$ & $  \nexte \mid \r_1; \r_2 \mid \r_1 \oplus \r_2 \ | 
  \ \r_1^* \mid \psi? $ \\
  $\psi$ & $ ::= $ & $  \true \mid \false \mid p \mid \neg p \mid 
  \psi_1 \lor \psi_2 \mid
  \psi_1 \land \psi_2 \mid 
  \langle \r \rangle \psi_1 $ 
    \end{tabular}
\end{center}
As usual, $p \to \varphi'$ is syntactic sugar for
$\neg p \lor \varphi'$.

\begin{definition}[PDL semantics]\label{PDL-semantics}
    Given a transition system
    $T = (\Sigma, \operatorname{I}, \mathbf{P}, \trel, \vdash)$, the
    semantics %
    $\sem[]{\zeta} \subseteq \Sigma$ of
    {\rm \text{PDL}} formulae, for $\zeta \in \set{\varphi, \psi}$ 
    and of $\sem[]{\r}: \Sigma \to \pow{\Sigma}$ 
    programs over $T$ is as follows:
    \begin{center}
      \begin{tabular}{r@{\hspace{2pt}}c@{\hspace{2pt}}l@{\hspace{40pt}}r@{\hspace{2pt}}c@{\hspace{2pt}}l}
        $\sem[]{\true}$ & $\triangleq$ & $\Sigma$ &
        $\sem[]{\false}$ & $\triangleq$ & $\varnothing$ \\
        $\sem[]{p}$ & $\triangleq$ & $\set{\sigma \in \Sigma \mid
        \sigma \vdash p}$ & $\sem[]{\neg p}$ & $\triangleq$ & $\set{\sigma \in \Sigma \mid
        \sigma \vdash \neg p}$\\ 
        $\sem[]{\zeta_1 \lor \zeta_2}$ & $\triangleq$ &
      $\sem[]{\zeta_1} \cup \sem[]{\zeta_2}$ &
      $\sem[]{\zeta_1 \land \zeta_2}$  & $\triangleq$ &
      $\sem[]{\zeta_1} \cap \sem[]{\zeta_2}$ \\
      $\sem[]{[\r_1]\varphi_1}$ & $\triangleq$ & $\set{\sigma \mid \sem[]{\r}\sigma \subseteq \sem[]{\varphi_1}}$ &
      $\sem[]{\langle\r_1\rangle\varphi_1}$ & $\triangleq$ & $\set{\sigma \mid \sem[]{\r}\sigma \cap \sem[]{\varphi_1} \neq \varnothing}$ \\
      $\sem[]{\nexte}(\sigma)$ & $\triangleq$ & $\set{\sigma' \mid \sigma \trel \sigma'}$ \\
      $\sem[]{\r_1 ; \r_2}(\sigma)$ & $\triangleq$ & $\sem[]{\r_2}\left ( \sem[]{\r_1}(\sigma)\right )$ &
      $\sem[]{\r_1 \oplus \r_2}(\sigma)$ & $\triangleq$ & $\sem[]{\r_1}(\sigma) \cup \sem[]{\r_2}(\sigma)$ \\
      $\sem[]{\mathsf{r_1^*}}(\sigma)$ & $\triangleq$ & $ \bigvee \set{\sem[]{\mathsf{r_1}}^k(\sigma) \mid k \in \N}$ &
      $\sem[]{\psi?}(\sigma)$ & $\triangleq$ & $ \set{\sigma} \cap \sem[]{\psi}$
      \end{tabular}
    \end{center}
  \end{definition}

\subparagraph{PDL as $\MOKA$.}
  To each ACTL formula $\varphi$ we assign a $\MOKA$ program
  as follows: to each formula $\varphi$ we associate a 
  program $\PtoR{\varphi}$ which computes \textit{counterexamples} to $\varphi$, while to 
  each formula $\psi$ we associate a program $\PtoE{\psi}$ that computes \textit{examples} to 
  $\psi$. Since the logic is based on $\kat$, the encoding of programs $\r$ is just the identity\footnote{In the sense that 
  it is the program on our domain of stacks with the same name, and we overload the symbols.}, 
  with $\PtoE{\psi?} = \PtoE{\psi}$. 
  \begin{center}
    \begin{tabular}{r@{\hspace{2pt}}c@{\hspace{2pt}}l@{\hspace{50pt}}r@{\hspace{2pt}}c@{\hspace{2pt}}l}
       $\PtoR{\,\true}$ & $\triangleq$ & $\zero$ &  $\PtoE{\,\true}$ & $\triangleq$ & $\one$ \\ 
       $\PtoR{\,\false}$ & $\triangleq$ & $\one$ & $\PtoE{\,\false}$ & $\triangleq$ & $\zero$ \\ 
      $\PtoR{\,p}$ & $\triangleq$ & $\np$ & $\PtoE{\,p}$ & $\triangleq$ & $\p$ \\
      $\PtoR{\,\neg p}$ & $\triangleq$ & $\p$ & $\PtoE{\,\neg p}$ & $\triangleq$ & $\np$ \\
      $\PtoR{\varphi_1\land \varphi_2}$ & $\triangleq$ &
      $\PtoR{\varphi_1} \oplus \PtoR{\varphi_2}$
      &  $\PtoE{\psi_1\land \psi_2}$ & $\triangleq$ &
      $\PtoE{\psi_1}; \PtoE{\psi_2}$
      \\
      $\PtoR{\varphi_1\lor \varphi_2}$ & $\triangleq$ & 
      $\PtoR{\varphi_1} ; \PtoR{\varphi_2}$ 
      &  $\PtoE{\psi_1\lor \psi_2}$ & $\triangleq$ & 
      $\PtoE{\psi_1} \oplus \PtoE{\psi_2}$ \\
      $\PtoR{[\r]\varphi_1}$ & $\triangleq$ & $\push; \PtoE{\r}; 
      \PtoR{\varphi_1}; \pop$ 
      & $\PtoE{\langle\r\rangle\psi_1}$ & $\triangleq$ & $\push; \PtoE{\r}; 
       \PtoE{\psi_1}; \pop$ \\
    \end{tabular}
  \end{center}

  \begin{theorem}[Formula satisfaction as program verification - PDL version]
    \label{th:counterexamples-pdl}
    Given 
      $T = (\Sigma, \operatorname{I}, \mathbf{P}, \trel, \vdash)$, for all frames $\frame{\sigma, \Delta}\concat S \in
    \str{\operatorname{F}_{\Sigma}}$,
    for any \emph{PDL} formula $\varphi$
    \begin{equation*}
      \sem[]{\PtoR{\varphi}} \set{\frame{\sigma, \Delta}\concat S} =
      \set{\frame{\sigma, \Delta}\concat S \mid \sigma \notin
        \sem[]{\varphi}}
    \end{equation*}
  \end{theorem}

  \begin{proof}
    In order to prove the previous theorem we first need to prove that: 
    for any existential formula $\psi$ and command $\r$, as defined above, and for $\set{\frame{\sigma, \Delta}\concat S}
    \in  \pow{\operatorname{F}_{\Sigma}}$, then
  \begin{align}
    \sem[]{\PtoE{\psi}}\set{\frame{\sigma, \Delta}\concat S} & = 
      \set{\frame{\sigma, \Delta}\concat S \mid \sigma \in \sem[]{\psi}} \label{eq:example1} \\
      \sem[]{\PtoE{\r}}\set{\frame{\sigma, \Delta}\concat S} & = 
      \set{\frame{\sigma', \Delta}\concat S \mid \sigma' \in \sem[]{\r}(\sigma)} \label{eq:examples}
  \end{align}

  We proceed by induction on the structure of $\psi$ and $\r$. For $\psi$, 
all cases are straightforward with the exception of $\psi = \langle\r\rangle\psi_1$. 
We assume the inductive hypotheses:
\begin{align*}
  \forall \sigma,\Delta,S~\sem[]{\PtoE{\r_i}}
  \set{\frame{\sigma,
  \Delta}\concat S} & = \set{\frame{\sigma', \Delta}\concat S \mid
  \sigma'\in \sem[]{\r_i}(\sigma)} \\ 
  \sem[]{\PtoE{\psi_i}}
  \set{\frame{\sigma,
  \Delta}\concat S} & = \set{\frame{\sigma, \Delta}\concat S \mid
  \sigma\in \sem[]{\psi_i}} %
\end{align*} 
\begin{align*}
  \sem{\PtoE{\langle\r\rangle\psi_1}}\set{\frame{\sigma, \Delta}\concat S}
  & =
  \sem{\push;\PtoE{\r};\PtoE{\psi_1};\pop}\set{\frame{\sigma,
  \Delta}\concat S} \\
  & =
  \sem{\PtoE{\r};\PtoE{\psi_1};\pop}\set{
    \frame{\sigma,\Delta}\concat\,
    \frame{\sigma,\Delta}\concat S} \\
  & = \sem[]{\PtoE{\psi_1};\pop}\left ( \sem[]{\PtoE{\r}}\set{
    \frame{\sigma,\Delta}\concat\,
    \frame{\sigma,\Delta}\concat S} \right ) \\
  & \overset{\heartsuit}{=} \sem[]{\PtoE{\psi_1};\pop}\set{
    \frame{\sigma',\Delta}\concat\,
    \frame{\sigma,\Delta}\concat S \mid \sigma' \in \sem[]{\r}(\sigma)}\\
  & \overset{\clubsuit}{=} \sem[]{\pop}\set{
    \frame{\sigma',\Delta}\concat\,
    \frame{\sigma,\Delta}\concat S \mid \sigma' \in \sem[]{\r}(\sigma), \sigma' \in \sem[]{\psi_1}} \\ 
  & = \sem[]{\pop}\set{
    \frame{\sigma',\Delta}\concat\,
    \frame{\sigma,\Delta}\concat S \mid \sigma' \in \sem[]{\r}(\sigma) \cap \sem[]{\psi_1}} \\ 
    & = \set{\frame{\sigma,\Delta}\concat S \mid \exists \sigma' \in \sem[]{\r}(\sigma) \cap \sem[]{\psi_1}} %
\end{align*}
where at $(\heartsuit)$ we used the inductive hypothesis for $\r$ and at $(\clubsuit)$ we used the inductive hypothesis for $\psi_1$. 

Consider now the $\r$ cases: 
\subparagraph*{($\r=\nexte$)} By definition: \begin{align*}
  \sem[]{\PtoE{\nexte}}\set{\frame{\sigma, \Delta}\concat S} 
  & = \sem[]{\nexte}\set{\frame{\sigma, \Delta}\concat S} \\ 
  & = \set{\frame{\sigma', \Delta}\concat S \mid \sigma \trel \sigma'} \\
  & = \set{\frame{\sigma', \Delta}\concat S \mid \sigma' \in \sem[]{\nexte}(\sigma)}
\end{align*}
\subparagraph*{($\r=\r_1; \r_2$)} 
\begin{align*}
  \sem[]{\PtoE{\r_1; \r_2}}\set{\frame{\sigma, \Delta}\concat S} 
  & = \sem[]{\r_1; \r_2}\set{\frame{\sigma, \Delta}\concat S} \\ 
  & = \sem[]{\r_2} \left ( \sem[]{\r_1}\set{\frame{\sigma, \Delta}\concat S} \right ) \\
  & \overset{\heartsuit}{=} \sem[]{\r_2}\set{\frame{\sigma', \Delta}\concat S \mid \sigma'\in \sem[]{\r_1}(\sigma)} \\
  & \overset{\clubsuit}{=} \set{\frame{\sigma'', \Delta}\concat S \mid \sigma''\in \sem[]{\r_2}(\sem[]{\r_1}(\sigma))} \\
  & = \set{\frame{\sigma'', \Delta}\concat S \mid \sigma''\in \sem[]{\r_1; \r_2}(\sigma)} %
\end{align*}
where at $(\heartsuit)$ we used the inductive hypothesis for $\r_1$ and
at $(\clubsuit)$ we used the inductive hypothesis for $\r_2$.
\subparagraph*{($\r=\r_1 \oplus \r_2$)} Analogous to the case $\r=\r_1; \r_2$. 
\subparagraph{($\r = \r_1^*$)} Follows from the fact that $\r_1^{k+1} = \r_1; \r_1^{k}$. 
\subparagraph{($\r = \psi?$)} 
\begin{align*}
  \sem[]{\PtoE{\psi?}}\set{\frame{\sigma, \Delta}\concat S} 
  &\overset{\clubsuit}{=} \sem[]{\PtoE{\psi}}\set{\frame{\sigma, \Delta}\concat S} \\ 
  & = \set{\frame{\sigma, \Delta}\concat S \mid \sigma\in \sem[]{\psi}} \\
  & = \set{\frame{\sigma, \Delta}\concat S \mid \sigma\in \set{\sigma} \cap \sem[]{\psi}} \\
  & = \set{\frame{\sigma, \Delta}\concat S \mid \sigma\in \sem[]{\psi?}(\sigma)} %
\end{align*}
where at $(\clubsuit)$ we used the inductive hypothesis.

\bigskip
  We can now get back to the proof of the theorem:
  \begin{equation*}
    \sem[]{\PtoR{\varphi}} \set{\frame{\sigma, \Delta}\concat S} =
    \set{\frame{\sigma, \Delta}\concat S \mid \sigma \notin
    \sem[]{\varphi}}
  \end{equation*}
  We proceed by structural induction on $\varphi$. All cases are straightforward with 
  the exception of $\varphi = [\r]\varphi_1$.  
  We assume the inductive hypothesis:
      \begin{equation*}
        \forall \sigma,\Delta,S.~\sem{\PtoR{\varphi_1}}
        \set{\frame{\sigma,
        \Delta}\concat S} = \set{\frame{\sigma, \Delta}\concat S \mid
        \sigma \notin \sem[]{\varphi_1}}
      \end{equation*}
      \begin{align*}
        \sem{\PtoR{[\r]\varphi_1}}\set{\frame{\sigma, \Delta}\concat S}
        & =
        \sem{\push;\PtoE{\r};\PtoR{\varphi_1};\pop}\set{\frame{\sigma,
        \Delta}\concat S} \\
        & =
        \sem{\PtoE{\r};\PtoR{\varphi_1};\pop}\set{
          \frame{\sigma,\Delta}\concat\,
          \frame{\sigma,\Delta}\concat S} \\
        & = \sem[]{\PtoR{\varphi_1};\pop}\left ( \sem[]{\PtoE{\r}}\set{
          \frame{\sigma,\Delta}\concat\,
          \frame{\sigma,\Delta}\concat S} \right ) \\
        & \overset{\heartsuit}{=} \sem[]{\PtoR{\varphi_1};\pop}\set{
          \frame{\sigma',\Delta}\concat\,
          \frame{\sigma,\Delta}\concat S \mid \sigma' \in \sem[]{\r}(\sigma)}\\
        & \overset{\clubsuit}{=} \sem[]{\pop}\set{
          \frame{\sigma',\Delta}\concat\,
          \frame{\sigma,\Delta}\concat S \mid \sigma' \in \sem[]{\r}(\sigma), \sigma' \notin \sem[]{\varphi_1}} \\ 
          & = \set{\frame{\sigma,\Delta}\concat S \mid \exists \sigma' \in \sem[]{\r}(\sigma), \sigma' \notin \sem[]{\varphi_1}} \\ 
          & = \set{\frame{\sigma,\Delta}\concat S \mid \sem[]{\r}(\sigma) \centernot\subseteq \sem[]{\varphi_1}} 
      \end{align*}    
      where $(\heartsuit)$ is justified by~\eqref{eq:examples} and at $(\clubsuit)$ we used the inductive hypothesis.
\end{proof}
\end{toappendix}

  \bibliography{biblio}

\begin{thebibliography}{10}

\bibitem{DBLP:conf/concur/Baldan0P20}
Paolo Baldan, Barbara K\"onig, and Tommaso Padoan.
\newblock Abstraction, up-to techniques and games for systems of fixpoint
  equations.
\newblock In {\em Proceedings of CONCUR 2020}, volume 171 of {\em LIPIcs},
  pages 25:1--25:20, 2020.

\bibitem{BallKY05}
Thomas Ball, Orna Kupferman, and Greta Yorsh.
\newblock Abstraction for falsification.
\newblock In {\em Proceedings of CAV 2005}, volume 3576 of {\em LNCS}, pages
  67--81. Springer, 2005.

\bibitem{DBLP:conf/tacas/BallPR01}
Thomas Ball, Andreas Podelski, and Sriram~K. Rajamani.
\newblock Boolean and {C}artesian abstraction for model checking {C} programs.
\newblock In {\em Proceedings of TACAS 2001}, volume 2031 of {\em LNCS}, pages
  268--283. Springer, 2001.

\bibitem{DBLP:conf/lpar/BandaG10}
Gourinath Banda and John~P. Gallagher.
\newblock Constraint-based abstract semantics for temporal logic: {A} direct
  approach to design and implementation.
\newblock In {\em Proceedings of LPAR 2016}, volume 6355 of {\em LNCS}, pages
  27--45. Springer, 2010.

\bibitem{BBL:RCTL}
Ilan Beer, Shoham Ben{-}David, and Avner Landver.
\newblock On-the-fly model checking of {RCTL} formulas.
\newblock In {\em {CAV}}, volume 1427 of {\em LNCS}, pages 184--194. Springer,
  1998.

\bibitem{BruniGGR21}
Roberto Bruni, Roberto Giacobazzi, Roberta Gori, and Francesco Ranzato.
\newblock A logic for locally complete abstract interpretations.
\newblock In {\em Proceedings of LICS 2021}, pages 1--13. {IEEE}, 2021.

\bibitem{DBLP:conf/pldi/BruniGGR22}
Roberto Bruni, Roberto Giacobazzi, Roberta Gori, and Francesco Ranzato.
\newblock Abstract interpretation repair.
\newblock In {\em Proceedings of PLDI 2022}, pages 426--441. {ACM}, 2022.

\bibitem{DBLP:journals/jacm/BruniGGR23}
Roberto Bruni, Roberto Giacobazzi, Roberta Gori, and Francesco Ranzato.
\newblock A correctness and incorrectness program logic.
\newblock {\em Journal of the ACM}, 70(2):15:1--15:45, 2023.

\bibitem{DBLP:journals/tocl/BustanG03}
Doron Bustan and Orna Grumberg.
\newblock Simulation-based minimization.
\newblock {\em ACM Transactions on Computational Logic}, 4(2):181--206, 2003.

\bibitem{DBLP:conf/cav/ClarkeGJLV00}
Edmund~M. Clarke, Orna Grumberg, Somesh Jha, Yuan Lu, and Helmut Veith.
\newblock Counterexample-guided abstraction refinement.
\newblock In {\em Proceedings of CAV 2000}, volume 1855 of {\em LNCS}, pages
  154--169. Springer, 2000.

\bibitem{DBLP:journals/jacm/ClarkeGJLV03}
Edmund~M. Clarke, Orna Grumberg, Somesh Jha, Yuan Lu, and Helmut Veith.
\newblock Counterexample-guided abstraction refinement for symbolic model
  checking.
\newblock {\em Journal of the ACM}, 50(5):752--794, 2003.

\bibitem{CGKPV:MC-book}
Edmund~M. Clarke, Orna Grumberg, Daniel Kroening, Doron~A. Peled, and Helmut
  Veith.
\newblock {\em Model checking, 2nd Edition}.
\newblock {MIT} Press, 2018.

\bibitem{DBLP:journals/toplas/ClarkeGL94}
Edmund~M. Clarke, Orna Grumberg, and David~E. Long.
\newblock Model checking and abstraction.
\newblock {\em ACM Transactions on Programming, Languages and Systems},
  16(5):1512--1542, 1994.

\bibitem{DBLP:conf/laser/ClarkeKNZ11}
Edmund~M. Clarke, William Klieber, Milos Nov{\'{a}}cek, and Paolo Zuliani.
\newblock Model checking and the state explosion problem.
\newblock In {\em Tools for Practical Software Verification, LASER,
  International Summer School 2011}, volume 7682 of {\em LNCS}, pages 1--30.
  Springer, 2011.

\bibitem{cousot21}
Patrick Cousot.
\newblock {\em Principles of Abstract Interpretation}.
\newblock MIT Press, 2021.

\bibitem{DBLP:conf/popl/CousotC77}
Patrick Cousot and Radhia Cousot.
\newblock Abstract interpretation: {A} unified lattice model for static
  analysis of programs by construction or approximation of fixpoints.
\newblock In {\em Proceedings of POPL 1977}, pages 238--252. {ACM}, 1977.

\bibitem{DBLP:conf/popl/CousotC79}
Patrick Cousot and Radhia Cousot.
\newblock Systematic design of program analysis frameworks.
\newblock In {\em Proceedings of POPL 1979}, pages 269--282. {ACM} Press, 1979.

\bibitem{DBLP:journals/ase/CousotC99}
Patrick Cousot and Radhia Cousot.
\newblock Refining model checking by abstract interpretation.
\newblock {\em Automated Software Engineering}, 6(1):69--95, 1999.

\bibitem{DBLP:conf/popl/CousotC00}
Patrick Cousot and Radhia Cousot.
\newblock Temporal abstract interpretation.
\newblock In {\em Proceedings of POPL 2000}, pages 12--25. {ACM}, 2000.

\bibitem{DBLP:conf/sas/CousotGR07}
Patrick Cousot, Pierre Ganty, and Jean{-}Fran{\c{c}}ois Raskin.
\newblock Fixpoint-guided abstraction refinements.
\newblock In {\em Proceedings of SAS 2007}, volume 4634 of {\em LNCS}, pages
  333--348. Springer, 2007.

\bibitem{DBLP:journals/toplas/DamsGG97}
Dennis Dams, Rob Gerth, and Orna Grumberg.
\newblock Abstract interpretation of reactive systems.
\newblock {\em ACM Transactions on Programming, Languages and Systems},
  19(2):253--291, 1997.

\bibitem{DBLP:reference/mc/DamsG18}
Dennis Dams and Orna Grumberg.
\newblock Abstraction and abstraction refinement.
\newblock In Edmund~M. Clarke, Thomas~A. Henzinger, Helmut Veith, and Roderick
  Bloem, editors, {\em Handbook of Model Checking}, pages 385--419. Springer,
  2018.

\bibitem{dp:lattices-order}
Brian~A. Davey and Hilary~A. Priestley.
\newblock {\em Introduction to Lattices and Order}.
\newblock Cambridge University Press, 2002.

\bibitem{DBLP:conf/sas/GiacobazziQ01}
Roberto Giacobazzi and Elisa Quintarelli.
\newblock Incompleteness, counterexamples, and refinements in abstract
  model-checking.
\newblock In {\em Proceedings of SAS 2001}, volume 2126 of {\em LNCS}, pages
  356--373. Springer, 2001.

\bibitem{DBLP:journals/iandc/GiacobazziR06}
Roberto Giacobazzi and Francesco Ranzato.
\newblock Incompleteness of states w.r.t. traces in model checking.
\newblock {\em Information and Computation}, 204(3):376--407, 2006.

\bibitem{DBLP:journals/jacm/GiacobazziRS00}
Roberto Giacobazzi, Francesco Ranzato, and Francesca Scozzari.
\newblock Making abstract interpretations complete.
\newblock {\em Journal of the ACM}, 47(2):361--416, 2000.

\bibitem{DBLP:conf/cav/GrafS97}
Susanne Graf and Hassen Sa{\"{\i}}di.
\newblock Construction of abstract state graphs with {PVS}.
\newblock In {\em Proceedings of CAV 1997}, volume 1254 of {\em LNCS}, pages
  72--83. Springer, 1997.

\bibitem{DBLP:journals/iandc/GrumbergLLS07}
Orna Grumberg, Martin Lange, Martin Leucker, and Sharon Shoham.
\newblock When not losing is better than winning: Abstraction and refinement
  for the full mu-calculus.
\newblock {\em Information and Computation}, 205(8):1130--1148, 2007.

\bibitem{DBLP:conf/cav/GuptaS05}
Anubhav Gupta and Ofer Strichman.
\newblock Abstraction refinement for bounded model checking.
\newblock In {\em Proceedings of CAV 2005}, volume 3576 of {\em LNCS}, pages
  112--124. Springer, 2005.

\bibitem{Book:DynLog}
David Harel, Jerzy Tiuryn, and Dexter Kozen.
\newblock {\em Dynamic Logic}.
\newblock MIT Press, Cambridge, MA, USA, 2000.

\bibitem{DBLP:journals/csur/JhalaM09}
Ranjit Jhala and Rupak Majumdar.
\newblock Software model checking.
\newblock {\em {ACM} Computing Surveys}, 41(4):21:1--21:54, 2009.

\bibitem{DBLP:reference/mc/JhalaPR18}
Ranjit Jhala, Andreas Podelski, and Andrey Rybalchenko.
\newblock Predicate abstraction for program verification.
\newblock In Edmund~M. Clarke, Thomas~A. Henzinger, Helmut Veith, and Roderick
  Bloem, editors, {\em Handbook of Model Checking}, pages 447--491. Springer,
  2018.

\bibitem{DBLP:journals/toplas/Kozen97}
Dexter Kozen.
\newblock Kleene algebra with tests.
\newblock {\em ACM Transactions on Programming, Languages and Systems},
  19(3):427--443, 1997.

\bibitem{DBLP:conf/lics/LarsenT88}
Kim~Guldstrand Larsen and Bent Thomsen.
\newblock A modal process logic.
\newblock In {\em Proceedings of LICS 1988}, pages 203--210. {IEEE} Computer
  Society, 1988.

\bibitem{DBLP:conf/csl/Leiss91}
Hans Lei{\ss}.
\newblock Towards {K}leene algebra with recursion.
\newblock In {\em Proceedings of CSL 1991}, volume 626 of {\em LNCS}, pages
  242--256. Springer, 1991.

\bibitem{DBLP:journals/fmsd/LoiseauxGSBB95}
Claire Loiseaux, Susanne Graf, Joseph Sifakis, Ahmed Bouajjani, and Saddek
  Bensalem.
\newblock Property preserving abstractions for the verification of concurrent
  systems.
\newblock {\em Formal Methods in System Design}, 6(1):11--44, 1995.

\bibitem{DBLP:conf/sas/Masse02}
Damien Mass{\'{e}}.
\newblock Semantics for abstract interpretation-based static analyzes of
  temporal properties.
\newblock In {\em Proceedings of SAS 2002}, volume 2477 of {\em LNCS}, pages
  428--443. Springer, 2002.

\bibitem{mine17}
Antoine Min{\'{e}}.
\newblock Tutorial on static inference of numeric invariants by abstract
  interpretation.
\newblock {\em Foundations and Trends in Programming Languages},
  4(3-4):120--372, 2017.

\bibitem{NielsonN10}
Flemming Nielson and Hanne~Riis Nielson.
\newblock Model checking \emph{Is} static analysis of modal logic.
\newblock In {\em Proceedings of FoSSaCS 2010}, volume 6014 of {\em LNCS},
  pages 191--205. Springer, 2010.

\bibitem{NielsonSN02}
Flemming Nielson, Helmut Seidl, and Hanne~Riis Nielson.
\newblock A succinct solver for {ALFP}.
\newblock {\em Nordic Journal of Computing}, 9(4):335--372, 2002.

\bibitem{NielsonN02}
Hanne~Riis Nielson and Flemming Nielson.
\newblock Flow logic: {A} multi-paradigmatic approach to static analysis.
\newblock In {\em The Essence of Computation, Complexity, Analysis,
  Transformation. Essays Dedicated to Neil D. Jones [on occasion of his 60th
  birthday]}, volume 2566 of {\em LNCS}, pages 223--244. Springer, 2002.

\bibitem{DBLP:journals/pacmpl/OHearn20}
Peter~W. O'Hearn.
\newblock Incorrectness logic.
\newblock In {\em Proceedings of POPL 2020}, volume~4, pages 10:1--10:32.
  {ACM}, 2020.

\bibitem{DBLP:conf/esop/Ranzato01}
Francesco Ranzato.
\newblock On the completeness of model checking.
\newblock In {\em Proceedings of ESOP 2001}, volume 2028 of {\em LNCS}, pages
  137--154. Springer, 2001.

\bibitem{DBLP:conf/sas/RanzatoT02}
Francesco Ranzato and Francesco Tapparo.
\newblock Making abstract model checking strongly preserving.
\newblock In {\em Proceedings of SAS 2002}, volume 2477 of {\em LNCS}, pages
  411--427. Springer, 2002.

\bibitem{DBLP:journals/logcom/RanzatoT07}
Francesco Ranzato and Francesco Tapparo.
\newblock Generalized strong preservation by abstract interpretation.
\newblock {\em Journal of Logic and Computation}, 17(1):157--197, 2007.

\bibitem{schmidt}
David Schmidt.
\newblock Binary relations for abstraction and refinement.
\newblock Technical report, Kansas State University, 2001.

\bibitem{Schmidt98}
David~A. Schmidt.
\newblock Data flow analysis is model checking of abstract interpretations.
\newblock In {\em Proceedings of POPL 1998}, pages 38--48. {ACM}, 1998.

\bibitem{DBLP:conf/sas/SchmidtS98}
David~A. Schmidt and Bernhard Steffen.
\newblock Program analysis \emph{as} model checking of abstract
  interpretations.
\newblock In {\em Proceedings of SAS 1998}, volume 1503 of {\em LNCS}, pages
  351--380. Springer, 1998.

\bibitem{SchmidtS98}
David~A. Schmidt and Bernhard Steffen.
\newblock Program analysis \emph{as} model checking of abstract
  interpretations.
\newblock In {\em Proceedings of SAS 1998}, volume 1503 of {\em LNCS}, pages
  351--380. Springer, 1998.

\bibitem{DBLP:conf/tacas/ShohamG04}
Sharon Shoham and Orna Grumberg.
\newblock Monotonic abstraction-refinement for {CTL}.
\newblock In {\em Proceedings of TACAS 2004}, volume 2988 of {\em LNCS}, pages
  546--560. Springer, 2004.

\bibitem{Steffen91}
Bernhard Steffen.
\newblock Data flow analysis as model checking.
\newblock In {\em Proceedings of TACS 1991}, volume 526 of {\em LNCS}, pages
  346--365. Springer, 1991.

\bibitem{Steffen93}
Bernhard Steffen.
\newblock Generating data flow analysis algorithms from modal specifications.
\newblock {\em Sci. Comput. Program.}, 21(2):115--139, 1993.

\bibitem{DBLP:conf/tapsoft/StirlingW89}
Colin Stirling and David Walker.
\newblock Local model checking in the modal mu-calculus.
\newblock In {\em Proceedings of TAPSOFT'89}, volume 351 of {\em LNCS}, pages
  369--383. Springer, 1989.

\bibitem{PW83}
Pierre Wolper.
\newblock Temporal logic can be more expressive.
\newblock {\em Information and Control}, 56(1/2):72--99, 1983.

\bibitem{ZhangNN12}
Fuyuan Zhang, Flemming Nielson, and Hanne~Riis Nielson.
\newblock Model checking as static analysis: Revisited.
\newblock In {\em Proceedings of IFM 2012}, volume 7321 of {\em LNCS}, pages
  99--112. Springer, 2012.

\end{thebibliography}
  
  \end{document}